 \documentclass[final,pdftex]{Preamble/ectaart}

\usepackage{Preamble/PreambleMain}
\usepackage{Macros/macroMain}

\usepackage[none]{hyphenat}

\begin{document}
\begin{frontmatter}

\title{Simple subvector inference on  sharp identified set in affine models \protect\thanksref{T1}}
\runtitle{Simple subvector inference on sharp identified set}
\thankstext{T1}{This version is July 12, 2024. Earlier versions of the paper was  circulated with the title ``Inference on Scalar Parameters  in Set-Identified Affine Models'' and   ``Inference in high-dimensional Set-Identified Affine Models''. The work is based on  a chapter in my PhD dissertation \citep{gafarov2017essays}.  First-draft date: November 10, 2015.}

\begin{aug}
\author{\fnms{Bulat} \snm{Gafarov}\thanksref{t1}\ead[label=e1]{bgafarov@ucdavis.edu}}

\thankstext{t1}{I am extremely grateful to Joris Pinkse and Patrik Guggenberger for their very helpful and detailed comments. I would like to thank (in alphabetical order) Donald Andrews, Andres Aradillas-Lopez, Christian Bontemps, Ivan Canay, Peng Ding, Graham Elliott, Zheng Fang, Dalia Ghanem, Joachim Fryberger, Ronald Gallant, Michael Gechter, Marc Henry, Keisuke Hirano, Sung Jae Jun, Nail Kashaev, Francesca Molinari, Demian Pouzo, Adam Rosen, Thomas Russell, Andres Santos, Xiaoxia Shi, Jing Tao, and Alexander Torgovitsky
for their comments and suggestions.}

\runauthor{B. Gafarov}

\address{University of California, Davis, Department of Agricultural and Resource Economics \\
}

\end{aug}

\begin{abstract}

This paper studies a regularized support function estimator for bounds on components of the parameter vector in the case in which the identified set is a polygon. 
The proposed regularized estimator has three important properties: (i) it has a uniform asymptotic Gaussian limit in the presence of flat faces in the absence of redundant (or overidentifying) constraints (or vice versa); (ii) the bias from regularization does not enter the first-order  limiting distribution; (iii) the estimator remains consistent for sharp (non-enlarged) identified set for the individual components even in the non-regualar case.
These properties are used to construct \emph{uniformly valid} confidence sets for an element $\theta_{1}$ of a parameter vector $\theta\in\mathbb{R}^{d}$ that is partially identified by affine moment equality and inequality conditions.
The proposed confidence sets can be computed as a solution to a small number of linear and convex quadratic programs, leading to a substantial decrease in computation time and guarantees a global optimum.
As a result, the method provides a uniformly valid inference in applications in which the dimension of the parameter space, $d$, and the number of inequalities, $k$,  were previously computationally unfeasible ($d,k=100$).
The proposed approach can be extended to construct confidence sets for intersection bounds,  to construct joint polygon-shaped  confidence sets for multiple components of $\theta$, and to find the set of solutions to  a linear program.
Inference for coefficients in the linear IV regression model with an interval outcome is used as an illustrative example.

\end{abstract}


\end{frontmatter}

\newpage
\setcounter{page}{1} 
\section{Introduction}\label{sec:Introduction}

Strong econometric assumptions can lead to poor estimates. 
Sometimes, moment inequalities can provide alternative estimates under weaker assumptions.
Linear models with interval-valued outcome data are a good example.\footnote{Other examples of affine-moment inequalities include monotone instrumental variables (\citet{manski2000monotone}, \citet{Freyberger201541}) and models with missing data (\citet{manski2003partial}).}
It is common practice to replace income-bracket data with the corresponding midpoints when estimating the returns to schooling (\cite{trostel2002estimates}).
However, the conventional approach is applicable only under strong assumptions about the distribution of the residual term.\footnote{Another common approach is to assume Gaussian distribution for the residuals and apply the maximum likelihood method (\citet{stewart1983least}). } 
The affine moment inequality approach to interval-valued  data proposed by \cite{manski2002inference} can set-identify the return to schooling without such strong assumptions. 

Multiple methods can be used to construct confidence sets (CS)  for parameters defined by moment inequalities.
The pioneering procedures of \cite{chernozhukov2007estimation} and \cite{andrews2010inference} (AS) and their subsequent refinements by 
\cite{bugni2014inference} (BCS) and \cite{kaido2015inference} (KMS) are powerful statistical methods that solve this inference problem in the small-dimensional case.
Some applications, such as panel or semiparametric regression models with interval-measured outcome variables, have a high-dimension  parameter space,  which poses a computational challenge for the existing procedures.\footnote{For example, \cite{trostel2002estimates} consider a panel regression with more than 60 variables that include country fixed effects, time effects,  exogenous demographic control variables, and their interactions.}

I propose a novel regularized support  function estimator for the lower and upper extremes of the identified set   for an element $\theta_{1}$ of an unknown parameter vector $\typVector{\theta}\in\R^{d}$ in models defined by affine moment equalities and inequalities.
In the example of returns to schooling, $\theta_{1}$ corresponds to the returns to schooling and $\typVector{\theta}\in\R^{d}$ to the full vector of the regression coefficients.
The novel estimator has a closed-form asymptotic Gaussian  distribution, which I use to construct uniformly valid confidence bounds and confidence intervals for $\theta_{1}$. 
The proposed set has valid asymptotic coverage probability uniformly over a class of data-generating processes (DGP). 
Uniformity in DGP is a desirable property, as it results in better coverage-probability control  in small samples compared to point-wise analogs in \emph{nonregular statistical models} such as the affine moment inequality model.

The regularized support function proposed in this paper is a solution to a convex quadratic program that minimizes the sum of $\theta_{1}$ and a penalty $\mu_{n}\left\Vert \typVector{\theta}\right\Vert ^{2}$ with $\mu_{n}\to0$, subject to the sample moment restrictions. 
If the set of optima for $\mu_n=0$ is not a singleton, this additional convex term selects the optimum with the minimal norm as $n$ increases. 
The standard errors are computed using the sample variance of the weighted moment conditions at the unique optima. 
To correct the asymptotic bias resulting from the regularization exactly, I suggest using the argmin of the regularized program with a larger tuning parameter $\kappa_{n}\to0$. 
If $\kappa_{n}/\mu_{n}\to\infty$ as $n\to\infty,$ then the bias correction does not affect the asymptotic distribution of the estimator. 
To achieve a uniformly valid confidence interval (CI), I replace the exact correction with an upper bound on the maximum of $\mu_{n}\left\Vert \typVector{\theta}\right\Vert ^{2}$ over the argmin set of the nonregularized program. 

The proposed CIs have several attractive statistical and computational properties that make them viable in high-dimensional affine moment inequality models. 

First, the estimator of the regularized support function has an asymptotically linear (Bahadur-Kiefer) representation that provides an easy-to-compute asymptotic standard error.
Consequently, this paper is the first to propose a closed-form estimator of the bounds on $\theta_{1}$ convex moment inequality models  with asymptotic Gaussian distribution in the non-regular case (in the absence of strict convexity). 
In contrast, the estimator of the ordinary support function estimator used in the existing literature (\citet{beresteanu2008asymptotic}, \citet{kaido2014asymptotically}, \citeauthor{Freyberger201541} (2015, FH), \citet{gafarov2018delta}, among others) can have a non-Gaussian asymptotic distribution, which complicates uniform inference. 
To establish the uniform remainder bound in the Bahadur-Kiefer expansion, I developed a novel second-order directional envelope theorem, which is a theoretical result of independent interest.

Second, the proposed approach requires only a fraction of the computational time of the existing uniform procedures if  $\typVector{\theta}$ has a large number of dimensions.
The computational cost is low since it involves only four quadratic programs, it does not require any resampling, and it depends on covariance of the moment conditions at two points. 
 In the asymptotic analysis, I   only consider the case of a fixed dimension of $\theta$ and a fixed number of inequalities for analytical simplicity. 
 The goal of the present study is to focus on the computational difficulties resulting from the high-dimensional moment inequalities.
(There are recent papers that are concerned with the impact of the growing number of inequalities on the statistical properties of the inference procedures; see, for example, \cite{belloni2018subvector}.)

The computation time for my procedure increases slowly in the dimension of $\theta\in\R^{d}$ and takes only 0.1 second  for $d=20$ and $k=40$ moment inequalities and 2.5 seconds for $d=100$ and $k=200$.   As a result, the proposed method  can address  parameter $\theta$ with a large dimension and a large number of moment conditions. 
In contrast, the existing uniform-inference methods for moment inequalities proposed by AS, KMS, and BCS are based on costly nonconvex optimization.
\label{reply:R1.1} In fact, despite the fact that the moment inequalities are convex (linear) the standardized moment conditions are not convex which can potentially result in creation of multiple local optima which complicate computation for these statistical procedures (see Appendix Sections \ref{sec:Convergnce rate} and \ref{sec:globalOptimum}).

I provide an example of an affine moment inequality model by showing that the number of local optimal solutions in existing uniform procedures (AS, BCS, and KMS) can grow exponentially with dimension $d$. 
As a result, the procedures take more computational time and can produce misleadingly short CIs if the optimization routine fails to find the global optimum.
It takes 100 seconds to compute the CI of AS in an affine model with $d=20$ and $k=40$ moment inequalities which is 1000 times longer than the newly proposed method.

The simulation evidence suggests that {the computational speed gains} come without a substantial loss of statistical power.
In the non-regular cases, the proposed uniform CIs have length properties that are not worse than those of the existing uniform methods. (The proposed \emph{uniform} CI has a length within simulation error from the projection CI of AS in the Monte Carlo (MC) design considered in this paper.) In the regular cases, the novel confidence bounds attain the efficiency bound of the usual support function estimators (as shown in \cite{kaido2014asymptotically}).

\label{reply:R2.1.1}
The proposed idea of regularizing support functions for linear moment inequality inference gained a subsequent development in a recent work \cite{cho2023simple}.\footnote{The first draft of \cite{cho2023simple} was circulated on 27 May 2019 on Arxiv depository two years after the initial distribution of the present paper as a PhD dissertation chapter in \cite{gafarov2017essays}.}
The authors argue that if one is satisfied with an inference on an \emph{enlarged} identified set, one can simply regularize the moment inequalities by adding random noise to their coefficients under milder regularity conditions.  
This operation restores a Gaussian limit of the perturbed estimated support function, allowing conventional bootstrap inference on the \emph{enlarged} identified set. 
Since the inference is done on the enlarged identified set, one does not need to impose any constraint qualification conditions on the moment inequalities --- they are satisfied automatically with probability 1 after adding the random noise to the coefficients.
Unfortunately, this approach results in confidence sets with zero power against all local alternatives corresponding to other nested enlargements of the identified set. 
In practical terms, it means that the confidence sets can be very conservative if too much noise was added or fail to control size in small samples if the added noise was insufficient (there is no theory that would determine the minimal level of required noise for a given sample size).    
In contrast, the theory provided in the present paper explicitly studies the impact of the choice of tuning parameters on the size of the regularized identified set. 
Such an analysis requires imposing constraint qualification conditions on the moment inequalities (Assumption  \ref{ass:ULICQ}).
Under these conditions, Theorem \ref{thm:bounds} shows that for sufficiently small value of the tuning parameters $\mu_n$ and $\kappa_n$  both lower and upper bounds based on the regularized value function coincide exactly with the nonregularized support function in the regular case when the set of primal solutions is a singleton. 
More generally, under the maintained assumptions, the regularized estimators remain consistent for the bounds on the sharp (non-enlarged) identified set, thus providing non-trivial local power against the relevant alternatives.

Identified sets defined by affine inequalities appear in various economic applications, in particular those dealing with discrete variables and shape restrictions. Linear models with interval outcome, which were originally studied in \cite{manski2002inference} and \cite{haile2003inference},  are just one example of affine inequalities. Other examples include bounds on marginal effects in dynamic discrete choice panel models (\cite{honore2006bounds}, \cite{torgovitsky2016nonparametric, torgovitsky2018partial}), bounds on average treatment effects (\cite{kasy2016partial}, \cite{Laffer2018}, \cite{russell2017sharp}), nonparametric IV models with shape restrictions (\cite{manski2000monotone}, \cite{Freyberger201541}), errors in variables (\cite{MOLINARI200881}),  intersection bounds (\cite{honore2006bounds2}), revealed  preference restrictions (\cite{kline2016bounding}), and game-theoretic models (\cite{syrgkanis2017inference}).
This paper focuses on the case with finitely many affine unconditional inequalities that are non-overidentifed and that result in a non-empty identified set. 
This, of course, rules out many interesting economic applications that involve overidentifying moment inequalities \citep[for example,][]{shi2018estimating}, or where moment inequalities are nonlinear \citep[for example,][]{pakes2015moment}.

I also contribute to the growing literature on inference on non-differentiable functions and regularized estimators. 
Bounds on components of a parameter characterized by linear moment conditions considered in this paper are an example of a nondifferentiable (\emph{nonregular}) function of a parameter (the expectation of the data) that has an asymptotically normal estimator (the sample mean).  
Distributions of nondifferentiable functions of the normal estimator are hard to approximate using standard methods. 
(See Section~\ref{subsec:asymptoticDistribution} for a detailed discussion.)
I propose differentiable lower and upper bounds that converge to the nondifferentiable function of interest as the sample size grows.
Since the bounds are regular parameters themselves, the standard delta method and bootstrap can be used to conduct one-sided or two-sided inference on the bounds. 
In the regular case,  these bounds collapse and coincide with the original parameter of interest, which results in a $\sqrt{n}$-consistent and asymptotically normal estimator.
In the non-regular case, the bounds converge at a slower rate and result in a locally biased estimator, which is acceptable for valid one-sided inference. 
 
Another interesting statistical problem that appears in many applications is inference for extrema of finitely many means of random variables. 
It is known as \emph{ intersection-bounds problem } (\cite{hall2010bootstrap} and \cite{chernozhukov2013intersection}) and can be framed as a value of a linear program.
The regularized support function estimator can also be used for uniform delta-method CSs in this setting (see Appendix Section~\ref{subsec:overidentification}).
The approach considered here is expected to have statistical properties similar to \cite{chernozhukov2013intersection}, but has the advantage of closed-form standard errors and critical values, which correspond to the standard normal distribution.

The paper is structured as follows. Section~\ref{sec:Setup} describes the setup, gives examples, and summarizes the available literature.
Section~\ref{sec:results} provides the main results, applies them to uniform inference on projections, and discusses extensions (overidentified inequality models, joint CSs, characterization of argmin sets). 
Section~\ref{sec:Monte-Carlo} provides the results of the Monte Carlo experiments. 
Section~\ref{sec:Conclusion} concludes.

The following notational conventions are used.
$\bydef$  denotes definitions. 
 $\expect\left[\cdot\right]$  denotes expectation with respect to a probability distribution $\measTrue$. 
Uppercase English letters  denote random variables (scalar, vector, or matrix valued), and lowercase letters  denote the corresponding realizations, for example, $\Data_i$ and $\data_i$. 
 $\measEmp$ denotes the empirical distribution and $f\left(0^+\right)$ denotes $\lim_{x\downarrow0}f\left(x\right)$.
The vector $\ej\bydef(0,...1,...0)'$ is the $j$-th coordinate vector, where 1 occurs at position $j$. $\ej$ is the projector on the $j$-th coordinate. 
The symbol $\jSet$ denotes a finite set of indices $\jSet\bydef\{i_1,...,i_\ell\}\subset\numbers$ and $\jMat\bydef(e_{i_1},...,e_{i_\ell})^\prime$ as a coordinate projection matrix in the corresponding Euclidean space. 
 $\left|\jSet\right|$  denotes the cardinality of the set $\jSet$.  u.h.c. stands for upper-hemicontinuous correspondence.
% The symbol $\text{sVar}\left(x\right)$  denotes  the sample variance, where $\text{sVar}\left(x\right) = \sMean{x_i^{2}}-\left(\sMean x_i\right)^{2}.$ 

\section{Setup, motivating examples, and related literature}\label{sec:Setup}

\subsection{\label{subsec:GeneralSetup} Support function and projections of identified sets}

\label{R3.5.c}
Consider a support function  for a polygon    $\setID\subset\R^{d}$ (a set defined by a system of linear equalities and inequalities) that depends on a data generating process  parametrized by a measure $\measTrue$ evaluated at a direction $e_1\bydef(1,0,\dots)^\prime 
\in\R^d$,\footnote{See, for example, \citet[][ chapter 13]{rockafellar1970convex}}
\begin{equation}
\begin{array}{ccc}
\funLValO\bydef\displaystyle\min_{ {\theta}\in\setID}\eOne {\theta}. 
\end{array} \label{prog:minP}
\end{equation}
The set $\setID$,  also referred to as an \emph{identified set} for a parameter vector $ {\theta}\in\R^{d}$, consists of all solutions to  the following system of affine moment equalities/inequalities: 
\begin{equation}
\begin{cases}
\expect\moment{W}{\theta}  =0, & j\in\setEq,\\
\expect\moment{W}{\theta}  \leq0, & j\in\setIneq.
\end{cases}\label{eq:GenericMomentConditions}
\end{equation} 
Here,  $\moment{W}{\theta} \bydef  \sum_{\ell=1}^d \Data_{j\ell} \theta_\ell -\Data_{j(d+1)}$.
I consider a setup with finitely many unconditional moment functions (that is, $\left|\setEq\cup\setIneq\right|=k$,  $\card{\setEq}=p$, $0\leq p \leq d$, $k<\infty$).
The random matrix corresponding to an individual observation $\Data$ has probability measure $\measTrue$ with  sample space $\R^{k\times\left(d+1\right)}$,
\begin{align*}
  W &=  \left(\begin{array}{c}
 W'_1\\
 \vdots \\
  W'_k
\end{array} \right) ,\quad W_j = (W_{j1},\dots,W_{j(d+1))})', \quad j=1,\dots,k.
\end{align*}
The econometrician observes an i.i.d. sample  $\left\{ \data_{i}\in\R^{k\times\left(d+1\right)}|\,i=1,...,n\right\} $ of random matrix $W$.  
 There is a straightforward way to extend the analysis to the case of dependent data as long as a CLT for averages of $\data_{i}$ remains valid.

This setup reduces to a finite-dimensional parametric statistical model once    the support function evaluated at $e_1$ is represented as a linear program, 
\begin{align}
\funLValO&=\min_{   {\theta}\in\R^{d} }  \eOne {\theta} \label{eq:primalLP} \\
\text{s.t. } &   e_j^\prime  \Ap \theta= e_j^\prime\bp , \quad j\in\setEq,\nonumber\\ 
&   e_j^\prime  \Ap \theta\leq e_j^\prime\bp , \quad j\in\setIneq.\nonumber 
\end{align}
Here, the coefficients on the left-hand side, $\Ap$, and the right-hand side, $\bp$, taken together constitute  matrix $\expect \Data \bydef (\Ap|\bp)$. It  means $W$’s expectation is a $k\times (d + 1)$ matrix whose first $d$
columns is $A_P$ and the last column is the $k$-dimensional vector  $b_p$.
Following optimization theory terminology, I occasionally refer to the support function at $e_1$ as \emph{value} of program \eqref{eq:primalLP} and to the corresponding argmin set as the set of \emph{optimal solutions}. 

The focus on the first coordinate of $\theta$ as an objective of \eqref{prog:minP} is without loss of generality.
As discussed in Section~\ref{subsec:subvectors}, the support function evaluated at any unit direction vector $a\in\R^d$ can be represented in the form \eqref{prog:minP} after some redefinition of the parameter space (matrices $ \Ap$ and $\bp $ can be functions of the unit vector $a$).

To ensure boundedness of the support function, I  assume that system (\ref{eq:GenericMomentConditions}) includes inequalities that make the identified set compact,
\begin{equation}
-\infty<-\underline{c}_{\ell}\leq\theta_{\ell}\leq \bar{c}_{\ell}<\infty \quad \text{for } \ell=1,...,d,  \label{eq:Box}
\end{equation}
for some constants $\underline{c},\bar{c}\in\R^d_+$. 
Moreover, I impose the following assumption:
\begin{assumption}
\label{assu:Non--empty} $\Theta\left(\measTrue\right)$ is non-empty
for the probability measure $\measTrue$. 
\end{assumption}
This assumption implies that $\funLValO<+\infty$. 
It is valid whenever the model corresponding to \eqref{eq:GenericMomentConditions} is correctly specified.

Under Assumption \ref{assu:Non--empty}, support functions evaluated at  $\pm e_j$   characterize  projections of $\setID$ on individual coordinates $j$ of $\theta$.
In particular, the marginal (projected) identified set for the coordinate $\theta_{1}$ can be represented as an interval  $\setMarginal=\left[\funLValO,\funUValO\right]$, 
where the bounds are support functions evaluated at   directions $\eOneVec$ and $-\eOneVec$.
Indeed, the upper bound can also be written as a (minus) support function at $e_1$, 
\begin{equation}
\funUValO\bydef\max_{ {\theta}\in\setID}  \{ \eOne {\theta}\}=-\min_{ {\theta}\in\setID}  \{-\eOne {\theta}\}.
\end{equation}
The analysis for the upper bound is analogous to the one for the lower bound, so from here on I focus on the lower bound.

\subsection{\label{subsec:examples} Motivating example: Linear IV model with interval-valued outcome }

The polygon-shaped identified sets considered in this paper appear in many econometric models that feature discrete data, as mentioned in the introduction. 
Many of these applications are concerned with instrumental variables. 
 I use the linear IV model with interval outcome \citep[for example,][]{manski2002inference,chernozhukov2007estimation} to illustrate ideas throughout the paper. Other applications of the proposed method include monotone IV \citep{manski2000monotone}  and nonparametric IV \citep{Freyberger201541}.

Consider a linear model for a random vector $(Y,X,Z)$ satisfying  
$$\expect\left[Y- {\theta}^{\prime} X|  Z\right] =0,$$
where $\theta\in\R^d$ is the vector of regression coefficients.
The true outcome $Y$ is unobserved;  only its a.s. bounds  $\left[\underline{Y},\overline{Y}\right]$ are observed.
Suppose that the vector of instrumental variables $Z$ has finite support $\left\{   z_{1},\dots,  z_{K}\right\} \subset\R^{d}$.
% \footnote{If $S_{Z}$ is infinite, one can estimate an enlargement of $\setMarginal$
% using a finite number of unconditional moment inequalities. See
% \citet{chernozhukov2007estimation} for details. \citet{andrews2013inference}
% provide conditions for sharp characterization of the identified set
% by a finite number of unconditional moment functions. I leave the case of infinite number of moment inequalities for future extensions.}
In this case, the model implies a polygon-shaped identified set $\setID$ for $\theta$ defined by a set of moment inequalities,
\begin{equation}
\begin{cases}
\expect\left[\underline{Y}\Ch{  Z=  z_{j}} \right]\leq {\theta}^{\prime}  \expect \left[ X\Ch{  Z=  z_{j}}\right] , & j=1,...,K,\\
\expect\left[\overline{Y}\Ch{  Z=  z_{j-K}} \right]\geq {\theta}^{\prime} \expect \left[X\Ch{  Z=  z_{j-K}} \right], & j=K+1,...,2K.
\end{cases}\label{eq:intervalOutcome}
\end{equation}

These inequalities can be represented in the standard form (\ref{eq:GenericMomentConditions})
with $p=0$, $k=2K$, and the following observation matrix:

\begin{align*}
\Data_{j\ell}\bydef\begin{cases}
- X_{\ell}\Ch{  Z=  z_{j}} , & \text{ for }j=1,\dots,K, \ell =1,\dots,d,\\
  X_{\ell}\Ch{  Z=  z_{j-K} } , & \text{ for }j=K+1,\dots,2K, \ell =1,\dots,d,
\end{cases}\\
\Data_{j(d+1)}\bydef\begin{cases}
\underline{Y}\Ch{  Z =  z_{j} } , & \text{ for }j=1,...,K,\\
-\overline{Y}\Ch{   Z =  z_{j-K} } , & \text{ for }j=K+1,...,2K.
\end{cases}
\end{align*}

 If it is known a priori that for some support points $j$  $$\expect\left[\underline{Y}\Ch{  Z=  z_{j}} \right] = \expect\left[\overline{Y}\Ch{  Z=  z_{j-K}} \right],$$
then one can replace the corresponding pair of inequalities with a single equality,
\begin{equation}
\expect\left[\frac{1}{2}(\underline{Y}+\overline{Y})\Ch{  Z=  z_{j}} \right]= {\theta}^{\prime}  \expect X\Ch{  Z=  z_{j}} .
\end{equation}
In this case, $p$ (the number of equality restrictions) is equal to the number of such support points.
One can further incorporate additional shape restrictions information such as signs of components of $ {\theta}$ in the form of linear inequalities to narrow the identified set.  

This example appears in the context of the estimation of return to schooling using survey data. 
 \citet{trostel2002estimates} study economic returns to schooling
for 28 countries using data from the International Social Survey Programme  from 1985 to  1995. They estimate the conventional \citet{mincer1974schooling}
model of earnings (the human capital earnings function), which has $Y$, the log of hourly wages,  satisfying  
\begin{equation}
\expect\left[Y - {\theta}^{\prime} X|  Z\right]=0 ,\label{eq:schooling}
\end{equation}
where the first regressor, $  X_{1}$, is the years of schooling; the other components of $X$ play a role of additional controls.
The component $\theta_{1}$ is then interpreted as the return to schooling. 
It is equal to the percentage change in wages due to an additional year of schooling.
To correct for the endogeneity bias in $\theta_{1}$ resulting from omitting the latent ability variable, one can use an instrument vector $Z$  that correlates with $X_{1}$ (for example, an indicator of whether a good school is in  proximity or an indicator of  the quarter of birth).
 
Exact measurements of $Y$ are not available for some countries (including the US); only hourly-income- bracket data $\left[\underline{Y},\overline{Y}\right]$ are
available.

Instrumental variables $Z$ (coinciding with the control variable)  considered in \citet{trostel2002estimates} take discrete values. 
The variables include annual fixed effects, union status, marital status, age and age squared, and country-year dummies (in the case of the aggregate equation).
With inclusion of the country and time effects,  $d$ can be larger than 60.
Because of the large number of support points for $Z$, the corresponding system~\eqref{eq:GenericMomentConditions} would also have a large number of linear moment inequalities $k$.

The conventional technique to estimate IV regression with interval-outcome data is to replace the interval data with the corresponding midpoints and estimate the model using the OLS method.
This technique is valid only under the unreasonably strong condition
\begin{align}
\expect\left[\left(Y-\frac{1}{2}(\underline{Y}+\overline{Y})\right)  Z\right] & =0.\label{eq:midpoint}
\end{align}
If Equation \eqref{eq:midpoint} is violated, then the OLS estimator with midpoints is inconsistent for the true parameter $\theta$.\footnote{The OLS is, however, consistent for the best linear predictor of the midpoint that may not have the desirable economic interpretation; see, for example,  \citet{shi2020uniform}.}
Without assuming that \eqref{eq:midpoint} is satisfied, support function estimators (defined below) provide consistent bounds on marginal identified sets of the true parameter $\theta$ \citep[see][]{beresteanu2008asymptotic,bontemps2012set}.

When constructing valid CSs  on projections of identified sets, strictly speaking, one needs only a lower bound on the support function at $e_1$ in  \eqref{prog:minP}.
Other applications may instead require estimators of an upper bound on the (minimum) value of an optimization problem. 
The upper bound can be used, for example, to bound the set of optimal solutions or to construct a consistent \emph{inner} estimator of a convex identified set \citep[the inner set estimator has important applications in inference; see, for example, ][]{bugni2017inference}.

\subsection{\label{subsec:asymptoticDistribution} Review of existing results on the asymptotic distribution of support function estimators for polygon-shaped identified sets }
 
Following \cite{beresteanu2008asymptotic} the parameter $\funLValO$ (the support function at $e_1$)  can be estimated using a sample analog,
\begin{align}
\hat{\underline{v}}_n &=\min_{ {\theta}\in\R^{d} } \eOne {\theta}  \\
\text{s.t. } & \begin{cases}
\sMean\moment{W_i}{\theta} =0, & j\in\setEq,\\
\sMean\moment{W_i}{\theta} \leq0, & j\in\setIneq,
\end{cases} 
\end{align} 
where the observations $W_i$  are independent copies of the random matrix $W$.
The asymptotic distribution of this estimator has been extensively studied in the case of the strictly convex identified set.
The strict convexity implies that  the  support function  at $e_1$, $\hat{\underline{v}}_n$ , is differentiable in the sample mean $\sMean W_i$ (under additional regularity conditions discussed below).  
As a consequence, its sample analog has an asymptotic Gaussian distribution, admits of bootstrap inference, and attains semiparametric efficiency \citep{kaido2014asymptotically}.
\label{reply:R1.3.1}In contrast to the strictly convex case, the Gaussian limit is not guaranteed anymore in the general (non-deterministic) linear moment inequality model (\cite {kaido2014asymptotically} only allow for deterministic linear inequalities under particular regularity conditions).

The stochastic programming approach \citep{shapiro1991asymptotic} enables an asymptotic analysis of the support function estimators for a fixed direction in the general convex case, which includes the linear moment inequality model. 
In this section, I briefly review the two main ideas in this approach, Lagrangian duality of convex programs and the delta method for directionally differentiable functions.  
I conclude with an overview of the alternative inference approaches and their relation to stochastic programming methods.

The Lagrangian duality theory provides additional, often more convenient, formulations of convex programs.
Namely, the (primal) program in \eqref{prog:minP} has the same value as its Lagrangian dual formulation, 
\begin{align}
\funLValO=\max_{ {\lambda}\in\R^{p}\times\R_{+}^{k-p}} & \left\{ - {\lambda}^{\prime}\bp\right\} \label{eq:DualProgram0},\\
\text{s.t. } &  {\lambda}^{\prime}\Ap=-\eOne,\nonumber 
\end{align}
given Assumption~\ref{assu:Non--empty}.
Iff the dual program has a bounded set of solutions $\argminL(\measTrue)$, the support function for a given direction has bounded directional derivatives in $\Ap,\bp$ (defined precisely below).
Proposition 5.45 in \cite{bonnans2013perturbation} provides a necessary and sufficient condition for $\argminL(\measTrue)$ to be bounded:
\begin{condition}[Slater's]
\label{con:slater} There exist $ \theta \in \setID$  s.t. $\expect\moment{W}{\theta}  < 0$ for all $ j\in\setIneq$.
\end{condition}

Condition~\ref{con:slater}
enables another, Lagrangian  min/max, representation of  \eqref{prog:minP}, which is particularly convenient for computing derivatives of the support function in a given direction  for the delta-method-based inference procedures.
Suppose that $\Lambda\subset{\R}^{p}\times\R_{+}^{k-p}$ is some compact set that contains $\argminL(\measTrue)$. 
The support function for a given direction $ e_1$ can then be represented as \citep[see, for example,][p. 437]{bonnans2013perturbation} 
\begin{equation}
\funLValO=\min_{ {\theta}\in \Theta }\max_{ {\lambda}\in\Lambda} \{\eOne {\theta} +{\lambda}^{\prime}(\Ap\theta-\bp)\}.  \label{prog:minmax}
\end{equation}
The directional derivative for the value of this min/max program  is  provided by a corresponding envelope theorem    \citep[for example,][Theorem 7.28]{shapiro2014lectures}.
Namely, the envelope theorem gives a derivative of any perturbed version of the min/max program,
\begin{equation}
\funLValO[P,t]\bydef\min_{ {\theta}\in\Theta}\max_{ {\lambda}\in\Lambda} \{\eOne {\theta}  +{\lambda}^{\prime}((\Ap +t h_{A,t})\theta-(\bp+t h_{b,t}))\},  \label{prog:minmaxPerturbed}
\end{equation}
with respect to a scalar $t$  for any uniformly converging sequence of directions $h_t\bydef(h_{A,t},h_{b,t})\to(h_{A},h_{b})$.
The derivative takes the form
\begin{equation}
\lim_{t\to0} \frac{\funLValO[P,t]-\funLValO[P]}{t}=\min_{ {\theta}\in \argminT(\measTrue)}\max_{ {\lambda}\in\argminL(\measTrue)} \{ {\lambda}^{\prime}(  h_{A} \theta- h_{b})\},  \label{prog:minmaxDerivative}
\end{equation}
where $\argminT(\measTrue)$ and $\argminL(\measTrue)$ are sets of primal  and dual optima for  \eqref{prog:minP}.
Both sets can be nonsingletons as illustrated below.

\begin{example}[Bivariate linear IV with an interval outcome]\label{exa:runningExample} 
Suppose that $\theta\in\R^2$, $X=Z$, and regressors $Z_{1}$ and $Z_{2}$  take values in $\left\{ 0,1\right\} $ with $\expGeneric Z_{1}=\expGeneric Z_{2}=\frac{1}{2}$. 
As in \eqref{eq:intervalOutcome}, the identified set for $\theta$ can be characterized by eight inequality constraints,
\begin{equation}
\expGeneric\left[\underline{Y}\psi_z(Z)\right]\leq \expGeneric\left[Z_{1}\psi_z(Z)\right] \theta_{1}+\expGeneric\left[Z_{2}\psi_z(Z)\right]\theta_{2}\leq\expGeneric\left[\bar{Y}\psi_z(Z)\right], \label{eq:fullsystem}
\end{equation}
where indicator functions $\psi_z(Z)=\Ch{Z=z}$ correspond to all combinations of $z\in \{0,1\}^2$.
Suppose, for illustrative purposes, we are interested only in the identified set $\setID$  defined by the following  subsystem of four inequalities:
\begin{equation}
\expGeneric\left[\underline{Y}Z_{1}\right]\leq\frac{1}{2}\theta_{1}+\theta_{2}\expGeneric\left[Z_{1}Z_{2}\right]\leq\expGeneric\left[\bar{Y}_{i}Z_{1}\right],\\
\expGeneric\left[\underline{Y}\left(1-Z_{1}\right)\right]\leq  \theta_{2}\expGeneric\left[\left(1-Z_{1}\right)Z_{2}\right]\leq\expGeneric\left[\bar{Y}\left(1-Z_{1}\right)\right].\label{eq:subsystem}
\end{equation}
In order to represent the identified set on a diagram, suppose further that the a.s. bounds on the outcome variable  $Y$ satisfy $\expect[\bar{Y}|Z_1=i]=- \expect[\underline{Y}|Z_1=i] = \frac{1}{2}\Delta_{i}\geq0 $  for $i\in\{0,1\}$ (that is, $\Delta_{i}$ is the average length of the outcome interval depending on $Z_1$). 
\begin{figure}[H]
\begin{centering}
\begin{tabular}{ccc}
\includegraphics[scale=0.5,trim={2cm 0 1cm 0 }]{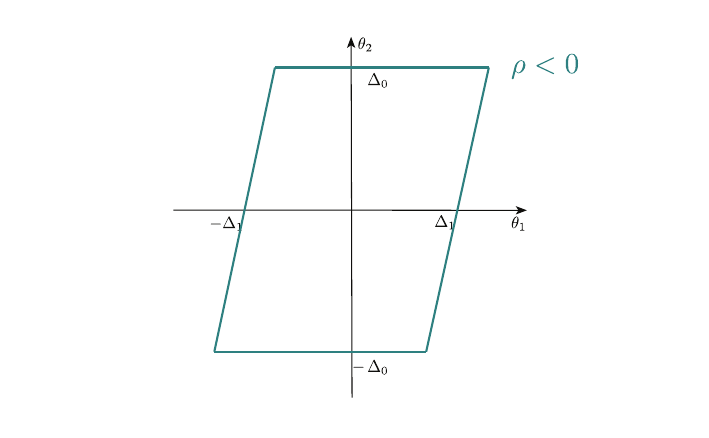}& \includegraphics[scale=0.5,trim={2cm 0 1cm 0 }]{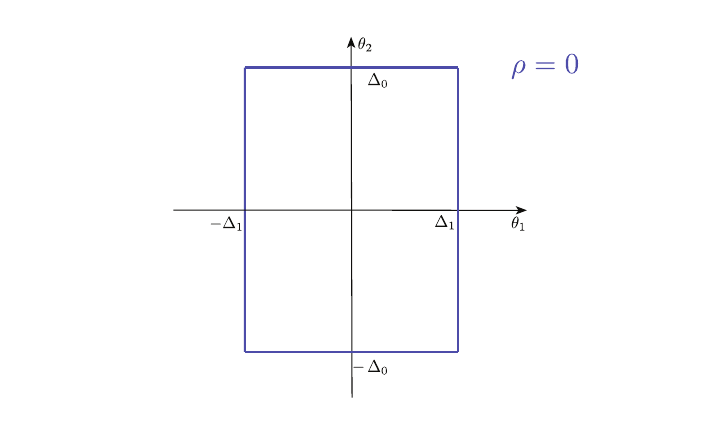}& \includegraphics[scale=0.5,trim={2cm 0 1cm 0 }]{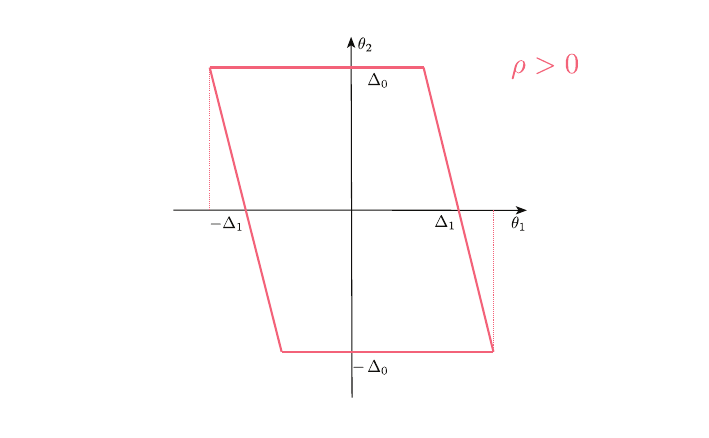}\tabularnewline
\end{tabular}
\par\end{centering}
\caption{\label{fig:Identified-sets-Example2}The identified sets in Example~\ref{exa:runningExample}  for various values of $\rho$.}
\end{figure}

The shape of the full identified set $\setID$ depends on the value of $\rho\bydef \expGeneric\left(Z_{1}Z_{2}\right)$. 
The corresponding marginal identified set for $\theta_1$ can be written in explicit form, $\setMarginal =\left[-\Delta_{1}-2\abs{\rho}\Delta_{0},\Delta_{1}+2\abs{\rho}\Delta_{0}\right]$.
If $\rho=0$,  $\argminT_2$---the second coordinate of the set of primal optima of the program in \eqref{prog:minP}---is not uniquely defined.

The case of a non-singleton set of dual optima occurs when the number of binding inequalities at the primal optimum is larger than the dimension of the parameter space. 
Suppose that we further restrict the parameter space by imposing $\theta_2=0$.
Consider the following two moment inequalities from system \eqref{eq:fullsystem}:
\begin{equation}
- \frac{1}{2} \theta_{1} \leq -\expGeneric\left[\underline{Y}Z_{1}\right],\\ 
 - \frac{1}{2} \theta_{1} \leq  -\expGeneric\left[\underline{Y} \right].\label{eq:subsystemIntersectionBounds}
\end{equation} 
The support function at $e_1$ corresponding to \eqref{eq:subsystemIntersectionBounds} takes the explicit form 
\begin{equation}
    \funLValO = \argminT_1 = 2 \min \{\expGeneric \underline{Y}Z_{1},\expGeneric\underline{Y}\}.
\end{equation}
If $\expGeneric\underline{Y}=\expGeneric\underline{Y}Z_{1}$, then both inequalities in \eqref{eq:subsystemIntersectionBounds} are binding at the optimum,  resulting in a non-singleton dual-optimum set (compare with the parameter-on-the-boundary problem and  the intersection-bounds problem  considered in \citet{andrews2001testing} and \citet{chernozhukov2013intersection}, respectively).
Indeed, the  dual formulation of the support function at $e_1$  takes form
\begin{align}
\funLValO=\max_{ {\lambda}\in\R_{+}^{2}} & \left\{    \expGeneric\left[\underline{Y} \right] \lambda_1+  \expGeneric\left[\underline{Y}Z_1 \right] \lambda_2 \right\} =   \expGeneric\left[\underline{Y} \right]  \max_{ {\lambda}\in\R_{+}^{2}}   \left\{    \lambda_1+  \lambda_2 \right\} \label{eq:DualProgramIntersectionBounds},\\
\text{s.t. } &   -\frac{1}{2}\lambda_1- \frac{1}{2}\lambda_2 =   -1.\nonumber 
\end{align}
This program has the constant value of the objective function on the entire optimization domain, which coincides with its argmin set,
\begin{equation}
   \argminL =\{\lambda\in\R_+^2| \lambda_1+ \lambda_2 =2\}. \QEDB
\end{equation}
\end{example}

The derivative \eqref{prog:minmaxDerivative} depends on particular optimal primal and dual solutions that are selected by a given perturbation $(h_{A},h_{b})$, unless both sets of solutions are singletons.
I refer to parameter values $\Ap$ and $\bp$ resulting in nonsingleton (primal or dual) solutions as \emph{nonregular}.

\cite{shapiro1991asymptotic} proposes a generalization of the delta method for the directionally differentiable functions of asymptotic normal estimators.
Under Condition~\ref{con:slater}, the sample analog of $\funLValO$ has a solution (and thus is well defined) with probability approaching 1. 
Theorems 3.4 and 3.5 in \cite{shapiro1991asymptotic} provide an asymptotic distribution of the sample support function in a given direction.
It takes form
\begin{equation}
  \frac{1}{\sqrt{n}}(\hat{\underline{v}}_n-\funLValO[P]) \wTo \min_{ {\theta}\in \argminT(\measTrue)}\max_{ {\lambda}\in\argminL(\measTrue)} \{ {\lambda}^{\prime}(  \eG_{A}(P) \theta- \eG_{b}(P))\},  \label{eq:asyDistributionShapiro}
\end{equation}
where $\eG(P)\bydef(\eG_{A}(P),\eG_{b}(P))$ is the limiting zero-mean Gaussian process for $\eG_n(P)\bydef\frac{1}{\sqrt{n}} \sum_{i=1}^n(\data_i-\expect \Data)$ indexed by $\measTrue$.
In general, the limit \eqref{eq:asyDistributionShapiro} is non-Gaussian since either of the two sets $\argminT$ and $\argminL$ can be non-singleton. (It does reduce to a Gaussian distribution when both sets are singletons.)

Several robust methods have been designed for inference on components of $\theta$ defined by moment inequalities (nonlinear, in general).
The dominant approach has been to test individual values of $\theta$ or their subvectors and then to invert the tests \citep[for example,][]{chernozhukov2007estimation,andrews2001testing,andrews2019inference,chernozhukov2019inference, cox2019simple,kaido2019confidence}. 
The main advantage of the test-inversion approach is that fixing $\theta$ simplifies the asymptotic distribution of the test statistics and allows for valid inference under weak assumptions (in particular, one can allow for an empty true identified set as in \citet{andrews2019spuriousinference}). 

Although having attractive statistical properties, the test-inversion methods can become computationally intractable in settings with high-dimensional parameters $\theta$ since they are based on grid search and resampling methods.\footnote{See further discussion of the computational issues in Appendix \ref{sec:Discussion}.} In contrast, delta-method  inference using the support function estimators for a fixed direction proposed in this paper remains a computationally tractable (frequentist) option in  high-dimensional settings such as mentioned in Section \ref{subsec:examples} since it fully uses the linear programming structure of the problem.

Inference using the sample support function for a fixed direction in the non-regular case faces three challenges: (i) nondifferentiability makes the standard bootstrap inconsistent \citep[see][]{Fang2018inference}; (ii) discontinuity in the directional derivative \eqref{prog:minmaxDerivative} results in poor (uniform) approximation of \eqref{eq:asyDistributionShapiro} by the numerical bootstrap \citep[see][]{dumbgen1993nondifferentiable,Hong2018numerical};
(iii) the estimator is necessarily (asymptotically) biased \citep[see][]{hirano2012impossibility}. 
The regularized support function estimator proposed in the next section addresses these challenges in a robust and computationally tractable way.

%%%%%%%%%%%%%%%%%%%%%%%%%%%%%%%%%%%%%%%

\section{Main results}\label{sec:results}
\subsection{\label{subsec:Regularization} Bounds based on the regularized primal program }

The main source of inference complications, the lack of smoothness in the support function, resolves itself when the sets of primal, $\argminT,$ and dual, $\argminL,$ optima are singletons (see equations \eqref{prog:minmaxDerivative} and \eqref{eq:asyDistributionShapiro}). Consequently, my proposal is to consider a regularized support function at $e_1$ that has a unique solution and approximates the original support function from a known direction, either from above or from below. 
\footnote{\label{rep1.4}If both  primal and dual solutions of the regularized program are unique, all the directional derivatives of the regularized support function at $e_1$  given by  \eqref{prog:minmaxDerivative}  coincide and are given by formulas
\begin{equation}
\frac{\partial \minV\mnP }{\partial (\Ap)_{ij}}=\argminL_i\mnP\argminT_j \mnP\text{ and } \frac{\partial \minV \mnP}{\partial (\bp)_i}=-\argminL_i\mnP,
\end{equation}
where  $\argminT \mnP$ and $\argminL\mnP$ are, correspondingly, the primal and dual solutions  of the regularized program with a regularization parameter $\mn$. }
Since the direction of such a regularization bias is known, the corresponding estimators of the regularized support  function  at $e_1$  admit  standard one-sided  delta-method and bootstrap inference  based on the normal limiting distribution. 

Without regularization, both primal and dual solutions of \eqref{prog:minmax} can be non-singletons. 
To keep the problem analytically tractable, I   focus on the case in which the dual set is a singleton by assumption, and I propose a regularization for the primal problem.
In this way, I can explicitly consider only the bias from the primal regularization and develop the necessary bias correction. 
Appendix Section~\ref{subsec:overidentification} discusses the complementary case of dual regularization.
 
The following \emph{regularized primal program} is strictly convex for any $\mu>0$ and hence has a unique primal solution\label{reply:R1.4} $\argminT (\mu, P)$ and approximates the optimal value of program~(\ref{prog:minP}):
\begin{align}
\funLVal[\mu ] & =\underset{ {\theta}\in\setID}{\mbox{min}}\left\{\eOne {\theta}+\mu \ltwo{ {\theta}}^{2}\right\}.\label{prog:regProg}
\end{align}

The corresponding dual program has a unique solution $\argminL(\mu,P)$ iff the set of constraints satisfies the linear independence constraint qualification (LICQ; see \citet[p.178]{shapiro1991asymptotic}  and \cite{ wachsmuth2013licq}), 
\begin{condition}[LICQ]
\label{con:LICQ}  The matrix of gradients of binding constraints has a full rank for any $ {\theta}\in\setID$. 
\end{condition}

\begin{rem}
There is a purely computational reason to ensure that LICQ holds. 
If it is violated,  Newton-type algorithms, which typically guarantee a (fast) quadratic rate of convergence to a stationary point,  have a linear rate of convergence or do not converge at all (see, for example, \cite{golishnikov2006newton}). 
\end{rem}

The set of DGPs that satisfy LICQ is not closed since a limit of a sequence of linearly independent matrices can be a reduced rank matrix. 
This means that depending on the size of the smallest singular value of the gradients of the set of binding constraints at any point may be arbitrary small while still satisfying LICQ.
As a result, the dual solutions $\argminL(\mu,P)$ may be arbitrary large, implying large derivatives of $\minV\mnP $ 
and may require very high sample sizes for reasonable precision of the delta-method inference (see equation \eqref{eq:boundsOnLinearApproximation} in Lemma \ref{lem:boundOnVgrowth} in Appendix Section \ref{app:smoothness} for details).
In the next section, I provide sufficient conditions for LICQ that explicitly ensure that the class of DGP under consideration is closed, so that we can uniformly control the quality of the delta-method inference within this class. 
Then I illustrate these conditions in the context of Example~\ref{exa:runningExample}. 

\subsubsection{Testable sufficient conditions for uniqueness of dual solutions}\label{subsec:UniqueDual}
LICQ is often considered  hard to verify \citep[see][]{kaido2019constraint}. 
In fact, a direct test of this assumption would face two problems: (i) the set $\setID$ is unknown but can only be estimated with an error; (ii) the set of binding inequalities at each $\theta$ is also unknown.
Because of their multiple-testing nature, both problems would complicate inference beyond the practical level.  
Moreover,   LICQ does not provide an explicit bound on the dual variables required for uniform validity analysis.

Some authors, including \cite{hsieh2017inference}, make a high-level assumption about boundedness of the dual variables.
Others (KMS and \cite{andrews2019inference}) focus on test inversion and thus only require restriction on non-degenerate covariance matrix of the moment functions.
Both KMS and \cite{andrews2019inference} compute critical values of a test at a particular point $\theta$  which they then invert. (\cite{kaido2019confidence} additionally propose a method for computationally efficient interpolation of confidence sets based on test inversion.)
KMS, for example, only need a bounded dual variable for each of the bootstrap draws for their purposes.   
The local linear program that is used in KMS bootstrap has a bounded dual variables almost surely under the aforementioned covariance constraints since its inequality constraints have random bootstrapped coefficients.
In contrast, I need stronger assumptions to be able to estimate the nuisance parameter, the argmin set $\argminT(\mnP)$ and the dual solutions $\argminL(\mnP)$, which allows me to avoid the test inversion stage and achive compuational gains.

I propose a sufficient condition for LICQ in the form of bounds on the values of two auxiliary optimization programs.
The values of these programs, in turn, explicitly characterize an upper bound on the dual variables, allowing for a uniform asymptotic analysis. 
Specifically, within this class of DGP satisfying this assumption, we can uniformly control the quality of the delta-method inference by establishing an explicit uniform bound on the relevant second-order directional derivatives.

To introduce the sufficient conditions for LICQ, I  use the following notation.
For any $\jSet\subset\setIneq$, let the matrix $\activeMat=(e_{i_1},...,e_{i_\ell})^\prime$ correspond to $ \activeSet\bydef \setEq \cup \jSet = \{i_1,...,i_\ell\} $, a set of active constraints.
Let $\eta_1(\cdot)$ be the smallest left singular value function of a matrix; that is, $\eta_1(A) \bydef \sqrt{ \min_{u}( u^\prime  A A^\prime u/ u^\prime  u ) }$ .
The sufficient conditions can now be formulated as the following two assumptions on every submatrix $\activeMat (\Ap|\bp)$ that include all $p$ equality constraints and either $d-p$ or $1+(d-p)$ inequality constraints. 

\begin{assumption}\label{ass:ULICQ} Measure $\measTrue$  satisfies two conditions:

\begin{enumerate}[label={\textbf{\Alph*.}},
  ref={\theassumption.\Alph*}]
 \item \label{assu:RankCondition} For any combination $\jSet$ consisting of  all  $p$ equality constraints and $d-p$ inequality constraints, the corresponding submatrices of coefficients $\activeMat (\Ap|\bp) $ of the full set of constraints $ (\Ap|\bp)$  have singular values that are uniformly bounded from below by a positive number $\eta(\measTrue)$.
 
 \item \label{assu:NoOveridentification} Any combination  $\jSet$ consisting of all  $p$ equality constraints and $d-p+1$ inequality constraints  cannot be simultaneously satisfied as equality at any point $\theta\in\setID$.
  \end{enumerate}
\end{assumption}

Assumption \ref{ass:ULICQ} can be summarized using two characteristics:
\begin{align}
\eta(\measTrue)&\bydef\min_{\text{s.t. }\begin{matrix}
\jSet  \subset\setIneq \\
\card\jSet  = d-p 
\end{matrix}}\eta_1\left(\activeMat (\Ap|\bp) \right) >0,\label{eq:rankD}\\ 
 s(\measTrue)&\bydef\min_{\text{s.t. } \begin{matrix}
\jSet\subset  \setIneq \\
\card\jSet =  d-p+1\\
{\theta}\in  \setID
\end{matrix}}\ltwo{\activeMat  (\Ap\theta-\bp)} >0.\label{eq:GenericMomentConditions-1}
\end{align}
These numbers measure how close a given DGP $\measTrue$ to a violation of LICQ condition in population and determine how many observations are required to meet LICQ and have non-empty feasible set for a sample analog of program \ref{prog:regProg} with a given probability (see Lemma \ref{lem:wellDefined} in Appendix).

Both characteristics $\eta(\measTrue)$ and $ s(\measTrue)$ can be consistently estimated using a plug-in approach.
As long as a sample analog of $\setID$ is nonempty, sample analogs of $\eta(\measTrue)$ and $ s(\measTrue)$ will be generically positive.
In principle, a critical value can be obtained for a formal test of hypothesis $\eta(\measTrue)\geq \underline{\eta}$ and $s(\measTrue)\geq \underline{s}$  for any given pair of numbers $\underline{\eta},\underline{s}$ along the lines of \cite{cragg1997inferring}.\footnote{See also Appendix Remark \ref{rem:contastLICQ2} for an alternative representation of Assumption \ref{ass:ULICQ} as a single characteristic minimal bound on a different subset of submatrices  of $ (\Ap|\bp)$.}
I leave this formal test for future research.

Assumption~\ref{ass:ULICQ} rules out more than $d$ binding inequality constraints at any point $\theta\in\setID$. 
There are some special empirical applications resulting in a singleton identified set $\setID$ in which such \emph{overidentifying} inequality constraints can appear, which can result in multiplicity of dual solutions.\footnote{See, for example, \cite{gafarov2014identification} and \citet{shi2018estimating} who consider cases of infinitely many inequality conditions.}
This multiplicity can be eliminated using a regularization of the dual program \eqref{eq:DualProgram0}; that is, Assumption~\ref{assu:NoOveridentification} can be relaxed within the framework proposed in this paper, but such an extension is left for future research.
I briefly discuss this proposal in  Appendix Section \ref{subsec:overidentification}.

It is instructive to see what Assumption~\ref{ass:ULICQ} implies for our running example.
\renewcommand\thmcontinues[1]{Continued}
\begin{example}[ continues=exa:runningExample   ]\label{exa:runningExample5}
In this setup, there are four moment inequality conditions defined in \eqref{eq:subsystem}. They correspond to the following matrix of coefficients: 
\begin{equation}
    (\Ap|\bp) = \left(\begin{array}{cc|c}
\frac{1}{2} & \expGeneric\left[Z_{1}Z_{2}\right] & \expGeneric\left[\bar{Y} Z_{1}\right]\\
-\frac{1}{2} & -\expGeneric\left[Z_{1}Z_{2}\right] & -\expGeneric\left[\underline{Y}Z_{1}\right]\\
0 &\expGeneric\left[\left(1-Z_{1}\right)Z_{2}\right]& \expGeneric\left[\bar{Y}\left(1-Z_{1}\right)\right]\\
0 & -\expGeneric\left[\left(1-Z_{1}\right)Z_{2}\right]& -\expGeneric\left[\underline{Y}\left(1-Z_{1}\right)\right]
\end{array}\right)
\end{equation}
Here $p=0$, so to check Assumption~\ref{assu:RankCondition} one need to consider all subsets with two rows out of four, six combinations in total. 
For example, the submatrix with rows $\jSet=\{1,2\}$ takes form
\begin{equation}
    \activeMat (\Ap|\bp) = \left(\begin{array}{cc|c}
\frac{1}{2} & \expGeneric\left[Z_{1}Z_{2}\right] & \expGeneric\left[\bar{Y}Z_{1}\right]\\
-\frac{1}{2} & -\expGeneric\left[Z_{1}Z_{2}\right] & -\expGeneric\left[\underline{Y}Z_{1}\right]
\end{array}\right).
\end{equation}
This matrix has a full row rank iff $\expGeneric\left[\bar{Y}Z_{1}\right]\neq \expGeneric\left[\underline{Y}Z_{1}\right]$ or $\Delta_1>0$. As result, we just verified that for $\jSet=\{1,2\}$ we have $\eta_1\left(\activeMat (\Ap|\bp) \right) >0$.
Similarly, the matrix corresponding to rows $\jSet=\{3,4\}$ has full rank iff both $\Delta_0>0$ and $\expGeneric\left[\left(1-Z_{1}\right)Z_{2}\right]\neq0$.
Under those conditions, submatrices with pairs of rows $\jSet\in\big\{\{1,3\},\{1,4\},\{2,3\}\{2,4\}\big\} $  are also all full rank, and thus their left singular values are all positive.
To summarize,   Assumption~\ref{assu:RankCondition} is satisfied iff   
\begin{align}
&\Delta_0>0,\Delta_1>0, \label{eq:non-trivial-intervals}\\
&\expGeneric\left[\left(1-Z_{1}\right)Z_{2}\right]\neq0.\label{eq:noMulticollinearity}
\end{align}
What do these conditions mean in practical terms?
Inequalities~\eqref{eq:non-trivial-intervals} imply that the upper and lower bounds on $Y$ are different from each other on average, conditional on $Z_1$, while  inequality~\eqref{eq:noMulticollinearity}  implies that the instrument $Z_2$ (with values in $\{0,1\}$) is not perfectly correlated with $Z_1$. 
In fact, in the case where the bounds $\underline{Y}$ and $\bar{Y}$ coincide with probability 1, one can replace the corresponding pair of moment inequalities with a single equality.
The degenerate case $\expGeneric\left[\left(1-Z_{1}\right)Z_{2}\right]=0 $  is analogous to the multicollinearity problem in the usual linear regression setup.  

Inequalities~\eqref{eq:non-trivial-intervals}   imply that Assumption~\ref{assu:NoOveridentification} is satisfied. 
To illustrate a violation of Assumption~\ref{assu:NoOveridentification}, suppose that we add one more moment condition corresponding to instrument $Z_2$, $$\expGeneric\left[\underline{Y}Z_{2}\right]\leq  \theta_{1}\expGeneric\left[Z_{1}Z_{2}\right]+ \frac{1}{2}\theta_{2}.$$ 
Assumption~\ref{assu:NoOveridentification} would be violated if, for example, this additional inequality was binding at the corner points of the original identified set, $\theta=(-\Delta_{1}\mp 2\rho\Delta_{0}, \pm \Delta_{0})$. 

It is worth contrasting Assumption~\ref{ass:ULICQ} with LICQ:
the latter only restricts the gradients of the submatrices of $\Ap$ at any point where the corresponding constraints are active. 
It turns out that instead of checking the active constraints at any given point, one can restrict the singular values of submatrices $\activeMat (\Ap|\bp)$.
In this simple example, we can manually verify that LICQ holds under Assumption~\ref{ass:ULICQ}  (the general proof is given in Lemma \ref{lem:LICQ} in the Appendix Section \ref{subsec:LICQdiscussion}).
First, let us consider pairs of constraints with collinear gradients that potentially could violate LICQ.
For example, the submatrix of $\Ap$ corresponding to rows $\jSet=\{1,2\}$,
\begin{equation}
    \activeMat  \Ap  = \begin{pmatrix}
\frac{1}{2} & \expGeneric\left[Z_{1}Z_{2}\right]  \\
-\frac{1}{2} & -\expGeneric\left[Z_{1}Z_{2}\right]  
\end{pmatrix},
\end{equation}
is always a reduced rank matrix. 
Nevertheless, the corresponding two inequality constraints cannot be binding simultaneously by  the Rouch\'e-Capelli  theorem (also known as Kronecker–Capelli theorem) since  $ \activeMat (\Ap|\bp)$ for them has rank equal  to 2 which  is larger than 1, the  rank of $\activeMat \Ap$.
Correspondingly, these two constraints do not violate LICQ at any point (since they do not intersect) despite having collinear gradients.
Second, for constraint pairs like $\jSet=\{1,3\}$ we have $ \rank (\activeMat  \Ap )= \rank ( \activeMat (\Ap|\bp)) = 2$, and hence the corresponding corner point (intersection of these constraints) exist as a unique solution to $(\activeMat  \Ap) \theta = \activeMat \bp$. That corner point does not violate  LICQ since $ \rank(\activeMat  \Ap )=d=2$.
By this logic, we can prove that all the points in this example have at most 2 binding constraints, all of which have  linearly independent gradients. 
Third, the faces of the polygon are always defined by a subset of the inequalities defining their corners.
Hence the inequalities defining the faces are also linearly independent (a matrix of full rank has all submatrices of full rank) and any point on a face does not violate LICQ.
Finally, we do not need to check the LICQ in the internal points of $\setID$ since there are no binding constraints at those points by definition. 
Thus, in this example, we just verified that LICQ indeed holds under Assumption  \ref{ass:ULICQ}. $\QEDB$
\end{example}

Appendix Section \ref{subsec:LICQdiscussion} contains further technical details about Assumption \ref{ass:ULICQ}.

\subsubsection{Tighter bounds on the support function at a given direction} \label{subsec:tighterBounds}

Clearly, for any positive $\mu$, the value of the regularized program \eqref{prog:regProg} $\minV(\mu\measTrue)$ is larger than $\minV (\measTrue)$ by at least $\mu\ltwo{\argminT(\mu,\measTrue)}^{2}$.
A tempting approach would be to correct for this regularization bias by subtracting $\mu\ltwo{\argminT(\mu,\measTrue)}^{2}$. 
Unfortunately, this correction makes the corrected value function $\minV(\mu,\measTrue)-\mu \ltwo{\argminT(\mu,\measTrue)}^{2}$ non-differentiable in the paramteters $A_P,b_P$  if  the non-regularized value function   $\minV (\measTrue)$ itself is non-differentiable. 
So to preserve differentiability of the bias-corrected value function, which is crucial for delta-method inference,  one cannot use argmin of the actual regualrized programm for bias corrections. 
Instead, I suggest to tighten the bounds  on $\minV (\measTrue)$ using one the following two corrections,
\begin{align}
\minVin(\mu,\kappa,\measTrue) &\bydef \minV(\mu,\measTrue)-\mu \ltwo{\argminT(\kappa,\measTrue)}^{2},\label{eq:innerBound} \\
\minVout(\mu,\measTrue) &\bydef\minV(\mu,\measTrue)-\mu \ltwo{\theta^*}^{2},\label{eq:outerBound}
\end{align}
where $\argminT(\kappa,\measTrue)$ is an argmin of the regularized program \eqref{prog:regProg} with a larger tuning parameter $\kappa$ instead of $\mu$, $\theta^*$ is  any point in $ \argminT(\measTrue)$. 
As $\mu$ and $\kappa$ shrink to zero, these bounds continuously shrink to $\minV (\measTrue)$ above and below, respectively.  The expressions in \eqref{eq:innerOuterBounds} provide valid conservative bounds from above and below for $\minV(\measTrue)$ that are useful for uniform one-sided inference (coverage probability in this case can be higher than nominal level). For the remainder of the paper, the focus will be on the confidence bounds that cover $\minV (\measTrue)$ from below. So $\minVout(\mu,\measTrue) $ will play the major role. The uses of $\minVin(\mu,\kappa,\measTrue)$ for uniform inference are discussed in Appendix Section \ref{subsec:overidentification} 

\begin{thm}
\label{thm:bounds}For any DGP parameterized by $\measTrue$ satisfying Assumption~\ref{assu:Non--empty} and any $\kappa\geq \mu \geq 0$, the following bounds hold:
\begin{equation}
\minVout(\mu,\measTrue) \leq \minV(\measTrue) \leq \minVin(\mu,\kappa,\measTrue). \label{eq:innerOuterBounds}
\end{equation}
If,  in addition,  $\measTrue$ satisfies  Assumption ~\ref{ass:ULICQ}--\ref{assu:NoOveridentification}, then there exist $\bar{\mu}(\measTrue)>0$ such that $\minVin(\mu,\kappa,\measTrue)=\minV(\measTrue)$
for any fixed $\mu<\kappa<\bar{\mu}(\measTrue)$. Furthermore, if $\argminT(\measTrue)$ is a singleton, then $\minVout(\mu,\measTrue)=\minV(\measTrue)$ for any $\mu<\bar{\mu}(\measTrue)$.
\end{thm} 
 
 \begin{proof}
 See Appendix Section~\ref{sec:Proof-of-Theorem1}.
 \end{proof}

% The bounds \eqref{eq:innerOuterBounds} are valid even in the general nonlinear case, which is beyond the scope of this paper. 
% For linear programs, the stronger result holds:  there exists a finite cutoff value $\bar{\mu}(\measTrue)>0$ at which the upper bound  $\minVin(\mu,\kappa,\measTrue)$ coincides with $\minV(\measTrue)$.
Although for a given DGP $P$, characterized by $(A_P,b_P)$, the cutoff $\bar{\mu}(\measTrue)$ is well defined, it can change discontinuously as a result of a small change in $(A_P,b_P)$. It makes a consistent estimation of $\bar{\mu}(\measTrue)$ a challenging task. 
So instead of trying to estimate $\bar{\mu}(\measTrue)$, I suggest using two appropriately chosen shrinking sequences $\mu_n$ and $\kappa_n$ instead of fixed tuning parameters $\mu$ and $\kappa$.

The gap   between $\minV(\measTrue)$  and $\minVout(\mu_n,\measTrue) $  is at least $\mu_n (\ltwo{\theta^*}^{2}-\ltwo{\argminT(\mu_n,\measTrue)}^2)\geq0$; for $   \minVin(\mu_n,\kappa_n,\measTrue) $ the gap is at most $\mu_n (\ltwo{\argminT(\mu_n,\measTrue)}^{2}-\ltwo{\argminT(\kappa_n,\measTrue)}^2)\leq0$.
In the nonregular case, that is where $\argminT(\measTrue)$ is non-singleton, the gap for $\minVout(\mu_n,\measTrue) $ is shrinking to $0$ at rate $\mu_n$ which still results in a consistent estimators of $\minV (\measTrue)$ at rate $\mu_n$ that can be slightly slower than the regular parametric rate $1/\sqrt{n}$ (for example, $ {\sqrt{\ln{n}}}/{\sqrt{n}}$).
As a result, the corresponding confidence sets cover the projections of the sharp identified set, not an enlargement of it (in contrast to other studies, including \cite{cho2023simple}, who do not use consistent estimators of the sharp bounds on the identified set and study an enlarged identified set instead).
The gap for $\minVin(\mu_n,\kappa_n,\measTrue)$ becomes exactly equal to zero for some sufficiently large $n$ in both regular and nonregular cases, resulting in exact corresponding point-wise one-side confidence sets.

\subsection{\label{subsec:inference} Large-sample results for the estimated regularized support function}

Consider the analog estimator of $\funLVal[\mu_{n}]$ for some sequence $\mu_{n}$,
\begin{equation}
\funLValHat[\mu_{n}]\bydef\underset{ {\theta}\in\setID[\measEmp]}{\mbox{min}}\left\{\eOne {\theta}+\mu_{n}\ltwo{ {\theta}}^{2}\right\} \label{prog:minSSAregularized}.
\end{equation}
As we shall see in this subsection, this estimator admits an asymptotic linear (Bahadur-Kiefer) expansion with an explicit approximation error, which determines the precision of asymptotic normal and bootstrap inference. 
This allows me to establish the validity of the corresponding inference methods \emph{uniformly} over a reasonable class of DGP. 
Uniform asymptotic validity is crucial for small-sample performance of inference methods in the presence of possible discontinuous changes of the asymptotic distribution of estimators (for example, $\funLValHat[\mu_{n}]$) with respect to DGP parameters (for example, $\measTrue$).

\subsubsection{A class of DGPs under consideration}

I consider the class of all measures $\Measures=\Measures(\underline{\eta},\underline{s},\varepsilon,\bar{M})$ that satisfy Assumptions~\ref{assu:Non--empty}-\ref{assu:Moments}  with some    positive constants $\underline{\eta}$, $\underline{s}$, $\varepsilon$, and $\bar{M}$.
\begin{assumption} 
\label{assu:Moments}There exist  $\varepsilon>0$  and $\bar{M}<\infty$ such that
\begin{align}
\expect\left\Vert \Data\right\Vert  & ^{2+\varepsilon} <\bar{M}.
\end{align}
\end{assumption}
To summarize, every  $\measTrue\in \Measures$ satisfies $\setID\neq \emptyset$,
$\eta(\measTrue)>\underline{\eta}$, $s(\measTrue)>\underline{s}$,  and $\expect\left\Vert \Data\right\Vert^{2+\varepsilon} < \bar{M}^{2+\varepsilon}$. \label{reply:R3.5.c.precompact}
% \footnote{Note that the class of distributions $\Measures$ is, by construction, precompact in the weak topology, i.e. closure of this class is compact.  This compactness property simplifies the uniform convergence arguments, since most of the functions used in the proofs are continuous and which make them uniformly bounded on the compact.}

The class $\Measures$  includes DGPs  (parameterized by measure $\measTrue$) with multiplicity of primal solutions of \eqref{prog:minP} but assumes uniqueness of dual solutions. 
This property is necessary to show that program~\eqref{prog:minSSAregularized} has a nonempty sample argmin and a unique vector of Lagrange multipliers in large enough samples with probability approaching 1 uniformly in $\measTrue\in\Measures$ as $n\to\infty$ (see  Lemma~\ref{lem:wellDefined} in Appendix Section \ref{sec:Proof-of-Theorem2}). 
As discussed in Section~\ref{subsec:UniqueDual}, one can study a case with multiple dual solutions and a unique primal solution in a similar way (see Appendix Section~\ref{subsec:overidentification}).

\subsubsection{A higher-order envelope theorem and a strong approximation of the regularized support function}
A Bahadur-Kiefer expansion of estimator defined in program~\eqref{prog:minSSAregularized} can be established using the envelope theorem \eqref{prog:minmaxDerivative}.
As discussed in Section~\ref{subsec:asymptoticDistribution}, this theorem typically holds for the sample support function at $e_1$ .
However, this theorem  does not provide bounds on the higher-order directional derivatives.
So, without additional assumptions, the directional delta method of \cite{shapiro1991asymptotic} does not provide means to evaluate the error of the asymptotic approximation of the sample support function at $e_1$  by \eqref{eq:asyDistributionShapiro}. 
Since the limiting distribution changes discontinuously with the DGP $\measTrue$, poor performance of asymptotic methods based on the limit \eqref{eq:asyDistributionShapiro} should be expected when the parameter $\measTrue$ is close to a discontinuity point.
In other words, inference procedures that estimate \eqref{eq:asyDistributionShapiro} directly (for example, subsampling or directional bootstrap) will inevitably be only point-wise valid and can have poor finite-sample performance in such cases. (It is an implication of the impossibility theorem of \citet{hirano2012impossibility} for a functional that is directionally differentiable.)  

In contrast, the regularized support function at $e_1$  admits a stronger version of the envelope theorem that provides a bound on the error of the linear approximation uniformly over $\measTrue\in\Measures$.
I developed this novel bound using a \emph{second-order directional Taylor expansion} of a system of generalized inequalities that define the optimal solutions to the regularized program (see Lemmas~\ref{lem:directionalDerivative} and \ref{lem:boundOnVgrowth} in  Appendix Section \ref{app:smoothness}).
Using this result, the asymptotic linear (Bahadur-Kiefer) representation of the value function in program~\eqref{prog:minSSAregularized} follows almost immediately  (see  Lemma~\ref{lem:Bahadur} in Appendix  Section \ref{sec:Proof-of-Theorem2}):
  \begin{equation}
 \sqrt{n} ( \funLValHat[\mu_{n}]- \minV\mnP ) = \frac{1}{\sqrt{n}}\sum^{n}_{i=1}\argminL\mnP^\prime\momentVec{\data_i}{\argminT\mnP}   +O_\Measures(\frac{1}{\mn \sqrt{n}}). \label{eq:Bahadur}
\end{equation}
The first term on the right-hand side of this representation is a scaled sample average of a zero-mean random variable, which admits a uniform Gaussian approximation. (Indeed, the binding constraints have zero mean at $\argminT$, while the nonbinding constraints are multiplied by zero dual variables $\argminL$.) 
The residual term $O_\Measures(1)$ denotes a uniformly tight sequence of  a random process indexed by  $\measTrue\in\Measures$. Analogously, I  denote any random sequence  as  $ o_\Measures(1)$ if 
$\lim_{n\to\infty}\sup_{\measTrue \in \mathcal{P}}\measTrue\left( \ltwo{\zeta_n(\measTrue) }\geq \epsilon \right) =0$.

The Bahadur-Kiefer representation \eqref{eq:Bahadur} is important for three reasons. 
First, it suggests a coupling of $\funLValHat[\mu_{n}]$ with a Gaussian process (through the Yurinsky theorem). 
Second, it implies uniform validity of the optimization-free  score bootstrap, which is particularly computationally convenient.
Third, it suggests an  analog estimator of the asymptotic variance of the regularized support-function estimator,   
\begin{equation}
 \funLStdHat[\mu_{n}][][2]=\sMean{(\funLLagrangeHat[\mu_{n}][][\prime]\momentVec{\data_i}{\funLArgHat[\mu_{n}][]})^2},   \label{eq:standardErrorFormula}
\end{equation}
where  $\funLArgHat[\mu_{n}]$ and $\funLLagrangeHat[\mu_{n}]$ are, respectively, the optimum and the vector of Lagrange multipliers of (\ref{prog:minSSAregularized}), which are provided by  common constraint-optimization software packages. 
 
These implications are summarized in the following theorem. 

\begin{thm}
\label{thm:consistency}Consider any sequence $\mu_{n}$ such that
$\mu_{n}\to0$ and $\mu_{n}\sqrt{n}\to\infty$.  Then with probability approaching 1 uniformly in $\measTrue\in \Measures$,
\begin{align}
\lim_{n\to\infty} \sup_{\measTrue \in \mathcal{P}}&\pi(\sqrt{n}(\funLValHat[\mu_{n}]-\funLVal[\mu_{n}][\measTrue]), N\left(0,\funLStd[\mu_{n}][\measTrue][2]\right))  = 0,\label{eq:coupling}\\ % \label{eq:deltaMethodPW}
\funLArgHat[\mn] &=\funLArg[\mn] +  O_\Measures(\frac{1}{\mu_n\sqrt{n}}),\label{eq:thetaConsistency}\\ 
\funLLagrangeHat[\mn]&=\funLLagrange[\mn] +  O_\Measures(\frac{1}{\sqrt{n}}),\label{eq:lambdaConsistency}\\
 \funLStdHat[\mn]& = \funLStd[\mn]  +  o_\Measures(1). \label{eq:convergenceSpeedStd}
\end{align}
The function $\pi(\cdot,\cdot)$ is the Levy-Prohorov metric, which metricizes the weak topology of   probability measures (see \cite{van1996weak}).
%Here $O_{p}$ and $o_{p}$ do not depend on $\mu$
\end{thm} 
 \begin{proof}
 The strong approximation result \eqref{eq:coupling} is based on the uniform bound on the higher-order directional derivatives (implied by Assumptions \ref{assu:Non--empty},\ref{ass:ULICQ}) and the generalization of the \cite{yurinskii1978error} coupling proposed in \citet[][Proposition A.5.2 on p. 457]{van1996weak}  ( it is implied by \ref{assu:Moments} for i.i.d. data).
 See Appendix Section~\ref{sec:Proof-of-Theorem2} for details.
 \end{proof}

The coupling with a Gaussian random process \eqref{eq:coupling} is a  \emph{strong approximation} result, which can be understood using a geometric interpretation.
The distance between the difference $ \sqrt{n} ( \funLValHat[\mu_{n}]- \minV\mnP )$ and some sequence of zero-mean Gaussian r.v.s  $ N\left(0,\funLStd[\mu_{n}][\measTrue][2]\right)$ 
converges to zero with the same \emph{uniform} rate for all DGP $\measTrue \in \mathcal{P}$.
This property is stronger than conventional CLT-type results, since it does not require the existence of a limiting distribution for the approximating Gaussian variables. 

% (For example, the limit can vary depending on  the choice of the tuning sequence $\mu_{n}$ or     the DGP $\measTrue$, which itself can   vary in the asymptotic experiment.) 
% In practical terms, for the regularized support function at $e_1$ ,  it means the finite-sample performance of such a \emph{uniform} asymptotic Gaussian approximation will work well even if the DGP parameter $\measTrue$ results in a nondifferentiable support function in the absence of regularization (the case in which the usual sample support function at $e_1$  is no longer asymptotically Gaussian). 
% The implementation of the Gaussian approximation  can be based on the delta method or the score bootstrap using the sample analogs of the optimal solutions.
% In this paper, however, I focus on the delta-method implementation.

\subsection{Uniformly valid inference\label{subsec:pointwise-ci}}

 Theorems~\ref{thm:bounds} and \ref{thm:consistency} can be used to construct uniformly valid confidence bands (one-sided CS) for $\minV(\measTrue)$.
The corresponding generic algorithm takes the following form: \vspace{0.5cm}
\begin{enumerate}
    \item[Step 1.]  Compute the regularized sample support function $\funLValHat[\mu_{n}]$ at $e_1$  defined in \eqref{prog:minSSAregularized}.
    \item[Step 2.] Compute the standard error using \eqref{eq:standardErrorFormula}.
    \item[Step 3.] Compute a bias adjustment using sample analogs of either $\mu_{n} \ltwo{\argminT(\kappa_{n},\measTrue)}^{2}  $ (for an upper confidence band) or the norm of some point in the argmin set $\mu_{n} \ltwo{\theta^*}^{2}$   (for a lower confidence band).
\end{enumerate}
\vspace{0.5cm}

The following subsections explain specific implementations of this algorithm for uniformly valid CSs for a projection of the identified set on a single coordinate or multiple coordinates and for the argmin set of a linear program with estimated coefficients.

\subsubsection{Application to confidence sets for a scalar projection of the identified set  } \label{sec:DriftingDGP}

Let's revisit one of the primary objects of interest in the moment-inequality models: CSs on projections of the identified set, $\setMarginal\bydef\left[\funLValO,\funUValO\right]$ .

Theorems~\ref{thm:bounds} and \ref{thm:consistency}  suggest that the outer-bound  estimator for the minimal value,
\begin{equation}
\minV^{out}(\mn,\measEmp)   \bydef\funLValHat[\mu_{n}]-\mu_{n}\ltwo{\theta^*\left(\measEmp\right)}^{2},\label{eq:biasCorrectedEstimator}
\end{equation}
 is asymptotically unbiased (in the regular case) or biased downward (in the nonregular case). 
(Some downward bias is acceptable since the purpose of CSs is to cover the lowest point $\funLValO$ from below.)
To achieve this property, the estimator $\theta^*\left(\measEmp\right)$ should have a norm that is not smaller than the minimal norm in $ \argminT(\measTrue)$ with probability approaching 1.
We can use an estimator $\theta^*\left(\measEmp\right)$ that converges to the  point with the  coordinates
\begin{equation}
   \theta^*_i(\measTrue) \bydef  \max \{ {\theta^{+}_i(\measTrue)}, {\theta^{-}_i(\measTrue)}\},
\end{equation}
where
\begin{equation}
   {\theta}^{\pm}_i(\measTrue)\bydef\abs{ \underset{ {\theta}\in\setID[\measTrue],  {\theta_1}\leq\funLValO[\measTrue]+\mn} \min\left\{\pm {\theta_i}\right\} }.\label{eq:defThetaStar}
\end{equation}
By definition, $\ltwo{\theta^*}\geq \ltwo{\theta}$ for any $\theta \in \argminT(\measTrue)$. (See the proof of consistency in Lemma~\ref{lem:argminBound} in  Appendix A.)

This bound on $\ltwo{\theta^*}$ has two attractive properties. First, it can be (uniformly) consistently estimated using only $2k$ linear programs, which can be easily computed even in models with a very large-dimension $k$ and a large number of inequalities using interior-point numerical optimization methods.
Second, in the regular case, in which $\argminT(\measTrue)$ is a singleton, this estimator is asymptotically unbiased since it converges to  $ \ltwo{\theta^*} = \ltwo{\argminT(\measTrue)} $.
So for any such fixed $\measTrue$, by Theorem~\ref{thm:bounds}  we have $\minV^{out}(\mu_{n},\measTrue)=\minV(\measTrue)$ for sufficiently small $\mu_n$; that is, it is possible to construct CIs with correct (nonconservative) coverage in the regular case. 
  
Using the bias-corrected estimator  $\minV^{out}(\mn,\measEmp)$  and its analog for the upper bound, $\maxV^{out}(\mn, \measEmp)$, I construct the following  delta-method CSs:
\begin{equation}
\begin{cases}
 \mbox{CB}_{\alpha,n,\Measures} & =\left[\minV^{out}(\mn,\measEmp)-z_{1-\alpha}n^{-1/2} \hat{ { \underline{\sigma}}}^{reg}_n,\infty\right),\\
\mbox{CI}^{\theta_1}_{\alpha,n,\Measures} & =\left[\minV^{out}(\mn,\kappa_n,\measEmp)-z_{1-\alpha}n^{-1/2}\hat{ { \underline{\sigma}}}^{reg}_n; \maxV^{out}(\mn,\kappa_n,\measEmp)+z_{1-\alpha}n^{-1/2}\hat{ { \bar{\sigma}}}^{reg}_n\right],\\
\mbox{CI}_{\alpha,n,\Measures}^{\mathcal{S}} & =\mbox{CI}^{\theta_1}_{\alpha/2,n,\Measures}.
\end{cases}\label{eq:CIsPoint}
\end{equation}
 Here, $z_{1-\alpha}$ is $1-\alpha$ quantiles of the standard Gaussian distribution and 
 \begin{align}
     \hat{ {\underline \sigma}}^{reg}_n &\bydef \max\{\funLStdHat[\mu_{n}][][],\sigma_0\},\\
     \hat{ { \bar\sigma}}^{reg}_n &\bydef \max\{\funUStdHat[\mu_{n}][][],\sigma_0\}, \label{eq:sigma0reg}
 \end{align}
for some small positive number $\sigma_0$.  
 $\mbox{CB}_{\alpha,n,\Measures}$ is a one-sided CB for $\funLValO$ (and   for $\theta_1$ as well), $\mbox{CI}^{\theta_1}_{\alpha,n,\Measures}$ is a two-sided CI that covers any $\theta_1$ in the  identified set, and $\mbox{CI}_{\alpha,n,\Measures}^{\mathcal{S}}$ is a two-sided CI (based on the Bonferroni inequality) that covers the entire maginal identified set $\setMarginal$. 
The tuning parameter $\sigma_0$ is introduced to guarantee the nominal coverage in the cases in which only nonstochastic constraints are binding at the optimal solutions corresponding to the support functions; otherwise, this degeneracy can lead to superconsistent support-function estimators with a non-Gaussian limiting distribution.
One can set $\sigma_0=0$ if this concern is not appropriate in a particular application.

As before, let $\Measures$ contain all measures $\measTrue$ that satisfy Assumptions~\ref{assu:Non--empty}--\ref{assu:Moments} with some uniform positive constants $\underline{\eta}$, $\underline{s}$, $\varepsilon$, and $\bar{M}$. 
 
\begin{thm}
\label{thm:TheoremUniform} 
Suppose that $0<\alpha<1/2$, $\mu_{n}\to0$, and $\mu_{n}\sqrt{n}\to\infty$.
Then the following results hold:
 \begin{eqnarray*}
\liminf_{n\to\infty}\inf_{\measTrue\in\Measures} \measTrue\left(\setMarginal\subset\mbox{CB}_{\alpha,n,\Measures} \right)=\liminf_{n\to\infty}\inf_{\measTrue\in\Measures}\inf_{ {\theta}\in\setID}\measTrue\left(\theta_{1}\in\mbox{CB}_{\alpha,n,\Measures} \right) & \geq & 1-\alpha, \\
\liminf_{n\to\infty}\inf_{\measTrue\in\Measures}\measTrue\left(\setMarginal\subset\mbox{CI}_{\alpha,n,\Measures}^{\mathcal{S}}\right)\geq1-\alpha,\text{ }\liminf_{n\to\infty}\inf_{\measTrue\in\Measures}\inf_{ {\theta}\in\setID}\measTrue\left(\theta_{1}\in\mbox{CI}_{\alpha,n,\Measures}^{\mathcal{S}}\right) & \geq & 1-\alpha.
 \end{eqnarray*}
 \end{thm}
\begin{proof}
See Appendix \ref{sec:Proof-of-Theorem2}.
\end{proof}

Note that the worst-case (that is, the smallest) asymptotic coverage probability of $\mbox{CB}_{\alpha,n,\Measures}$ is exactly equal to $1-\alpha$ for a fixed regular DGP such that   $\argminT(\measTrue)$ is a singleton and $\lim_{n\to\infty}\funLStd[\mu_{n}][\measTrue][2]>\sigma_0\geq 0  $.
For a nonregular DGP (or sequence of DGPs), the asymptotic coverage can only be larger than $1-\alpha$.
How conservative is $\mbox{CB}_{\alpha,n,\Measures}$ in the nonregular case? 
We can evaluate it by comparing its average length  with that of a confidence bound that has exact point-wise asymptotic coverage probability in a Monte Carlo study (see Section \ref{sec:Monte-Carlo}). 
Theorem \ref{thm:coveragePW} in Appendix Section \ref{subsec:pointwise-results} provides an alternative bias correction that remains non-conservative in the non-regular case and results in shorter point-wise confidence bounds. 
This approach is based on using an estimator of the inner bound $\minVin(\mu,\kappa,\measTrue)$.

I conclude this section with a brief discussion of the tuning parameters. 
The theory of optimal choice of tuning parameter is beyond the scope of this paper. The following considerations, however, can provide some guidance for the optimal choice. 
 Theorem~\ref{thm:bounds} suggests that  the tuning parameters should be smaller than   $  \bar{\mu}(\measTrue)$ to avoid the bias in the first-order asymptotic distribution. 
 This choice is infeasible since  $\bar{\mu}(\measTrue)$ is unknown, so one has to let $\kappa_n$ and  $\mn$ go to zero.
 The optimal rates of $\kappa_n$ and  $\mn$ should balance the higher-order variance and the worst-case bias. 
A specific choice of the tuning parameters is  discussed in Section \ref{sec:Monte-Carlo}.

\subsubsection{Joint confidence sets for general subvectors   \label{subsec:subvectors}}
It is trivial to extend the analysis to 
$$\underline{v}(\measTrue;a)=\displaystyle\min_{ {\theta}\in\setID} a^\prime {\theta}$$
for any $a\in R^d$ with $\ltwo{a}=1$. 
Indeed,  Assumptions~\ref{assu:Non--empty}--\ref{assu:Moments}
  are invariant with respect to orthogonal transformations  of the coordinates; that is, they are satisfied for the following program (with $\tilde{\theta}=U^\prime \theta$, $\tilde{\Ap} = \Ap U$, and $a^\prime= \eOne U $  for any orthogonal matrix $U$):
\begin{align}
\underline{v}(\measTrue;a) &=\min_{ {\theta}\in\R^{d} } \eOne   \tilde{\theta}  \\
\text{s.t. } & \begin{cases}
e_j^\prime\tilde{\Ap}  \tilde{\theta} =e_j^\prime\bp, & j\in\setEq,\\
e_j^\prime\tilde{\Ap}  \tilde{\theta} \leq e_j^\prime\bp, & j\in\setIneq.
\end{cases} 
\end{align}
We can think of $\tilde{\Ap}$ as a coefficient matrix under a different measure $\tilde{\measTrue}$, $A_{\tilde{\measTrue}}$. 
The set of measures $\Measures$ from Section~\ref{sec:DriftingDGP} includes $\tilde{\measTrue}$ corresponding to all orthogonal transformations of $\Ap$.

The identified set $\setID$ is convex, so any projection of it can be characterized using support functions for a corresponding direction. One can construct a joint CS for $\setID$ as follows. For any set of directions $\mathcal{A}\subset R^d$,  take
$$ \text{CS}^\mathcal{A}_{\alpha,n}  = \{\theta| a\in\mathcal{A} , a^\prime \theta \leq -\minV^{out} (\mn,\measEmp;-a)  + c_{1-\alpha}n^{-1/2} \max\{\underline{\sigma}(\mn,\measEmp;-a),\sigma_0\}\},$$
where $c_{1-\alpha}$ is $1-\alpha$ quantiles of the maximum of the corresponding asymptotic Gaussian variables (also known as $\sup t$ statistics) that can be estimated using multiplier bootstrap enabled by the asymptotic linear representation, 
\eqref{eq:Bahadur}.\footnote{ 
I leave full analysis of mulitplier bootstrap procedure in this setup for future work; see additional discussion in Appendix Section \ref{sec:multiplier boostrap}.}
One can also use the Bonferroni inequality-based standard Gaussian critical value $c_{1-\alpha}= z_{1-\alpha/|\mathcal{A}|}$.

Choosing the set of directions appropriately $\mathcal{A}$, we can construct joint CSs for projections of $\setID$ on any subvectors $\theta$.
If $\mathcal{A}$ has finitely many elements,  $ \text{CS}^\mathcal{A}_{\alpha,n}$ is a polygon. 
So we can plot it directly without performing test inversion as in the one-dimensional case. 
The confidence set $\mbox{CI}_{\alpha,n,\Measures}^{\mathcal{S}}$ is a particular case of $ \text{CS}^\mathcal{A}_{\alpha,n}$  corresponding to $\mathcal{A}=\{e_1,-e_1\}$ and the Bonferroni estimate of $c_{1-\alpha}$.  
It seems natural to construct a joint CS for  $\theta$  based on directions that correspond to the normal vectors of the moment conditions.
For simplicity, assume that $p=0$. 
The original system \eqref{eq:GenericMomentConditions} may have some inequalities that are slack for any point $\theta\in\setID$. 
We can characterize the identified set $\setID$ as the solution to a tight system of inequalities 
\begin{equation}
   e_j^\prime \Ap  \theta \leq \underline{b}_j,   j\in\setIneq, \label{eq:tightIDSet}
\end{equation}
where
\begin{align}
 \underline{b}_j &=  \max_{ {(\vartheta,\theta)}\in\R^{d+1} }\vartheta  \label{eq:tightBoundsProg} \\
\text{s.t. } & \begin{cases}
\vartheta &= e_j ^\prime \Ap  \theta,  \\
e_\ell^\prime{\Ap} {\theta} & \leq e_\ell^\prime\bp, \ell\in\setIneq.
\end{cases} 
\end{align}
Every inequality in system \eqref{eq:tightIDSet} is active at least at one point in $\setID$ (any point in the argmax of \eqref{eq:tightBoundsProg}). Programs \eqref{eq:tightBoundsProg} meet Assumptions~\ref{assu:Non--empty}--\ref{assu:Moments}, and therefore the outer estimators $ \hat{\underline{b}}_j^\text{out}$ are half-median unbiased with the corresponding standard error estimators $\hat{\sigma}_j$. Then the following polyhedron CS will cover any point $\theta\in\setID$ with asymptotic probability of at least $1-\alpha$ uniformly over $\measTrue\in\Measures$,
$$ \text{CS}^\mathcal{N}_{\alpha,n}  =\big \{\theta \big|
 e_j^\prime \hat{A}_P  \theta \leq \hat{\underline{b}}_j^\text{out} + z_{1-\frac{\alpha}{k}} n^{-1/2} \max\{\hat{\sigma}_j,\sigma_0\} ,   j\in\setIneq \big\},$$
 where $z_{1-\frac{\alpha}{k}}$ is the critical value of the standard normal r.v. and $k$ is the number of (in)equality restrictions.
 The generalization to the case $p\neq0$ is straightforward.

%%%%%%%%%%%%%%%%%%%%%%%%%%%%%%%%%%%%%%%

%\input{Sections/Extension}

%\input{Sections/Discussion.tex}

\section{Monte Carlo experiments}\label{sec:Monte-Carlo}

\subsection{Overview}

In this Monte Carlo study our goal is to evaluate how the confidence bound's length and the corresponding coverage probability depends on a number of key factors: (i) how close gradients of the relevant moment inequality to the non-regular case (i.e. being collinear with the vector $e_1$) ; (ii) how tuning parameter choice affects the conservativeness of confidence bounds; (iii) how the dimension  of the problem and the number of inequalities affect length, coverage and computational time for the proposed methods. 
The last exercise in the list also involves a comparison with the AS projection implemented using KMS EM computational algorithm.

\subsection{Proximity to non-regular case}

To illustrate the advantages of uniform coverage compared to point coverage, we can study a simple design with $k=4$ moment inequalities where the angle $\omega$ between the gradient of one of the faces of the identified set and vector $e_1$ varies continuously. 
The expectation of the moment inequalities can be parametrized as follows
\begin{align*}
  \expect W &= \expect (\Ap|\bp)=\left(\begin{array}{cc|c}
-\cos(\NEangle) & -\sin\left(\NEangle \right)& \cos(\NEangle)   + \sin\left(\NEangle \right)  \\
    \cos(\NEangle) &     \sin(\NEangle)&   \cos(\NEangle)   + \sin\left(\NEangle \right)  \\
    0 & -1 & 1 \\
    0 & 1  & 1
\end{array} \right) .
\end{align*}
The shape of this set is a parallelogram analogous to the one in Figure~\ref{fig:Identified-sets-Example2} in with $\rho\geq0$.
The parameter $\omega\in[0^{\circ},36^{\circ}]$ defines the angle between the normal vectors of the rear sides of the parallelogram and the horizontal axis.
The value $\omega=0^{\circ}$ corresponds to a square-shaped identified set, also referred to as a \emph{non-regular case} because the sides are orthogonal to the gradient of the objective function. 
All other values are termed \emph{regular}.
In the vicinity of $\omega=0^{\circ}$   pointwise valid  $\mbox{CB}_{\alpha,n}$ may have a coverage probability below the nominal level because it lacks uniform validity.
The expectation $\expect W$ is parameterized to guarantee $\argminT_1=\argminT_2=-1$ and $\argmaxT_1=\argmaxT_2=1$ for all values of $\omega$. 
The components of $W_i$ are independent Gaussian random variables with variance $s_2^{2}=0.01$.
For each value of $\omega$, I compute the frequency of coverage and the excess average length over the identified set for $\mbox{CB}_{\alpha,n}$ and $\mbox{CB}_{\alpha,n,\Measures}$ based on the sample sizes $n\in\{100,10000\}$. 
The number of MC simulations is $1000$ for every combination of $n$ and $\omega$. 
The focus is on the nominal coverage probability $\alpha=0.95$.

As the main choice of tuning parameters, I use $\mu_n = \hat\mu_0 \sqrt{n^{-1}\ln\ln n}$ and  $\kappa_n = \hat\mu_0 \sqrt{n^{-1}\ln n}$ , where $$\hat\mu_0 = \sqrt{\frac{1}{n}\sum_{i=1}^n \big(  \argminL^\prime(0,\measEmp)\data_i e_1 - \frac{1}{n}\sum_{i=1}^n \argminL^\prime(0,\measEmp)\data_i e_1  \big)^2
}.$$  
While a theory of optimal choice of $\hat\mu_0$ is beyond the scope of this paper, this particular choice has some advantages: (i) $\hat\mu_0$  depends only on the behavior of the relevant moment inequalities as selected by non-zero components of $ \argminL^\prime(0,\measEmp)$;  (ii) it does not depend on the dimension of the   parameters $\theta$, since only the first column of $w_i$ is involved; (iii) it has the same scale as the standard deviation of the relevant components of $w_i$ which justifies our perturbation analysis (the impact of the regularization term $\mu_n \ltwo{\theta}^2$ is  larger  than the sample variation, $\sim \hat\mu_0 n^{-1/2}$).
As a result, this choice resulted in a good alignment of the theoretical predictions with the simulations, as can be seen below.

To appreciate the importance of uniformly valid confidence bounds, we start our study with point-wise valid confidence bounds $\mbox{CB}_{\alpha,n}$. 
Figure \ref{fig:pwFreq_opt1} panel (a) shows the corresponding coverage frequency for a small sample size $n=100$ and a large sample with $n=10,000$. 
For values of $\omega$ that are far from $0$, both sample sizes have observed a frequency of covering the correct bound of the identified set for $\theta_1$   close to the nominal coverage probability of $\alpha=0.95$. 
The same is true for the nonregular case with $\omega = 0$.
However, in the vicinity of the nonregular case $\omega \in  (0,4^\circ]$, large sample sizes are necessary to achieve a satisfactory coverage frequency. 
In contrast, the coverage frequency for the corresponding sample sizes for the uniformly valid confidence bounds $\mbox{CB}_{\alpha,n,\Measures}$ given on Figure \ref{fig:pwFreq_opt1} panel (b) is uniformly at least as large as   $\alpha$ for all values of $\omega$. 
One can see that the only point on Figure \ref{fig:pwFreq_opt1} panel (b) with coverage is significantly higher than $\alpha=0.95$ for $n=10,000$ is $\omega = 0$, that is, for most designs (or equivalently, directions of the support function), the uniformly valid confidence bounds $\mbox{CB}_{\alpha,n}$ have exact coverage.

The difference in coverage frequencies of $\mbox{CB}_{\alpha,n}$ and $\mbox{CB}_{\alpha,n,\Measures}$ can be understood by considering the behavior of the corresponding average excess length of the bounds, i.e. the Monte Carlo average of the difference between the corresponding confidence bounds and the true bound of the identified set $\argminT_1=-1$.
Figure \ref{fig:pwFreq_opt1} panel (c) compares the lengths of the two bounds for the small sample case $n=100$ where the point-wise inference becomes unreliable. 
One can see that the confidence bounds should have an excess length approximately equal to $0.025$ to match the coverage frequency with the nominal probability $\alpha$ in the proximity of $\omega =0$. 
The length of $\mbox{CB}_{\alpha,n,\Measures}$ is larger than necessary for $\omega=0$, while $\mbox{CB}_{\alpha,n}$ has the  correct length for $\omega=0$, but is too short in the neighborhood $ (0,4^\circ]$.
For values $\omega\geq 5$ both confidence bounds have approximately the same length, which also corresponds to the correct coverage probability predicted by Theorems  \ref{thm:bounds} and Appendix Theorem~\ref{thm:coveragePW}.

\subsection{Sensitivity to choice of  $\mu_n$ and $\kappa_n$} \label{subsec:sensitivityMC}

First, consider the effect of the tuning parameter  $\mu_n$ on the coverage probability for the uniformly valid confidence bounds $\mbox{CB}_{\alpha,n,\Measures}$.  
Figure \ref{fig:lengthCompare_opt1} panel (a) compares the coverage frequency of $\mbox{CB}_{\alpha,n,\Measures}$ for the sample size $n=100$ as a function of $\omega$ for two choices: (i) $\hat\mu_0 \sqrt{n^{-1}\ln\ln n}$ (baseline) and (ii) $\hat\mu_0 \sqrt{n^{-1}\ln n} $ (large). The baseline option has slightly less conservative coverage in the neighborhood of the nonregular design ($\omega=0$) and similar performance for other values of $\omega$. 
As a result, the baseline option is used for all the other simulations. 

Unlike its uniform counterpart, the point-wise valid confidence interval $\mbox{CB}_{\alpha,n}$  depends on additional tuning parameter sequence $\kappa_n$.
Theorem  \ref{thm:coveragePW} claims that for one-sided bounds $\mbox{CB}_{\alpha,n}$, the asymptotic coverage is exactly equal to $\alpha$ as long as $\kappa_n$ shrinks to zero slower than $\mu_n$.
This allows values of $\kappa_n$ to be smaller than $\mu_n$ and still result in (point-wise) valid inference.
Figure \ref{fig:lengthCompare_opt1} panel (b) compares three choices: (i) $\kappa_n= \hat\mu_0 \sqrt{n^{-1}\ln n} > \mu_n$; (ii) $ \kappa_n=\mu_n$; (iii)   $ \kappa_n=0<\mu_n$.
From the practitioner's perspective, options (ii) and (iii) are attractive since they only require the specification of one tuning parameter $\mu_n$.
Both additional choices (ii) and (iii) result in valid point-wise coverage in large samples (see Appendix Figures  
\ref{fig:pwFreq_opt2} and \ref{fig:pwFreq_opt3}).

For $n=100$, $\omega>4^\circ$ all three options have nearly indistinguishable coverage frequency.
The behavior is notably different between the three for $ \omega \in [0,4^\circ]$.
Choice (ii) still results in a lower probability of coverage than the required probability, but the problematic neighborhood $[0,2^\circ]$ is smaller than for choice (i).
 The choice $\mu_n=\kappa_n$ for $n=10,000$ is particularly in good alignment with the nominal coverage $\alpha=0.95$.
The choice (iii) essentially results in $\mbox{CB}_{\alpha,n}$ being the same length as $\mbox{CB}_{\alpha,n,\Measures}$ (see Appendix Figure  \ref{fig:lengthCompare_opt3}).
The similar performance in simulation suggests that  $\mbox{CB}_{\alpha,n}$  with $\kappa_n=0$ may have uniformly valid coverage like $\mbox{CB}_{\alpha,n,\Measures}$.
However, a formal study of uniform validity for $\mbox{CB}_{\alpha,n}$ with $\kappa_n=0$ is more difficult to conduct than for $\mbox{CB}_{\alpha,n,\Measures}$.

\begin{figure}[H]

\begin{centering}
\caption{\label{fig:pwFreq_opt1} Coverage frequency (a,b) and   average excess length (c)  in the $2$-dimensional  design for   $\mbox{CB}_{\alpha,n}$  and $\mbox{CB}_{\alpha,n,\Measures}$ as function of $\omega$ in the $2$-dimensional  design.}
 \begin{tabular}{cc}
\includegraphics[scale=0.9]{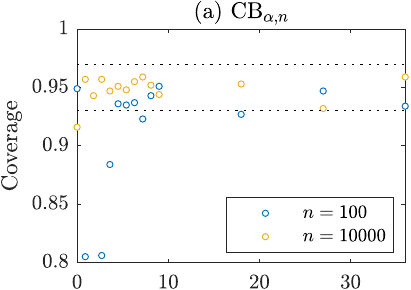}& \includegraphics[scale=0.9]{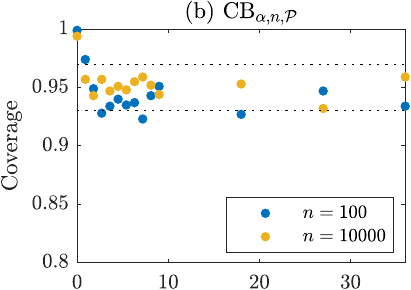} \tabularnewline
\end{tabular}
 \begin{tabular}{c}
 \includegraphics[scale=0.9]{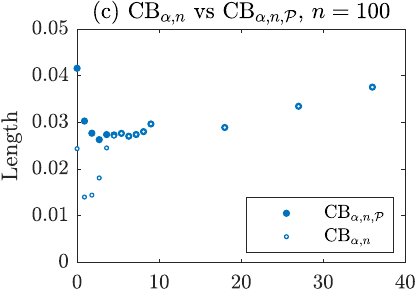}\tabularnewline
\end{tabular}
\par\end{centering}

   Note: The dotted lines correspond  to the asymptotic uniform 95\% confidence interval   for the parameter $p=0.95$ of the Bernoulli random variable based on a random sample of $1000$ simulations based on the Bonferroni correction for 14 hypothesis tests.
  Values of $\omega$  close to zero in panel (a) result in nonnegligible under-coverage. As the sample size grows, the problematic area shrinks.
   Values of $\omega$ close to zero in panel (b) result in nonnegligible conservative coverage.
\end{figure}

\begin{figure}[H]

\begin{centering}
\caption{\label{fig:lengthCompare_opt1} Sensitivity of coverage frequency to tuning parameter choices.}
  \begin{tabular}{cc}
\includegraphics[scale=0.9]{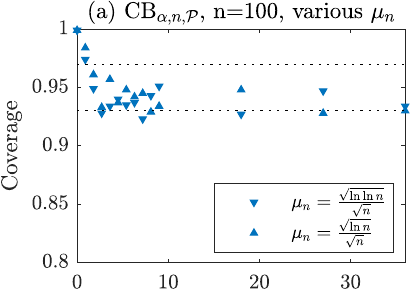}& \includegraphics[scale=0.9]{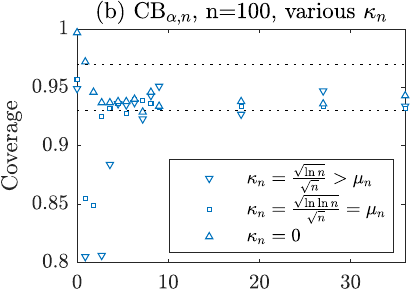}\tabularnewline
\end{tabular}

\par\end{centering}
   Note: The dotted lines correspond  to the asymptotic uniform 95\% confidence interval   for the parameter $p=0.95$ of Bernoulli random variable based on a random sample of $1000$ simulations based on Bonferroni correction for 14 hypothesis tests.
\end{figure}

\subsection{Effect of high dimensions}\label{subsec:highDMC}

In this section, we study the effect of the dimension of $\theta$ on the coverage probability and the length of the confidence bounds.
The design of the constraints matrix $\expect W=\expect (\Ap|\bp)$ is given by 
\begin{align*}
\Ap=\left(\begin{array}{c c}
1 & -a\\
0 & -\I[d-1]\\
0 & \I[d-1]
\end{array}\right) \text{ and } \bp=\left(\begin{array}{  c}
 0\\
   \iota_{2}\\
  \iota_{d-3}/\sqrt{d-3}\\
  \iota_{2}\\
  \iota_{d-3}/\sqrt{d-3}
\end{array}\right).
\end{align*}
The first line corresponds to an equality constraint that depends on a unit vector with direction $a\in\R^{d-1}$.
This equality constraint is deterministic, that is, $\operatorname{Var}W_{1j}=0$ for all $j=1,\dots,d+1$. Other constraints are inequality constraints with coefficients that are i.i.d Gaussian r.v. with $\operatorname{Var}W_{ij}=0.01$ for all $j=1,\dots,d+1$ and $i=2,\dots,(1+2d)$.
The identified set for the first coordinate $\theta_1$ in this design corresponds to values of a support function of a $d-1$ dimensional rectangular box $[-1,1]^2\times[-1/\sqrt{d-3},1/\sqrt{d-3}]^{d-3}$ in direction $a$ (and $-a$).
As before, the focus is on the lower bound, since by design the upper bound is symmetric.
When $d=3$ and $a=(1,0)^\prime$, this design reduces to the two-dimensional design considered in the previous subsection with $\omega=0$.

We consider two possible values for $a$: (a) regular case, $a=\iota_{d-1}/\sqrt{d-1}$; (b) nonregular case, $a=(1,0,\dots,0)'\in\R^{d-1}$.
In the regular case $\argminT = (-\iota_2,-\iota_{d-3}/\sqrt{d-3})$. In the nonregular case, the argmin set $\argminT$ is a convex hull of $2^{d-1}$ corner points with coordinates $(-1,\pm1,\pm1/\sqrt{d-3})$. 
Every such corner point has the same distance to $0$, equal to $\sqrt{3}$,  for $d>2$ regardless of the dimension $d$.
This normalization was chosen to make performance comparisons as dimension grows while keeping the diameter of the identified set fixed. In this way, we isolate the effect of the number of dimensions and constraint from the effect of the diameter of the identified on the dual variables.
By design, the asymptotic variance of the regularized estimators will not change with the dimension.
We used $n=1000$ data points and $N=1000$ Monte Carlo simulations. We use the first option for the tuning sequences, namely $\mu_n = \hat\mu_0 \sqrt{n^{-1}\ln\ln n}$ and $\kappa_n = \hat\mu_0 \sqrt{n^{-1}\ln n}$.

First, consider the coverage probability of the confidence bounds $\mbox{CB}_{\alpha,n}$ and $\mbox{CB}_{\alpha,n,\Measures}$. Panel (a) in Figure \ref{fig:ndFreq} shows that in the regular case both the uniform and the point-wise valid confidence set have a coverage frequency that is within the simulation error bound from the nominal coverage probability of $\alpha=0.95$.
At the same time, the coverage of the two procedures becomes noticeably different from each other for $d>40$.
It suggests that the size of the neighborhood where the uniform inference becomes important gets wider with $d$. 
The nonregular design provided on panel (b) in Figure \ref{fig:ndFreq} shows that uniformly valid confidence bounds have a coverage frequency nearly equal to $100\%$.
The point-wise bounds $\mbox{CB}_{\alpha,n}$  are less conservative for all $d$, but still reach the $100\%$ coverage frequency for $d\geq21$.
It suggests that the sample size required for exact asymptotic coverage for $\mbox{CB}_{\alpha,n}$ gets larger as dimension grows. 

The probability of coverage above the nominal level $\alpha=0.95$ in the nonregular case with high dimension  $d$ reflects the fact that as the number of corner points increases exponentially with $d$.
As a result, the nonregularized support function gets increasingly biased as it has asymptotic distribution of a minimum of $2^{d-1}$ Gaussian r.v. as evident from \eqref{eq:asyDistributionShapiro}.
Fortunately, this extreme conservatism only appears in the worst possible case. 
For a randomly chosen direction of a support function $a\in\R^{d-1}$, the performance is expected to be closer to that of panel (a) of Figure \ref{fig:ndFreq}, at least for sufficiently large sample sizes $n$.

As a benchmark, I use  $\mbox{CB}_{\alpha,n,AS}$---one-sided confidence bounds for $\argminT_1$  based on \cite{andrews2010inference} with Bonferroni critical values implemented using the fast E-A-M algorithm of \cite{kaido2015inference}.
This choice of benchmark is one of the fastest available uniformly valid procedure in the literature. The two alternative approaches, \cite{bugni2014inference} and \cite{kaido2015inference}, can provide uniformly valid CSs with a potentially shorter average length than $\mbox{CB}_{\alpha,n,AS}$.
However, both are expected to take considerably more time to compute because they add time-consuming   profiling or calibration steps.
AS results are only available in the nonregular case because of the software restrictions.\footnote{The code is available at https://molinari.economics.cornell.edu/programs.html.
This code is highly optimized for the case of inference on bounds on individual coordinates of the identified set. The code does not allow making confidence sets for arbitrary directions $a$.}
Nevertheless, AS is not adaptive and is expected to have a similar length regardless of whether we compute the bounds in regular or nonregular directions $a$. 
The coverage probability equal to the confidence sets to $100\%$  also applies to $\mbox{CB}_{\alpha,n,AS}$.
Figure \ref{fig:ndLen} shows that for $d\geq9$ for $n=1000$, the uniformly valid $\mbox{CB}_{\alpha,n,\Measures}$  is less conservative than $\mbox{CB}_{\alpha,n,AS}$.

\subsection{Computational time}

The main advantage of $\mbox{CB}_{\alpha,n,\Measures}$ compared to $\mbox{CB}_{\alpha,n,AS}$ is the computational gain that is achieved because of two factors: (i) using only linear and convex quadratic programs; (ii) no need to use multi-start procedures for global optimization of non-convex programs. 
Other alternative procedures like KMS and BCS would have additional burden of computing simulation-based critical values and are expected to be much slower.
Note that the computational cost does not depend on whether the design is regular or not, so Figure \ref{fig:time} compares the average computation time as a function of dimension $d$ on a modern multi-core laptop.\footnote{All simulations are using 2023 Macbook Pro with M2 Max CPU (12 cores) and 96 GB of RAM. The EAM algorithm for $\mbox{CB}_{\alpha,n,AS}$ takes advantage of all 12 cores. }
The computational time for the delta-method for regularized support functions grows very slowly with dimension $d$; the average time for $d=100$ is about 2 seconds.
In contrast, $\mbox{CB}_{\alpha,n,AS}$ is already nearly 1000 times slower for $d=21$.

\begin{figure}[H]
\caption{\label{fig:ndFreq}Coverage frequency in the $d$-dimensional  design as function of $d$ in (a)  regular   and  (b) non-regular cases.}
\begin{centering}
\begin{tabular}{cc}
\includegraphics[scale=0.9]{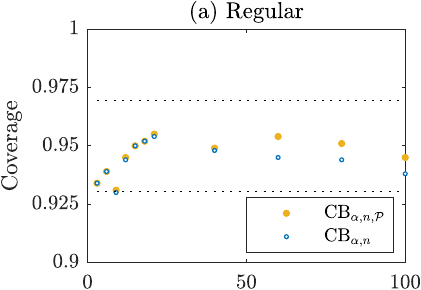}& \includegraphics[scale=0.9]{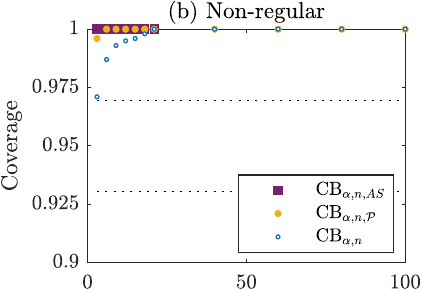}\tabularnewline
\end{tabular}
\par\end{centering}
   Note: The dotted lines correspond  to the asymptotic uniform 95\% confidence interval   for the parameter $p=0.95$ of Bernoulli random variable based on a random sample of $1000$ simulations based on Bonferroni correction for 11 hypothesis tests.
   AS coverage frequency is computed using $100$ simulations and for $d\leq21$, while the delta-method confidence bounds are based on $1000$ simulations. In all cases each simulation is based on $n=1000$ observations.
 
\end{figure}
\vspace{-0.5cm}
\begin{figure}[H]
\caption{\label{fig:ndLen}Length in the $d$-dimensional  design as function of $d$ in (a) the regular   and  (b) non-regular cases.}
\begin{centering}
\begin{tabular}{cc}
\includegraphics[scale=0.9]{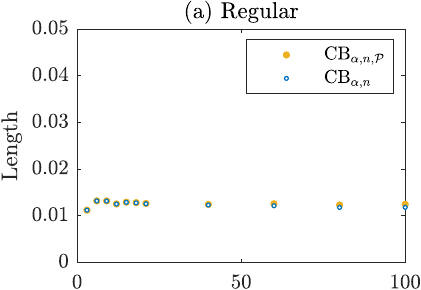}& \includegraphics[scale=0.9]{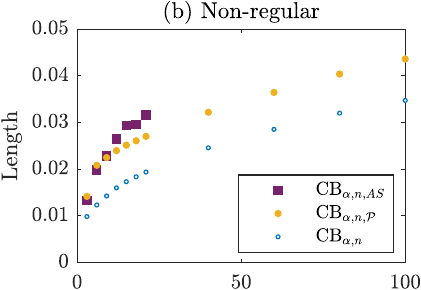}\tabularnewline
\end{tabular}
\par\end{centering}
   Note: AS length is computed using $100$ simulations and for $d\leq21$, while the delta-method average length are based on $1000$ simulations. In all cases each simulation is based on $n=1000$ observations.
\end{figure}

\vspace{-0.5cm}
\begin{figure}[H]

\begin{centering}
\caption{\label{fig:time}Computational time in the $d$-dimensional  design as function of $d$.}

\includegraphics[scale=0.8]{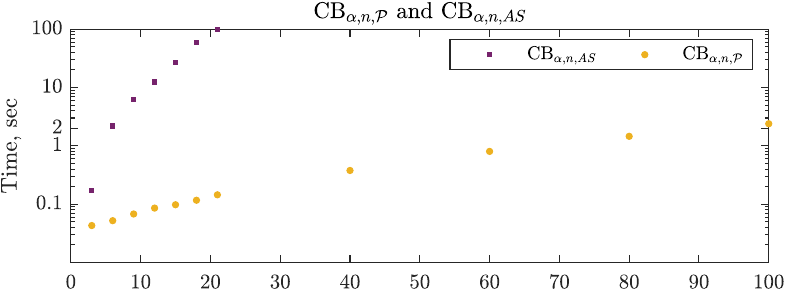} 
 
\par\end{centering}
     Note: The scale is logarithmic. AS average time is computed using $100$ simulations and for $d\leq21$ only, while the delta-method average times are based on $1000$ simulations. In all cases, each simulation is based on $n=1000$ observations.
\end{figure}

\section{Conclusion}\label{sec:Conclusion}

This paper demonstrated that the regularization approach provides a fast way to construct point-wise and uniform CSs for a $\theta_{1}$
that is comparable to or shorter than those of the existing literature. 
Monte Carlo simulations showed that the proposed CSs have good finite sample coverage properties.
The computational benefits of the new approach are particularly prominent if the dimension of $\theta$ is large. 
The regularization framework can be extended in a number of ways to allow for overidentification and joint inference. 
The proposed approach is attractive in applications such as a linear model with an interval-valued outcome variable and a large number of regressors and in problems with parameters represented as intersection bounds.

\label{reply:R1.2a} Focus on the affine inequalities    simplifies the large sample analysis. However, there are many interesting applications that characterize the identified set for the structural parameters using non-linear moment conditions, in particular, among the structural models of the industrial organization \citep{pakes2015moment}. Analysis of such non-linear moment inequalities goes beyond the scope of the current study, but would be a promising direction of future research.

\bibliographystyle{ecta}
%\bibliography{Mendeley}
\bibliography{main}

\newpage{}
\appendix
\setcounter{page}{1}  
 {\small

%\input{Appendix/Outline}
%\section{Figures} \label{subsec:figures}
%\input{Appendix/Figures.tex}
%\section{MC results}
%\input{Appendix/tables.tex}

 \section{Detailed proofs}
\newpageLemma
\subsection{Relation between Assumption \ref{ass:ULICQ} and LICQ} \label{subsec:LICQdiscussion}

As before, I use symbols $\activeSet\tP$, $\setIndexActive$ etc to denote the projectors on the coordinates with the corresponding indices. 
Let $\jdMat\bydef \gradMat$.  
Using this notation, LICQ implies that matrix $\activeMat\tP \Ap$
has full row rank for any $ {\theta}\in\setID$.

\begin{lem}[Sufficient condition for LICQ]
\label{lem:LICQ} Assumption   \ref{ass:ULICQ}  implies LICQ.

\end{lem}
\begin{prooF}{lem:LICQ}
Assumption \reF{assu:NoOveridentification} implies that $\activeSet\tP$
has at most $d$ elements at any $ {\theta}\in\setID$. Consider any point $ {\theta}\in\setID$.  The set $\setIneq$ includes (\reF{eq:Box}), so
$k\geq2d\geq d+1$.  It implies that there exists a set $\jSet$ with $\card{\jSet}=d$ such that $\activeSet\tP\subset\jSet$.   By Assumption \reF{assu:RankCondition}, $\rank\left[\jMat\expect\Data \right]=d$,
so $ M\bydef\activeMat\tP(\Ap,\bp)$ has full row rank which is equal to $\card{\activeSet}$. By definition, $\activeMat\tP\Ap\theta  =\activeMat\bp$. It implies 
by the Rouch\'{e}\textendash Capelli
theorem that the matrices $M_\theta\bydef\activeMat\tP \Ap$
and $M$ have the same rank. This result implies LICQ.
\end{prooF}
The inverse implication does not hold in general as the following remark shows.
\begin{rem} \label{rem:LICQimplications}
LICQ implies Assumption \reF{assu:NoOveridentification}
and that for any $ {\theta}\in\setID$
\begin{equation}
\rank\left[\activeMat\tP  \expect\Data \right]=\card{\activeSet\tP}  .
\label{eq:strongRankConstraint}
\end{equation}

Indeed,  suppose that LICQ holds. It immediately implies Assumption \reF{assu:NoOveridentification}. To see   \eqreF{eq:strongRankConstraint} consider any point $\theta\in\setID$ such that $\card{\activeSet\tP}\leq d$. By the Rouch\'{e}\textendash Capelli theorem and the full row rank property of $M_\theta$ correspondingly, 
\[
\rank\left(M\right)=\rank\left(M_\theta\right)=\card{\activeSet\tP}.
\]

\end{rem}

 \begin{rem}[Interpretation as LICQ imposed over the entire space]\label{rem:contastLICQ2}
Assumption   \ref{assu:NoOveridentification} is formulated in form of restriction on a norm of a slack vector, $\|\activeMat\Ap\theta-\activeMat\bp\|$ of any $d+1$ constraints at every point $\theta\in\setID$. 
If we require instead no overidentifying inequalities for all $\theta\in \R^d$, one can write this requirement  in a form similar to Assumption \ref{assu:RankCondition},
\begin{equation}
\min_{\text{s.t. }\begin{matrix}
\jSet  \subset\setIneq \\
\card\jSet  =d-p+1 
\end{matrix}}\eta_1\left(\activeMat (\Ap|\bp) \right) >0. \label{eq:alternativeAss5}
\end{equation}
The bound on the singular value \eqref{eq:alternativeAss5} implies that $rk (\activeMat  (\Ap|\bp))= d+1$.
Since for any combination of $d+1$ constraints $rk (\activeMat \Ap)\leq d$, by the aforementioned  Rouch\'e-Capelli theorem any combination of $d+1$ linear inequalities/equalities cannot be satisfied at once, that is, $\|\activeMat\Ap\theta-\activeMat\bp\|\geq0$. 
The converse is also true by the same theorem.

One can notice that equation \eqref{eq:alternativeAss5} (bound on singular values and rank of $(d+1)\times(d+1)$ matrices) implies equation \eqref{eq:rankD} (bound on singular values and rank of its $ d  \times(d+1)$ submatrices) for some particular choice of bounds on the minimal singular value $\eta (\measTrue)$. 
As a result, \eqref{eq:alternativeAss5} implies both Assumptions \ref{assu:RankCondition} and Assumption \ref{assu:NoOveridentification} as well as the LICQ for all $\theta\in \R^n$ (see Lemma \ref{lem:LICQ}).
By the argument in Remark \ref{rem:LICQimplications},  LICQ for all $\theta\in \R^n$ would in turn imply \eqref{eq:alternativeAss5}. 
So \eqref{eq:alternativeAss5} is a necessary and sufficient condition for LICQ to hold for all $\theta\in \R^n$ (not just on the identified set).
So,  Assumptions \ref{assu:RankCondition}-\ref{assu:NoOveridentification} are satisfied if there are no overidentifying inequality and equality constraints for all $\theta\in \R^n$.
 \end{rem}

\subsection{Topological properties of optimal solutions} \label{subsec:LICQlemmas}
Consider any distribution $P$ with support on $\R^{(k-2d)\times\left(d+1\right)}$
such that $(\Ap,\bp)\bydef\expect\Data$ exist. Let $\activeSet\tP\subset\{1,...,k\}$
be the set of indices of moment equality and inequality constraints
active at $\theta$, i.e.~all $j$ s.t. $\Moment{\theta}{\measTrue}\bydef\expect\moment{\Data}{\theta}=0.$  $\activeSet\tP$ can be empty.

\begin{lem}[Characterization of the optimal solution]
\label{lem:KKT}Under Assumption \reF{assu:Non--empty} for
any $\mu\geq0$ any minimizer $ {\theta}$
for Program \eqreF{prog:regProg}  is a solution to the corresponding
Karush\textendash Kuhn\textendash Tucker (KKT) optimality conditions
for some finite $ {\lambda}\in\R^{k}$, 
\begin{numcases}{}\label{eq:KKTsystem}
(\eOneVec+2\mu {\theta} )^\prime  =- {\lambda}^{\prime}\Ap,&\label{eq:KKT-FOC}\\ 
\Moment{\theta}{\measTrue}  =0 & $j\in\setEq$,\label{eq:KKT-EQ}\\
\Moment{\theta}{\measTrue} \leq0,\lambda_{j}   \geq0,\lambda_{j} \Moment{\theta}{\measTrue}    =0 & $j\in\setIneq$. \label{eq:KKT-Comlementarity}
\end{numcases}
 
\end{lem}

\begin{prooF}{lem:KKT}
  By Assumption \reF{assu:Non--empty}, $\setID\subset\Theta$
is non\textendash empty and closed, so the global optima for Program
(\reF{prog:regProg}) exist. Program (\reF{prog:regProg}) is convex
for any $\mu\geq0$, i.e.~the objective function is convex, the constraints
are affine. Assumption \reF{assu:Non--empty} implies Slater's condition.
Since the Program (\reF{prog:regProg}) is convex, any global optimum
$\funLArg$ of Program \reF{prog:regProg} satisfies \eqreF{eq:KKT-FOC}-\eqreF{eq:KKT-Comlementarity}
for some finite vector of Lagrange multipliers $ {\lambda}$
(maybe non\textendash unique) (see p.244 in \citet{boyd2004convex}).
\end{prooF}

If we introduce the notation $\funLLagrangian\bydef\theta_{1}+\mu\ltwo{ {\theta}}^{2}+ {\lambda}^{\prime}\MomentVec{\theta}{\measTrue},$ \eqreF{eq:KKT-FOC} becomes
$$\grad[ {\theta}]\funLLagrangian=0$$

Let $\argminTL\mP\bydef\left(\argminT\mP,\argminL\mP\right)$ be a set of solutions to \eqreF{eq:KKT-FOC}-\eqreF{eq:KKT-Comlementarity}. In order to have a unique solution $\lambda$  Program \eqreF{prog:regProg} need to meet a stronger constraint qualification condition LICQ defined earlier in Section~\ref{subsec:Regularization}.

\newpageLemma

\begin{lem}[Uniqueness of the optimal solutions]
\label{lem:UniqueMin} Suppose that both Assumption  \ref{assu:Non--empty}  and LICQ are satisfied.
Then for any $\mu\geq0$ the set of multipliers $\funLLagrange$ is a singleton. Moreover if $\mu>0$, then $\argminTL\mP$ is a singleton.
\end{lem}

\begin{prooF}{lem:UniqueMin}
By definition of $\activeMat\tP$, any $\theta$ and $\lambda$ satisfying \eqreF{eq:KKTsystem} satisfy   
\begin{equation}
\lambda^\prime(\activeMat\tP)^\prime\activeMat\tP = \lambda^\prime. 
\label{eq:non-zero-lambda}
\end{equation} 
So \eqreF{eq:KKT-FOC}  becomes \begin{equation}
\left(\eOneVec+2\mu {\theta}\right)^{\prime}   =- \gamma^\prime\activeMat\tP\Ap, \label{eq:FOC-gamma}
\end{equation}  where $\gamma^\prime\bydef\lambda^\prime(\activeMat\tP)^\prime\in\R^{\card{\activeSet\tP}}$. 
By LICQ, for any $ {\theta}\in\setID$ the matrix $A\bydef\activeMat\tP \Ap$ has full rank. Hence for any $\theta$ there can be at most one 
 $\gamma^{*}\in \R^{\card{\activeSet\tP}}$
 satisfying \eqreF{eq:FOC-gamma}. If $\eOneVec+2\mu \theta   = 0$ , then trivially $\lambda$ is a zero vector. Otherwise it is given by 
 \begin{equation}\gamma^{*} = - (AA^\prime)^{-1}A^\prime  \left(\eOneVec+2\mu {\theta}\right).  \label{eq:lambda}
\end{equation} Then $(\funLLagrange)^\prime \bydef (\gamma^{*})^\prime\activeMat\tP $ is the unique solution to \eqreF{eq:KKT-FOC}-\eqreF{eq:KKT-Comlementarity} for any solution $\theta$.
 
Now consider the case $\mu>0$.  The second order derivative matrix of $\funLLagrangian$ with respect to $ {\theta}$
at any solution $\argminTL\mP$
is $2\mu\I[d].$ It is positive definite for any $\mu>0$, so the
Second Order Sufficient Condition (SOSC) is satisfied at any point.
By Theorem 3.63 from \citet{bonnans2013perturbation} the second order
growth condition holds at $\funLArg$, i.e. $\exists\varepsilon>0$
and $c>0$ s.t. for $\forall {\theta}\in\Theta\left(\measTrue\right)$
s.t. $\ltwo{ {\theta}-\funLArg}<\varepsilon$ the following
inequality holds
\[
\theta_{1}+\mu\ltwo{ {\theta}}^{2}\geq\eOne\funLArg+\mu\ltwo{\funLArg}^{2}+c\ltwo{ {\theta}-\funLArg}^{2}.
\]
So the value of the objective function at $\funLArg$ is strictly
smaller than the value at any other point in a neighborhood of $\funLArg$.
Since for the convex program the set of global optima is convex and connected, it implies
that $\funLArg$ is the unique global minimizer. 
\end{prooF}

\newpageLemma

 {
\newcommand{\setOfAllActiveSets}{\{\activeSet\tP\}_{\theta\in\setID}}

It is assumed (see eq. \eqref{eq:Box}) that the system of inequalities includes deterministic constraints 
\begin{equation}
-\infty<-\underline{c}_{\ell}\leq\theta_{\ell}\leq \bar{c}_{\ell}<\infty \quad \text{for } \ell=1,...,d.
\end{equation}
These constraints define a compact set $\Theta$. 
By construction, $\Theta (P)\subset \Theta$ for all measures $P$.

\begin{lem} \label{lem:boundOnLambda}
Suppose that Assumptions~\ref{assu:Non--empty} and \ref{ass:ULICQ}  are satisfied and $\mu\leq 1/2$. Then
\begin{equation}
\ltwo{\funLLagrange}^2  \leq \CL^2 \bydef \frac{\CT^{3}}{\eta^{2}}<\infty,\label{eq:lagrangeNorm} 
\end{equation}
where $ \CT  \bydef (1 +\max_{ {\theta}\in\Theta}\ltwo{ \theta})$. 

\end{lem}

\begin{prooF}{lem:boundOnLambda}
Consider any point $\theta\in\funLArg$ and the corresponding $\activeSet\tP$. Let $A\bydef\activeMat\tP \Ap$ and $b\bydef\activeMat\tP \bp$. 
Let $\eta_A^2\bydef \eig(AA^\prime)$ so that equation~\eqreF{eq:lambda} implies  
\begin{equation}
\ltwo{\funLLagrange}  \leq \eta_A^{-1} \ltwo{\eOneVec+2\mu {\theta}} .\label{eq:immediateBoundOnLambda}
\end{equation}
By the variational property of eigenvalues, 
\begin{equation}
\eta_A^2 =\min_{v\in\R^\ell} \frac{v^\prime AA^\prime v}{v^\prime v}. \label{eq:etaA}
\end{equation} 
By Assumption \reF{assu:RankCondition} 
$$\eta^2\leq \eig((A|b)(A|b)^\prime)\bydef\min_{v\in\R^\ell} \frac{v^\prime(AA^\prime +bb^\prime ) v}{v^\prime v}. $$
Let $v_A$ be any minimizer of the r.h.s. of \eqreF{eq:etaA} such that ${v_A^\prime v_A}=1$ . Then
\begin{equation}
\eta^2\leq v_A^\prime(AA^\prime +bb^\prime ) v_A = (A^\prime v_A )^\prime(I_d +\theta \theta^\prime ) (A^\prime v_A) 
\end{equation} 
where the last equality holds since by definition $b = A\theta.$
Finally,  
\begin{equation}
\frac{\eta^2}{\eta_A^2}\leq \frac{(A^\prime v_A )^\prime(I_d +\theta \theta^\prime ) (A^\prime v_A) }{(A^\prime v_A )^\prime(A^\prime v_A) }\leq \ltwo{I_d +\theta \theta^\prime }\leq \CT. \label{eq:etaAbound}
\end{equation} 
Result \eqreF{eq:lagrangeNorm} then follows from \eqreF{eq:immediateBoundOnLambda} and \eqreF{eq:etaAbound} for any $\mu\leq 1/2$.
\end{prooF}
}

\begin{rem} \label{rem:inverse Norm of active constraints}
Equation~\eqreF{eq:etaAbound} provides bound for, $A$, a matrix with gradients of active moment conditions at any point $\theta \in \Theta$,
\begin{equation}
\ltwo{\left(A A^\prime\right)^{-1}} \leq\CT\eta^{-2}.
\end{equation}

\end{rem} 

\newpageLemma

The function $\phi\left(a,b\right)\bydef\sqrt{a^{2}+b^{2}}+a-b$ , considered in \citet{fischer1992special}, has the following property.

\begin{prop}\label{prop:FBproperty}
\begin{equation}
\phi\left(a,b\right)=0\mbox{ if and only if }a\leq0,b\geq0,ab=0.\label{eq:FBproperty}
\end{equation}
\end{prop}
It can be used to replace \eqreF{eq:KKT-Comlementarity} with an equivalent equality so that the KKT system becomes a system of equations. This result can be used to establish the continuity of solutions in $\mu$, as the following lemma shows.

\begin{lem}
\label{lem:continuityInMu}Under Assumptions \ref{assu:Non--empty} and \ref{ass:ULICQ},  $\argminTL\mP$  is u.h.c. in $\mu$;  $\minV\mP$ is continuous in  $\mu$ for $\mu\geq0$.  
\end{lem}

\begin{prooF}{lem:continuityInMu}
By Proposition \reF{prop:FBproperty}   equation \eqreF{eq:KKT-Comlementarity} is equivalent to 
\begin{equation}
\phi\left(\Moment{\theta}{\measTrue},\lambda_{j}\right)  =0 \mbox{ for }j\in\setIneq.
\label{eq:FBKKTsystem}
\end{equation}
Solutions to \eqreF{eq:KKT-FOC},\eqreF{eq:KKT-EQ},\eqreF{eq:FBKKTsystem}  coincide with solutions to
\begin{equation}
\Psi(\theta,\lambda;\mu,\measTrue)  \bydef \left\Vert \grad[ {\theta}]\funLLagrangian\right\Vert _{2}^{2}+\sum_{j\in\setEq}\left(\Moment{\theta}{\measTrue}\right)^{2}+\sum_{j\in\setIneq}\left(\phi\left(\Moment{\theta}{\measTrue},\lambda_{j}\right)\right)^{2} = 0.\label{eq:FBargmin}
\end{equation}    Lemmas~ \reF{lem:UniqueMin}-\reF{lem:boundOnLambda} imply that  $\funLLagrange[\mu]$ is unique and satisfies (\reF{eq:lagrangeNorm})
for any $\mu\in[0,1/2]$. So the solution to (\reF{eq:FBargmin}) coincides
with solutions of 
 
\begin{equation}
\begin{aligned}
& \underset{\theta,\lambda}{\min}
& & \Psi(\theta,\lambda;\mu,\measTrue)  \\ \label{eq:FBargmin-1}
& \text{s.t.} & &  {\theta}\in\Theta,{\lambda}\in\mathbb{R}^{k},\ltwo{ {\lambda}}\leq \CL.
\end{aligned}
\end{equation}
The objective function of this program is continuous in $\mu$ and the domain is a compact valued continuous correspondence in $\mu$. By the Maximum Theorem (see \citet{ok2007real}) $\funLXi$
is u.h.c. function of $\mu\geq0$. 

Function $\minV\mP=\eOne\funLArg+\mu\left\Vert \funLArg\right\Vert ^{2}$
is a composition of u.h.c. functions and hence, by Theorem
VI.2.1' from \citet{berge1963topological}, is u.h.c. in $\mu\in\mathbb{R}_{+}$. Since by definition $\funLVal$ is a single\textendash valued function, u.h.c. implies continuity in $\mu\geq0$ for any fixed $\measTrue$. 
 \end{prooF} 

\newpageLemma
\subsection{Smoothness properties } \label{app:smoothness}
 
{%%lemma directional derivative
\newcommand{\psiD}{\left(\begin{array}{c}
(\dX)^{\prime}\argminL + 2 \dmu\argminT\\
\matD (\dX\argminT-\dy)
\end{array}\right)}

\newcommand{\dKKTmat}{\left(\begin{array}{cc}
2\mu\I[d] & \dA^\prime \\
\dA &   0
\end{array}\right)}

In this section we will study the directional derivatives of the value and the optimal solutions of  Program~\eqref{prog:regProg}.
We will pursue this goal by taking a limit of the perturbed program defined below as the size of the perturbation goes to zero. 
Consider a perturbation in parameters $\expect\Data=(\Ap|\bp)$ and $\mu$ in a direction $h^\prime=(\text{vec}(\dW)^\prime,\hm)\in \R^{k \left(d+1\right)+1}$, where  $\dW \bydef(h_A,h_b) \in \R^{k\times\left(d+1\right)} $ and $\hm\in\R$. 
The corresponding perturbation in the constraints is $\dMomentVec{\theta}\bydef\dX \theta -\dy$.
Given these directions, for any $t\geq0$, $\mu>0$  we can define a perturbed program, 
\begin{eqnarray}
\underset{ {\theta}\in\Theta}{\min} & \eOne \theta+(\mu+t\dmu)\ltwo{\theta}^{2},\label{prog:regProgPerturbed}\\
\mbox{s.t. } & \begin{cases}
\Moment{\theta}{\measTrue} + t \dMoment{\theta}  =0 &\mbox{for }j\in\setEq,\\
\Moment{\theta}{\measTrue} + t \dMoment{\theta}   \leq 0 &\mbox{for } j\in\setIneq.
\end{cases}\nonumber 
\end{eqnarray}
Perturbations of inequality constrained programs can lead to changes in the set of the constraints active at the optimum in response to arbitrarily small perturbations. 
It is instructive to consider the following sets of constraints for the unperturbed program, i.e. with $\ha=0,\hb=0,\hm=0$,
\begin{eqnarray*}
\setP\mP& \bydef &  \left\{ j\in\setIneq|\argminL_j\mP>0\right\}  \cup\setEq,\\
\setM\mP & \bydef & \left\{ j\in\setIneq|\Moment{\argminT\mP}{\measTrue}>0\right\}, \\
\setO\mP & \bydef & \left\{ j\in\setIneq|\argminL_j\mP=0,\Moment{\argminT\mP}{\measTrue}=0\right\} ,\\
\activeSet\mP & \bydef & \setO\mP\cup\setP\mP.
\end{eqnarray*}
Set $\setP$ contains active inequality constraints with positive Lagrange multipliers and equality constraints. 
These constraints will remain active for small enough perturbations in any direction $\ha,\hb$ (by continuity of the Lagrange multipliers).
Set $\setM$ contains slack constraints. 
They will remain slack in response to sufficiently small perturbations (by continuity of the optimal solution and the constraints functions).
Set $\setO$ contains active inequality constraints with zero Lagrange multipliers. 
If we drop these constraints, the optimal solution will not change, but they play an important role in the perturbed program.
Constraints in $\setO$ become inactive in response to perturbations in some directions, no matter how small the perturbation is or remain active, and acquire positive Lagrange multipliers for other directions. 
The optimal solution $\xi$ would be fully differentiable iff  $\setO$ is empty as will be evident from the explicit formula for its directional derivative.
Finally, $\activeSet$ contains all active constraints at the optimal solution. 
 
Suppose that the perturbation size $t>0$ is small enough such that Program \eqreF{prog:regProgPerturbed} satisfies Assumptions~\ref{assu:Non--empty}-\ref{assu:NoOveridentification}. 
Then it has a unique solution $\argminTL_h(t)$ for $0\leq t<T$  which can be represented as $\argminTL_h (t) =\argminTL+t\dot{\argminTL}_h$, as the following lemma shows.
The directional derivative $\dot{\argminTL}_h^\prime\mP\bydef(\dot{\argminT}^\prime \mP ,\dot{\argminL}^\prime\mP)$ will depend on the following objects,
\begin{eqnarray*}
\setD\mP & \bydef &  \left\{ j\in\setO\mP|\dot{\argminL}_{h;j}\mP>0\right\} \cup \setP\mP,\\
\dA\mP & \bydef & \matD\mP \Ap,\\
\projD  & \bydef & \I[d]-\dA^\prime\left(\dA\dA^\prime\right)^{-1}\dA,\\
A^\dagger & \bydef & A^\prime\left(AA^\prime\right)^{-1}.
\end{eqnarray*}
I suppress the argument $\mP$. 
 
\begin{lem}[Local linear representation]
\label{lem:directionalDerivative}Suppose that Assumptions \ref{assu:Non--empty} and \ref{ass:ULICQ} hold for $\measTrue$. There is a neighborhood $\left[0,T\left(\mu,h,\measTrue\right)\right]$ in which 
Program \eqreF{prog:regProgPerturbed} has a unique solution $\argminTL_h (t)=\argminTL+t\dot{\argminTL}_h$  with
\begin{equation}
 \dot{\argminTL}_h = -\left(\begin{array}{cc}
\left(2\mu\right)^{-1}\projD & \dA^\dagger\\
(\matD)^\prime(\dA^\dagger)^\prime  & -2\mu(\matD)^\prime\left(\dA\dA^\prime\right)^{-1}
\end{array}\right)\psiD.
\end{equation}

\end{lem}

\begin{prooF}{lem:directionalDerivative}
By Lemma \reF{lem:UniqueMin}, if $t=0$ program (\reF{prog:regProgPerturbed})
has a unique solution $\argminTL$. Since this solution satisfies LICQ  and SOSC, it  is strongly regular by Proposition 5.38 from \BS.
The remaining argument follows the proof of Theorem 5.60 from \BS, which uses  an  implicit function theorem for generalized equations (Theorem 5.13 in the same book)  at a strongly regular solution.
We are going to apply it to the KKT conditions for Program (\reF{prog:regProgPerturbed}) at the strongly regular solution~$\argminTL$,
\begin{numcases}{}\label{eq:KKTsystem-perturbed}
(\eOneVec+2\left(\mu+\dmu t\right) {\theta} )^\prime  =- {\lambda}^{\prime}\left(\Ap+t\dX\right),& \label{eq:KKT-FOC-perturbed}\\
\Moment{\theta}{\measTrue} + t \dMoment{\theta} =0 & $j\in\setEq$,\label{eq:KKT-EQ-perturbed}\\
\phi\left(\Moment{\theta}{\measTrue} + t \dMoment{\theta},\lambda_{j}\right) =0 & $j\in\setIneq$. \label{eq:KKT-Comlementarity-perturbed}
\end{numcases}
By Theorem 5.60 of \BS, $\argminTL_h(t)$ is analytic in $t$ in some neighborhood $\left[0,T\left(\mu,h,\measTrue\right)\right]$ , i.e. it can be represented as power series. 
First, let us compute the linear term.
By the strong regularity and Theorem 5.13 in \BS, there exist a unique solution $\left(\dot{\argminT} ,\dot{\argminL}\right)$
to the following system of equations (this system is the gradient of \eqref{eq:KKT-FOC-perturbed}-\eqref{eq:KKT-Comlementarity-perturbed} with respect to $t$ at point $t=0$)
\begin{numcases}{} 
2\mu  \dot{\argminT}^{\prime}\I[d]+  \dot{\argminL}^{\prime}\EX =-\argminL^\prime\dX-2\dmu\argminT^\prime,&\label{eq:KKT-FOC-d}\\ 
 e_j^\prime \Ap \dot{\argminT}   +  \dMoment{\argminT} =0 & $j\in\setP\mP$,\label{eq:KKT-EQ-d}\\
\phi\left(e_j^\prime \Ap \dot{\argminT}   +  \dMoment{\argminT},\dot{\argminL}_{j}\right)  =0 & $j\in\setO\mP$, \label{eq:KKT-Comlementarity-d}\\
\dot{\argminL}_{h;j} =0 & $j\in\setM\mP$.\label{eq:KKT-EqQ-d}
\end{numcases}
This unique solution determines the set $\setD$. System \eqreF{eq:KKT-FOC-d}-\eqreF{eq:KKT-EqQ-d} can be represented in a matrix form:\footnote{Compare Equation \eqref{eq:directionalLinearization} with Equation (5.81) on page 186 in \cite{shapiro2014lectures}.}
\begin{align}
\dKKTmat\left(\begin{array}{c}
\dot{\argminT}\\
\matD\dot{\argminL}
\end{array}\right) & = -\psiD.\label{eq:directionalLinearization} 
\end{align}
In addition to that, $\dot{\argminL}=(\matD)^{\prime}\matD\dot{\argminL}$. One can check by direct computation that
\[
\dKKTmat^{-1}=\left(\begin{array}{cc}
\left(2\mu\right)^{-1}\projD & \dA^\dagger\\
 (\dA^\dagger)^\prime  & -2\mu \left(\dA\dA^\prime\right)^{-1}
\end{array}\right).
\]

Since the higher order derivatives of every constraint function and the objective function of Program \eqreF{prog:regProgPerturbed} with respect to $t$ are zero, the higher order directional derivatives of $\argminTL_h$ are equal to zero at $t=0$. Thus, the power series expansion of $\argminTL_h$ has only constant and linear terms. 
\end{prooF}

}

\newpageLemma
Now we can rewrite   Program \eqreF{prog:regProgPerturbed}  in an explicit form assuming $h_\mu=0$,
  \begin{eqnarray}
\underset{ {\theta}\in\Theta}{\min} & \eOne \theta+\mu\ltwo{\theta}^{2},\label{prog:regProgPerturbedExplicit}\\
\mbox{s.t. } & \begin{cases}
e_j^\prime (\Ap + t h_A) \theta = \bp + t h_b  &\mbox{for }j\in\setEq,\\
e_j^\prime (\Ap + t h_A) \theta \leq \bp + t  h_b &\mbox{for } j\in\setIneq.
\end{cases}\nonumber 
\end{eqnarray}

\begin{lem}
\label{lem:boundOnVgrowth}Suppose that Assumptions \ref{assu:Non--empty} and \ref{ass:ULICQ}  hold for $P$. There exist such $\delta>0$ such that for any $h = (h_A,h_b)\in \R^{k\times(d+1)}$ with norm $\ltwo{h}<\delta$ and any $\mu\in\left(0,1/2\right]$ and $t\in [0,1]$  the solution of Program \eqreF{prog:regProgPerturbedExplicit}   satisfies 
\begin{align}
\ltwo{\argminT_h(t)- \argminT}&\leq L_\theta \frac{\ltwo{t h}}{\mu} \label{eq:normThetadot},\\
\ltwo{\argminL_h(t)- \argminL}&\leq L_\lambda \ltwo{t h} \label{eq:normLambdadot},\\
\abs{\minVd(t)-\minV-t\argminL^\prime  (h_A\argminT-h_b)} &\leq L_v \frac{\ltwo{t h}^{2}}{\mu}, \label{eq:boundsOnLinearApproximation}
\end{align}
 where $L_\theta = \frac{ \sqrt{2\CT^3 } }{(\underline{\eta}/2)}$, $L_\lambda = \frac{(\ltwo{\expect\Data}+\delta)\CT^4+ \eta^2\CT^2}{(\underline{\eta}/2)^4}  $ , and $L_v = \frac{2\CT^3}{(\underline{\eta}/2)^{2}}	$.
\end{lem}

\begin{prooF}{lem:boundOnV}
First, consider $t\in[0,T(\mu,h,\measTrue)]$ with $T(\mu,h,\measTrue)$ defined in Lemma~\ref{lem:directionalDerivative}.
By Lemma \reF{lem:directionalDerivative} the value function $\minVd(t)  \bydef\eOneVec^\prime\argminTd(t)+\mu\ltwo{\argminTd(t)}^{2}$ can be represented as \begin{equation}
\minVd(t)  = \minV+t\left(\eOneVec+2\mu\argminT\right)^{\prime}\dot{\argminT}+\mu t^{2}\ltwo{\dot{\argminT}}^{2}. \label{eq:linearValueExpansion}
\end{equation}
First, consider the second term. Since by definition $\setP \subseteq \setD$,
 we have $\argminL^\prime  =\argminL^\prime(\matD)^{\prime}\matD.$ 
Correspondingly, $\argminL^\prime\Ap  = \argminL^\prime(\matD)^\prime\dA.$
By Lemma \reF{lem:UniqueMin},
$\left(\eOneVec+2\mu\argminT\right)^{\prime}=-\argminL^\prime\Ap.$
So 
\begin{align}
 \left(\eOneVec+2\mu\argminT\right)^{\prime}\projD & =-\argminL^\prime(\matD)^\prime(\dA\projD) =0,\label{eq:trick1}\\
\left(\eOneVec+2\mu\argminT\right)^{\prime}\dA^\dagger & =-\argminL^\prime(\matD)^\prime(\dA\dA^\dagger) =-\argminL^\prime(\matD)^\prime.\label{eq:trick2}
\end{align}
Equations (\reF{eq:trick1}) and (\reF{eq:trick2}) imply that
 \begin{equation}
 \left(\eOneVec+2\mu\argminT\right)^{\prime}\dot{\argminT} =-\argminL^\prime \dMomentVec{\argminT} .\label{eq:trick3}
 \end{equation} 
Second, by Lemma \reF{lem:boundOnLambda} and Remark~\reF{rem:inverse Norm of active constraints} for any  $\mu\leq1/2$ ,
\begin{equation}
\ltwo{\left(\dA\dA^\prime\right)^{-1}}   \leq\CT\eta ^{-2}(\measTrue)\text{ and }
\ltwo{\argminL}^2 \leq\CT^3\eta^{-2}(\measTrue).\label{eq:boundL}
\end{equation}
Then by the triangular inequality and inequalities  \eqreF{eq:boundL} (and the fact that $\CT\geq1$)
\begin{align}
\ltwo{\dot{\argminT}}^{2} &= \frac{1}{\left(2\mu\right)^{2}}\left(\argminL^\prime\dX\right) \projD\left(\argminL^\prime\dX\right)^\prime+
\dMomentVec{\argminT}^{\prime}(\matD)^\prime\left(\dA\dA^\prime\right)^{-1}\matD\dMomentVec{\argminT}\\
&\leq \ltwo{\dW}^{2}\frac{\CT^3}{\eta^{2}(\measTrue)}\left(  \frac{1}{\mu^2}+1\right),
\end{align} 
which implies \eqreF{eq:normThetadot}. Equation \eqreF{eq:normLambdadot} can be proven similarly,
\begin{align}
\ltwo{\dot{\argminL}}   = & \ltwo{(\matD)^\prime(\dA^\dagger)^\prime \left(\argminL^\prime\dX\right) - 2 \mu (\matD)^\prime\left(\dA\dA^\prime\right)^{-1}\matD\dMomentVec{\argminT} }\\
\leq & \frac{\ltwo{\expect\Data}\CT^4+ \eta^2(\measTrue)\CT^2}{\eta^4(\measTrue)}  \ltwo{\dW} 
\end{align} 
 
Finally, the bound in \eqreF{eq:boundsOnLinearApproximation} follows from equations \eqreF{eq:normThetadot}, \eqreF{eq:linearValueExpansion}, and \eqreF{eq:trick3}.

 To extend the argument to the entire interval $t\in[0,1]$, notice that $\eta(\measTrue)$ and $s(\measTrue)$ are Lipshitz-continuous. So there exist $\delta>0$ such that $\eta(\expect W + t h)>\underline{\eta}/2$ and $s(\expect W + t h)>\underline{s}/2$ for any $h = (h_A,h_b)\in \R^{k\times(d+1)}$ with norm $\ltwo{h}<\delta$ and any $t\in [0,1]$. 
This $\delta$ can be chosen uniformly for $\measTrue\in\Measures$.
Now we can replace $\ltwo{\expect\Data}$ with  $ {(\ltwo{\expect\Data}+\delta)}$ and $\eta(\measTrue)$ with $\underline{\eta}/2$ to obtain uniform constants   $L_\theta$, $L_\lambda$ , and $L_v$.

\end{prooF}

\subsection{Proof of Theorem~\ref{thm:bounds}\label{sec:Proof-of-Theorem1}}
\newpageLemma

\begin{lem}
\label{lem:constantTheta} Suppose that Assumptions \ref{assu:Non--empty} and \ref{ass:ULICQ} hold. There exists some $\bar{\mu}\left(\measure\right)>0$
such that for any $\mu<\bar{\mu}\left(\measure\right)$ the solution
to program \eqreF{prog:regProg}, $\argminT\mP$ is constant. 
\end{lem}

\begin{prooF}{lem:constantTheta}
Consider a direction $h^\prime=(\text{vec}(\dW)^\prime,\hm)\in \R^{k \left(d+1\right)+1}$ satisfying $\dW=0$, $\dmu=1$ and any $\mo>0$ in a neighborhood of $0$. By Lemma \reF{lem:directionalDerivative} we get
\begin{equation}
\dot{\argminT}\moP =-\frac{1}{\mo}\projD\moP\argminT\moP.\label{eq:thetadotmu0}
\end{equation} 

Substitute difference in eq.\eqreF{eq:KKT-FOC} (Lemma \reF{lem:KKT}) between $\mu=0$ and $\mo$ for $\argminT$ in \eqref{eq:thetadotmu0},
\begin{equation}
\dot{\argminT}\moP  =(2\mo^{2})^{-1}\projD\moP\Ap^\prime \left(\argminL\moP-\argminL\oP\right). \label{eq:derivativeOfThetaInMu}
\end{equation} 
Now we need to show that $\dot{\argminT}\moP =0$.  To establish this result, we need to study the behavior of the set of inequalities with positive Lagrange multipliers. 

Consider any $j\in\setIneq$. We know by Lemma \reF{lem:continuityInMu} that $\argminL_j\mP$ is continuous in $\mu$.  If $\argminL_j\oP>0$
then by continuity $\argminL_j\mP>0$ in some neighborhood $\left(0,\bar{\mu}_{j}(\measTrue)\right]$.  If $\argminL_j\oP=0$ set $\bar{\mu}_{j}=1$. Take $\mbP\bydef\min_{j\in\setIneq}\bar{\mu}_{j}(\measTrue)$. WLOG suppose that $\mo\in\left[0,\bar{\mu}\left(\measTrue\right)\right]$, so 
we get the inclusion 
\begin{equation}
\setP\oP\subseteq\setP\moP.\label{eq:indexSetsInclustions1}
\end{equation}
By definition of $\setD$
\begin{equation}
\setP\moP\subseteq\setD\moP.\label{eq:indexSetsInclustions2}
\end{equation}
By definition of the index matrices, inclusions \eqreF{eq:indexSetsInclustions1} and \eqreF{eq:indexSetsInclustions2} imply that
\begin{align}
\argminL\oP & =(\matD\moP)^{\prime}\matD\moP\argminL\oP,\\
\argminL\moP & =(\matD\moP)^{\prime}\matD\moP\argminL\moP,
\end{align}
so \begin{equation}
\Ap^\prime \left(\argminL\moP-\argminL\oP\right) = \dA^\prime \moP(\matD\moP)\left(\argminL\moP-\argminL\oP\right). 
\end{equation} 

Since by definition $\projD\moP\dA\moP=0$ and $\projD$ is a symmetric matrix, equation \eqreF{eq:derivativeOfThetaInMu} implies  $\dot{\argminT}\moP =0$. By Lemma~\reF{lem:continuityInMu} the single valued function $\argminT\mP$ is continuous for $\mu>0$. So the r.h.s. directional derivative being equal to zero implies that $\argminT\moP = \argminT\MP $ for any $\mo\in\left(0,\bar{\mu}\left(\measTrue\right)\right].$
\end{prooF}

\begin{rem} 
Equation \eqreF{eq:KKT-FOC} from Lemma \reF{lem:KKT} with $\mu=0$ and $\mu=\mo$ also implies
$$\argminL\moP=\argminL\oP-2\mo\argminT^\prime\MP\dA^\dagger\moP\matD\moP.$$
This implies that $\argminL\mP$ is Lipschitz at $\mu=0$. By Lemma \reF{lem:boundOnLambda}  the Lipschitz constant can be taken equal to $2\CL$. So, for any $j\in\setIneq$ with $\argminL_j\oP>0$, we can take $\bar{\mu}_j\left(\measTrue\right)= \argminL_j\oP/2\CL$. 
On top of that, Lemma~\reF{lem:directionalDerivative} implies that $\argminL\mP$ is Lipschitz in $\mu$ with the same constant for any $\mu\in[0,1/2]$
\end{rem}

\newpageLemma

\begin{proof}[Proof of Theorem~\ref{thm:bounds}]
 
 Take any $\theta^* \in \argminT(0,\measTrue)$. Since $\theta^*$ is a feasible point of Program~\eqref{prog:regProg}, 
\begin{equation}
\argminT_1(\mu,\measTrue) + \mu\ltwo{ \funLArg[\mu][\measTrue]}^{2}\leq    \theta^*_1 + \mu \ltwo{\theta^*}^{2}.
\end{equation}
By definition $\funLValO[\measTrue]= \theta^*_1$ which immediately implies the  l.h.s. inequality in \eqref{eq:innerOuterBounds}. 

The first coordinate $\argminT_1(\mu,\measTrue)$ is increasing in $\mu$ while the norm  $\ltwo{ \funLArg[\mu][\measTrue]}^{2}$ is decreasing in $\mu$. Since  $\kappa\geq\mu\geq0$, we get
 \begin{align}
\ltwo{ \funLArg[\mu][\measTrue]}^{2}&\geq  \ltwo{\funLArg[\kappa][\measTrue]}^{2},\\
\funLValO[\measTrue] &\leq  \argminT_1(\mu,\measTrue).
 \end{align}
 These two inequalities imply the r.h.s. inequality in \eqref{eq:innerOuterBounds},
\begin{equation}
\funLValO[\measTrue] \leq   \argminT_1(\mu,\measTrue) + \mu(\ltwo{ \funLArg[\mu][\measTrue]}^{2}-\ltwo{ \funLArg[\kappa][\measTrue]}^{2})=\minVin(\mu,\kappa,\measTrue).
\end{equation}
The remaining part of the Theorem's assertion follows from  Lemma~\ref{lem:constantTheta}. 
\end{proof}

\newpageLemma
\subsection{Proof of Theorem \ref{thm:consistency}\label{sec:Proof-of-Theorem2}}
\begin{prop} \label{prop:integrability}
Suppose that $\expect \abs{\xi}^{1+\epsilon}\leq \infty$ for some $\epsilon>0$. Then for any $r>0$  $$ \expect [ \abs{\xi} I\{\abs{\xi}\geq r \} ]\leq \expect \abs{\xi}^{1+\epsilon} / r^\epsilon.$$
\end{prop}
\begin{proof}
The result follows from the monotonicity of integrals.
\end{proof}

Let $\pi(P,Q)$ denote the Prokhorov distance between probability laws in $P$ and $Q$, which induces the weak topology (see p. 456 in \citet{van1996weak}). Let $$  G_{n}\left(\measTrue\right)\bydef\sqrt{n}\left(\vecM\left(\sMean{\data_i}\right)-\vecM\left(\expect\Data\right)\right). $$

\begin{lem}\label{lem:CLTandLLN} Consider $\mathcal{P} $, a class of distributions satisfying  Assumption 
 \reF{assu:Moments},  and any $\epsilon>0$. Then
\begin{eqnarray}
\lim_{n\to\infty}\sup_{\measTrue \in \mathcal{P}}\measTrue\left(\sup_{m\geq n}\ltwo{\frac{1}{m}\sum^{m}_{i=1}{\data_{i}-\expect W}}\geq\epsilon\right) & = & 0,\label{eq:1stMoments}\\
\lim_{n\to\infty}\sup_{\measTrue \in \mathcal{P}}\measTrue\left(\sup_{m\geq n}\ltwo{\frac{1}{m}\sum^{m}_{i=1}{\data_{i}\otimes\data_{i}}-\expect [ \Data \otimes \Data]}\geq\epsilon\right) & = & 0,\label{eq:2ndMoments}\\
\lim_{n\to\infty}\sup_{\measTrue \in \mathcal{P}}\pi( G_{n}\left(\measTrue\right),N\left(0, {\Omega}_{\measTrue}\right)) & = & 0, \label{eq:CLTnonuniform}
\end{eqnarray}
where $ {\Omega}_{\measTrue}=\text{Cov}_{\measTrue}\left(\vecM\left(\Data\right)\right)$.
\end{lem}

\begin{prooF} {lem:CLTandLLN}
Consider any combination of indices $r,\ell,j,m$. 
Assumption \reF{assu:Moments} together with the Schwarz inequality implies,
\begin{equation}
\expect\abs{W_{r,\ell}W_{j,m}}^{1+\varepsilon/2}   \leq\left(\expect\abs{W_{r,\ell}}^{2+\varepsilon}\expect\abs{W_{j,m}}^{2+\varepsilon}\right)^{1/2}\leq\bar{M}.\label{eq:uniformIntegrability}
\end{equation}
So the random variables $\abs{W_{r,\ell}}$ and $\abs{W_{r,\ell}W_{j,m}}$ have correspondingly finite $1+\varepsilon/2$ and $2+\varepsilon$ moments. The bound~\eqreF{eq:uniformIntegrability} on the moments is independent of $\measTrue\in\mathcal{P}$, so these random variables are uniformly integrable on $\mathcal{P}$ by Proposition~\reF{prop:integrability}. The limits (\reF{eq:1stMoments}) and (\reF{eq:2ndMoments})
follow immediately from Proposition A.5.1 in \citet{van1996weak}. The result (\reF{eq:CLTnonuniform}) follows from Proposition A.5.2 in the same book.
\end{prooF}

\newpageLemma
 \begin{lem}\label{lem:wellDefined} Suppose that $\measTrue\in\Measures$. Then   $\measEmp$ satisfies Assumptions~\ref{assu:Non--empty} and  LICQ with probability approaching 1 uniformly in $\measTrue\in\Measures$.
\end{lem}

\begin{prooF}{lem:wellDefined}
 
Consider any $\measTrue\in\Measures$.
Since such a $\measTrue$ satisfies Assumption~\ref{assu:Non--empty}, there exists a set of constraints $\jSet$ with $\card{\jSet}=d$ and containing $\setEq$ such that
\begin{equation}
    \theta^{\jSet}_\measTrue \bydef ((\Ap^\jSet)^\prime\Ap^\jSet)^{-1}(\Ap^\jSet) \bp^\jSet \in \setID,
\end{equation}
 where $\Ap^\jSet \bydef \jMat\Ap$ and $\bp^\jSet \bydef \jMat\bp$. 
 By Assumption~\ref{assu:NoOveridentification} for all $j\in\setIneq \backslash \jSet$ we have
 \begin{equation}
     e_j(\Ap  \theta^{\jSet}_\measTrue -\bp)\geq \underline{s}. 
 \end{equation}
The function $(\Ap,\bp)=\expect W $ is uniformly continuous on $\Measures$ since this class of measures in uniformly integrable, as was shown in the proof of Lemma~\ref{lem:CLTandLLN}.
The function $\theta^{\jSet}_\measTrue$ is uniformly continuous on $\Measures$ since the product matrix $ (\Ap^\jSet)^\prime\Ap^\jSet$ has eigenvalues uniformly bounded from below by $\eta^2$.
By Lemma~\ref{lem:CLTandLLN} and the continuous mapping theorem,
 \begin{equation}
 \lim_{n\to\infty}\sup_{\measTrue \in \mathcal{P}}\measTrue\{ \inf_{m\geq n;j\in \setIneq\backslash\jSet} \{   e_j(\hat{A}_m  \hat{\theta}_m^{\jSet} -\hat{b}_m)\}< \underline{s}/2 \}=0,
\end{equation}
where $\hat{A}_m$ and $\hat{b}_m$ are sample analog estimators of $\Ap$ and $\bp$ based on a sample of size $m$; $  \hat{\theta}_m^{\jSet}\bydef ((\hat{A}_m^\jSet)^\prime\hat{A}_m^\jSet)^{-1}(\hat{A}_m^\jSet) \hat{b}_m^\jSet$. This result implies that $\Theta(\measEmp)$ contains at least one element,  $ \hat{\theta}_n^{\jSet}$,  with probability approaching 1 uniformly in $\measTrue\in\Measures$ as $n\to\infty$.

Analogously, the continuous mapping theorem implies 
 \begin{equation}
 \lim_{n\to\infty}\sup_{\measTrue \in \mathcal{P}}\measTrue\{ \inf_{m\geq n}  \eta(\mathbb{P}_m) < \underline{\eta}/2, \inf_{m\geq n} s(\mathbb{P}_m) < \underline{s}/2\}=0.
 \end{equation}
The result follows from Lemma~\ref{lem:LICQ}.
 
\end{prooF}

\newpageLemma

  \begin{lem}\label{lem:Bahadur} Suppose  that  $\measTrue\in\Measures$ and that $\mn\to0$. Then
  \begin{equation}
 \funLValHat[\mu_{n}]= \minV\mnP +\sMean{\argminL\mnP^\prime\momentVec{\data_i}{\argminT\mnP}  } + O_\Measures(\frac{1}{\mn n}). 
\end{equation}
 
 \end{lem}
 \begin{prooF}{lem:Bahadur}
 First note that by Lemma~\ref{lem:wellDefined} the random variables $  \funLArgHat[\mn]$ and $
 \funLValHat[\mu_{n}]$ are well defined with probability approaching 1 uniformly in $\measTrue\in\Measures$. 
Consider $t=1/\sqrt{n}$ and $h$ such that $\dW =\sqrt{n} (\sMean{\data}-\expect\Data)$ and $\dmu=0$. 
The sequence of perturbations satisfies  $\sqrt{n}(\sMean{\data_i}-\expect\Data)=O_\Measures(1)$  (i.e. it has uniformly tight measures)  by \eqref{eq:CLTnonuniform} in Lemma~\reF{lem:CLTandLLN} and the fact that $\ltwo{{\Omega}_{\measTrue}}\leq \bar{M}$  (by Jensen's inequality).
The assertion of the Lemma follows from equation  \eqref{eq:boundsOnLinearApproximation} in Lemma~\reF{lem:boundOnVgrowth}.  In fact, by Lemma~\ref{lem:CLTandLLN} 
 \begin{equation}
     \lim_{n\to\infty}\sup_{\measTrue \in \mathcal{P}}\measTrue\left(\sup_{m\geq n}\ltwo{\frac{1}{m}\sum^{m}_{i=1}{\data_{i}-\expect W}}\geq\delta\right)=  0,
 \end{equation}
 so  the perturbation is small enough to preserve the results of Lemma~\ref{lem:boundOnVgrowth}, i.e. $\ltwo{th_W}<\delta$,  with probability approaching 1 uniformly as sample size grows.

\end{prooF}

\newpageLemma

  Let's introduce the following definitions
\begin{align*}
\Cov(\theta)\bydef & \expect\left[\matsq{\momentVec{\Data}{\theta}}\right]- \matsq{\MomentVec{\theta}{\measTrue}},\\
\funCovHat{\theta}\bydef & \sMean{\matsq{\momentVec{\data_i}{\theta}}}- \matsq{\sMean{\momentVec{\data_i}{\theta}}}.
\end{align*}
\begin{proof}[Proof of Theorem~\ref{thm:consistency}]
First note that  Lemma~\ref{lem:wellDefined} the random variables $  \funLArgHat[\mn]$ and $
\funLLagrangeHat[\mn]$ are well defined with probability approaching 1 uniformly in $\measTrue\in\Measures$.
Consider $t $ and $h$ as in the proof of Lemma~\ref{lem:Bahadur}.  
Equation~\eqreF{eq:coupling} follows from Lemma~\reF{lem:CLTandLLN},  Lemma~\reF{lem:Bahadur}, and Slutsky's theorem. 
The~results \eqreF{eq:thetaConsistency} and \eqreF{eq:lambdaConsistency} follow from Lemma  \reF{lem:boundOnVgrowth}. 
Finally, by the  triangular inequality  
\begin{align}
  \abs{\funLStdHat^2  - \funLStd ^2}  =&\abs{\trns{\funLLagrangeHat}\funCovHat{\funLArgHat}\funLLagrangeHat-\trns{\funLLagrange}\Cov(\funLArg)\funLLagrange}\leq\nonumber\\
&\abs{\trns{\funLLagrangeHat}\funCovHat{\funLArgHat}\funLLagrangeHat-\trns{\funLLagrange}\funCovHat{\funLArg}\funLLagrange}\nonumber\\+
&\abs{\trns{\funLLagrange}\funCovHat{\funLArg}\funLLagrange-\trns{\funLLagrange}\Cov(\funLArg)\funLLagrange}.\label{eq:consistancy of std, tri Inequality}
\end{align}
Together with \eqreF{eq:thetaConsistency}, \eqreF{eq:lambdaConsistency} and  Lemmas~\reF{lem:boundOnLambda} and \reF{lem:CLTandLLN}, it implies \eqreF{eq:convergenceSpeedStd}.  

\end{proof}

\newpageLemma

\newpageLemma
\subsection{Proof of Theorem \ref{thm:TheoremUniform}\label{sec:Proof-of-Theorem4}}

\begin{lem} \label{lem:argminBound}  
For any $\epsilon>0$ there exists $R\geq0$ such that for any $n$ the following uniform bound holds
\begin{equation}
   \sup_{\measTrue \in \mathcal{P}}\measTrue\left(\sqrt{n}\ltwo{\theta^*(\measEmp)-\theta^*(\measTrue)}\geq R\right)\leq  \epsilon .\label{eq:argminBoundConisitency} 
\end{equation}

\end{lem}
\begin{proof}
The proof is based on the delta method applied to  $ \theta^*(\measTrue) =  \theta^*(\Ap,\bp)$,  a composition of directionally differentiable functions. (Since the space of $\Ap,\bp$ is finite dimensional, the directional derivatives in G\^{a}teaux and Hadamard sense coincide. )

First note that  $ \funLValO[\measTrue] = \minV (\Ap,\bp)$ is directionally differentiable function. 
To see it, consider the minimax representation, which is valid since $\measTrue$ satisfies Assumption~\ref{assu:Non--empty},

\begin{equation}
\minV(\Ap,\bp)=\min_{ {\theta} \in\Theta}\max_{ {\lambda}\in\R^p\times\R_{+}^{k-p},\ltwo{{\lambda}}\leq \CL} \{\theta_{1}  +{\lambda}^{\prime}(\Ap\theta-\bp)\}.  \label{prog:minmax1}
\end{equation}
Here  we also used Lemma~\ref{lem:boundOnLambda} to bound $\lambda$ and make the domain compact. By Theorem 7.28 from  \cite{shapiro2014lectures} for any direction $(h_A , h_b)\in \R^{k\times (d+1)}$ we have
\begin{equation}
  \dot{\minV}  (\measTrue|h_A, h_b)\bydef\lim_{t\to 0^+}\frac{ \minV(\Ap+t h_A ,\bp +t h_b)-  \minV(\Ap,\bp)}{t} = 
 \min_{ {\theta} \in\argminT(\measTrue)}  \{ {\argminL(\measTrue)}^{\prime}(h_A\theta- h_b)\}.
\end{equation}
Similarly, 
\begin{equation}
    {\theta}^{\pm}_i(\Ap,\bp) = \abs{\min_{ {\theta} \in\Theta}\max_{ {\gamma}\in\R^p\times\R_{+}^{k-p},\ltwo{\gamma}\leq \CL ,\nu\geq0} \{\pm {\theta_i} +{\gamma}^{\prime}(\Ap  \theta-\bp ) +\nu (\theta_1-\minV(\measTrue)-\mn)   \}}.\label{eq:conservativeTheta}
\end{equation}
Once again, by Theorem 7.28 from  \cite{shapiro2014lectures} for any direction $(h_A , h_b)\in \R^{k\times (d+1)}$ we have
\begin{equation}
\dot{\theta}^{\pm}_i (\measTrue|h_A, h_b) = 
 \min_{ {\theta} \in\argmin \eqref{eq:conservativeTheta}    }  \{ {\underline{\gamma}(\measTrue)}^{\prime}(h_A\theta- h_b )\} +  \underline{\nu}(\measTrue)    \min_{ {\vartheta} \in\argminT(\measTrue)}  \{ {\argminL(\measTrue)}^{\prime}(h_A\vartheta- h_b)  \}.
\end{equation}
Here we used Proposition 2.47 from \cite{bonnans2013perturbation}, the chain rule for directional derivatives.
By the same proposition, the vector with maximal components $\theta^*(\measTrue)$ is directionally differentiable. 

Delta method (Theorem 7.67 in \cite{shapiro2014lectures}) and Lemma~\ref{lem:CLTandLLN} imply
\begin{equation}
    \sqrt{n}(\theta^*(\measEmp)-\theta^*(\measTrue)) = \dot{\theta}^* (\measTrue|G_{n}\left(\measTrue\right)) + o_p(1).
\end{equation}

By the compactness of $\Measures$, the vanishing term $ o_p(1)$ is uniformly bounded in probability in $\measTrue\in\Measures$. It is routine to compute a uniform bound on the directional derivative,  $\ltwo{\dot{\theta}^* (\measTrue|h_A, h_b)}\leq L_\Measures<\infty$.
The constant $ L_\Measures$ provides a uniform asymptotic bound  
\begin{equation}
   \sqrt{n}\ltwo{\theta^*(\measEmp)-\theta^*(\measTrue)} \leq  L_\Measures  \ltwo {G_{n}\left(\measTrue\right)}  + O_\Measures(1).
\end{equation}
By Lemma~\ref{lem:CLTandLLN}, this representation implies \eqref{eq:argminBoundConisitency}.

\end{proof}

\newpageLemma
 
\begin{proof}[Proof of Theorem~\ref{thm:TheoremUniform}]
The proof is analogous for all CI. Consider, for example, $  {\mbox{CB}}_{\alpha,n,\Measures} $. Pick an arbitrary measure $\measTrue\in\Measures$.
Consider any $\delta>0$  such that $z_{1-\alpha}\sigma^0>2\delta>0$.  Then by Lemma~\ref{lem:argminBound} and Theorem~\ref{thm:consistency} correspondingly there exist  $n(\delta,\epsilon)$ such that for any $n>n(\delta,\epsilon)$
\begin{align}
    &\inf_{\measTrue\in\Measures }    \measTrue\{\mn  \sqrt{n} \abs{  \ltwo{\theta^{* }(\measEmp)}^2-\ltwo{\theta^{* }(\measTrue)}^2 } \leq \delta\}  &\geq& 1-\epsilon,\\
    &\inf_{\measTrue\in\Measures}    \measTrue\{  z_{1-\alpha}\abs{\funLStdHat[\mn]  - \funLStd[\mn] }  \leq \delta\}  &\geq& 1-\epsilon,\\
    &\sup_{\measTrue\in\Measures}  \rho_n(\measTrue) &\leq&\delta
\end{align}

 By construction $ \ltwo{\theta ^{*}(\measTrue)}\geq \min_{ \theta\in\argminT(\measTrue)} \ltwo{\theta}$, so by Theorem~\ref{thm:bounds}
\begin{equation}
\argminT_1(\mu_{n},\measTrue) + \mn(\ltwo{ \funLArg[\mu_{n}][\measTrue]}^{2}-\ltwo{\theta ^*(\measTrue)}^2)\leq \funLValO[\measTrue].
\end{equation}
 Using this bound, 
\begin{align*}
\measTrue\left\{ \funLValHat[\mu_{n}]-\mu_{n} \ltwo{\theta ^*(\measEmp)}^2 -\funLStdHat[\mu_{n}]z_{1-\alpha}n^{-1/2}\leq\funLValO[\measTrue]\right\}  & \geq\\
\measTrue\left\{ \sqrt{n}(\funLValHat[\mu_{n}]-\funLVal[\mu_{n}][\measTrue])-\mu_{n}\sqrt{n} ( \ltwo{\theta^*(\measEmp)}^2-\ltwo{\theta^*(\measTrue)}^2 )\leq \funLStdHat[\mu_{n}] z_{1-\alpha}\right\}  & \geq \text{ (by \eqref{eq:normThetadot})}\\
\measTrue\left\{ \sqrt{n}(\funLValHat[\mu_{n}]-\funLVal[\mu_{n}][\measTrue])\leq \funLStd[\mn] z_{1-\alpha}-2\delta \right\}(1 -\epsilon)^2 &\geq \\
( \Phi(z_{1-\alpha} - \frac{2\delta}{\sigma^0}) -\delta)(1 -\epsilon)^2 & 
\end{align*}

 Since $\measTrue$ is   arbitrary, for any $n>n(\delta,\epsilon)$

\begin{equation}
\inf_{\measTrue\in\Measures}\min_{ {\theta}\in\setID[\measTrue]}\measTrue\left(\theta_{1}\in{\mbox{CB}}_{\alpha,n,\Measures}\right)\geq( \Phi(z_{1-\alpha} - \frac{2\delta}{\sigma^0}) -\delta)(1 -\epsilon)^2.\label{eq:liminfThm2}
\end{equation}
Hence,

\begin{equation}
\liminf_{n\to\infty}\inf_{\measSec\in\Measures}\min_{ {\theta}\in\setID[\measTrue_{n}]}\measTrue_{n}\left(\theta_{1}\in{\mbox{CB}}_{\alpha,n,\Measures}\right)\geq(1-\alpha).\label{eq:liminfThm2}
\end{equation}

\end{proof}

} % end small

 \section{Additional Results}

\subsection{Point-wise valid and non-conservative confidence intervals} \label{subsec:pointwise-results}

How conservative is $\mbox{CB}_{\alpha,n,\Measures}$ in the nonregular case? 
We can evaluate it by comparing its average length with that of a confidence bound that has exact point-wise asymptotic coverage probability in a Monte Carlo study.
As discussed in Section~\ref{subsec:asymptoticDistribution}, the existing (point-wise valid) nonconservative procedures for inference on support functions in the nonregular case are based on simulation methods (for example, bootstrap for directionally differentiable functions).
These procedures can be computationally costly in high-dimensional settings. 
Theorems~\ref{thm:bounds} and \ref{thm:consistency} suggest  alternative analytical confidence bounds on $\minV(\measTrue)$ using the following inner estimator: 
\begin{equation}
\minV^{in}(\mn,\kappa_n,\measEmp)   \bydef\funLValHat[\mu_{n}]-\mu_{n}\ltwo{\funLArgHat[\kappa_{n}]}^{2}.\label{eq:biasCorrectedEstimator}
\end{equation}
Theorem~\ref{thm:bounds} guarantees that the corresponding  population counterpart $\minVin(\mu,\kappa,\measTrue)$  coincides with $\minV(\measTrue)$  if  $\kappa_n$ and $\mu_n$ are sufficiently small.   
Theorem~\ref{thm:consistency} implies that if $\kappa_{n}$ converges to zero slower than $\mu_{n}$ and both converge slower that $1/\sqrt{n}$, then $\minV^{in}(\mn,\kappa_n,\measEmp) $ is an asymptotically normal estimator with variance $ \funLStd[\mn]$.
This property suggests the following point-wise-valid CSs:
\begin{equation}
\begin{cases}
\mbox{CB}_{\alpha,n} & =\left[\minV^{in}(\mn,\kappa_n,\measEmp)-z_{1-\alpha}n^{-1/2}\funLStdHat[\mu_{n}],\infty\right),\\
\mbox{CI}^{\theta_1}_{\alpha,n} & =\left[\minV^{in}(\mn,\kappa_n,\measEmp)-z_{1-\alpha}n^{-1/2}\funLStdHat[\mu_{n}];\maxV^{in}(\mn,\kappa_n,\measEmp)+z_{1-\alpha}n^{-1/2}\funUStdHat[\mu_{n}]\right],\\
\mbox{CI}_{\alpha,n}^{\mathcal{S}} & =\mbox{CI}^{\theta_1}_{\alpha/2,n},
\end{cases}\label{eq:CIsPoint}
\end{equation}
These CSs' properties are summarized in Theorem 4.
\begin{thm}
\label{thm:coveragePW} Suppose that Assumptions \ref{assu:Non--empty}\textendash \ref{assu:Moments}
hold (i.e. $P \in \mathcal{P}$),    $0<\alpha<1/2$,      $\kappa_{n}\to0$ and $\mu_{n}\to0$  are
  such that $\mu_{n}\sqrt{n}\to\infty$. Moreover, suppose that $\lim_{n\to\infty}\funLStd[\mu_{n}][\measTrue][2]>0$ and $\lim_{n\to\infty}\funUStd[\mu_{n}][\measTrue][2]>0$. Then
\begin{eqnarray*}
\lim_{n\to\infty}\measTrue\left(\setMarginal\subset\mbox{CB}_{\alpha,n}\right)=\liminf_{n\to\infty}\inf_{ {\theta}\in\setID}\measTrue\left(\theta_{1}\in\mbox{CB}_{\alpha,n}\right) & \geq & 1-\alpha,\mbox{ }\\
\lim_{n\to\infty}\measTrue\left(\setMarginal\subset\mbox{CI}_{\alpha,n}^{\mathcal{S}}\right)\geq1-\alpha,\text{ }\liminf_{n\to\infty}\inf_{ {\theta}\in\setID}\measTrue\left(\theta_{1}\in\mbox{CI}_{\alpha,n}^{\mathcal{S}}\right) & \geq & 1-\alpha.
\end{eqnarray*}
If in addition   $\mu_{n}/\kappa_{n}\to0$, then,
\begin{eqnarray*}
\lim_{n\to\infty}\measTrue\left(\setMarginal\subset\mbox{CB}_{\alpha,n}\right)=\liminf_{n\to\infty}\inf_{ {\theta}\in\setID}\measTrue\left(\theta_{1}\in\mbox{CB}_{\alpha,n}\right) & = & 1-\alpha.
\end{eqnarray*}
If further, the model has no equality constraints---that is, if $p=0$---then 
\begin{eqnarray}
\liminf_{n\to\infty}\inf_{ {\theta}\in\setID}\measTrue\left(\theta_{1}\in\mbox{CI}^{\theta_1}_{\alpha,n}\right) & = & 1-\alpha.
\end{eqnarray}
\end{thm}
\begin{proof}
See Appendix \ref{sec:Proof-of-Theorem3}.
\end{proof}

The confidence band $\mbox{CB}_{\alpha,n}$ and the confidence interval $\mbox{CI}_{\alpha,n}^{\theta_1}$ are asymptotically nonconservative for a fixed DGP that satisfies the assumptions of the theorem; that is, they have coverage of exactly $1-\alpha$. 
If $p>0$,  $\theta_{1}$ can be point-identified if there are equality constraints in the model that are orthogonal to $\eOneVec$. 
In this case, I recommend the Bonferroni-type CS  $\mbox{CI}_{\alpha,n}^{\mathcal{S}}$, which remains valid under point identification. 
The shorter $\mbox{CI}_{\alpha,n}^{\theta_1} $ (compare \citet{imbens2004confidence}) is valid only if $\theta_{1}$ is not point-identified.
 
It is interesting to compare the uniformly valid confidence bound $\mbox{CB}_{\alpha,n,\Measures}$ with its nonconservative point-wise counterpart $\mbox{CB}_{\alpha,n}$.
In the regular case, the two sets coincide asymptotically and have the same exact coverage of the support function $\minV(\measTrue)$. In the presence of flat face (the nonregular case), $\mbox{CB}_{\alpha,n}$ will be shorter than $\mbox{CB}_{\alpha,n,\Measures}$ with probability 1. 
Moreover, the confidence level of $\mbox{CB}_{\alpha,n}$ will remain asymptotically exact  (point-wise).
Despite this, I recommend using the uniformly valid bound $\mbox{CB}_{\alpha,n,\Measures}$  since it would have better control of the confidence level in small samples (because of its validity under drifting DGP sequences; see Remark~\ref{rem:drifitngDGP}).
 
\begin{rem}\label{rem:drifitngDGP}
 Theorem~\ref{thm:coveragePW} provides an asymptotic coverage probability for a given DGP with measure $\measTrue$.
The size of the sample required to achieve the nominal coverage of $1-\alpha$ with a given precision in this result can depend on $\measTrue$  through $\bar{\mu}(\measTrue)$ as defined in Theorem~\ref{thm:bounds}.
The cutoff $\bar{\mu}(\measTrue)$ can be arbitrarily close to zero.
It is possible to construct an example in which the sequence of measures $\measSec$  meets the assumptions of Theorem~\ref{thm:coveragePW} but   the asymptotic bias is larger than $1/\sqrt{n}$:
\begin{equation}
    \sqrt{n}\mn({\ltwo{\funLArg[\mn][\measSec]}-\ltwo{\funLArg[\kappa_n][\measSec]}})\to+\infty.
\end{equation} 
In other words, there are examples of DGP with a measure $\measTrue$ and some $\epsilon>0$ such that for any $n$ it is possible to find a measure $Q$ in a neighborhood of $\measTrue$ with \begin{equation}Q\left(\setMarginal\subset\mbox{CB}_{\alpha,n} \right)<1-\alpha-\epsilon.
\end{equation} 
In practical terms, this means that the large-sample theory with a fixed $\measTrue$ may provide a poor approximation for the true coverage probability of the point-wise-valid confidence set  $\mbox{CB}_{\alpha,n}$.  
(A similar concern applies to other existing point-wise-valid inference procedures.)
\end{rem}
 \subsection{Proof of Theorem \ref{thm:coveragePW}\label{sec:Proof-of-Theorem3}}

\begin{proof}[Proof of Theorem \ref{thm:coveragePW}] \noun{Step 1.} First, suppose that  $\mu_{n}/\kappa_{n}\to0$.
Consider
\begin{equation}
\argminZ_{n}  \bydef \sqrt{n}\frac{\minVin(\mu_n,\kappa_n,\measEmp)-\minVin(\mu_n,\kappa_n,\measTrue)+\minVin(\mu_n,\kappa_n,\measTrue)-\funLValO}{\funLStdHat} \label{eq:asymptoticDistribution}
\end{equation}
Since $\mn/\kappa_n\to 0$ and $\mn\to 0$, for all $n$ large enough, such that $\mn\leq \kappa_{n}\leq\bar{\mu}\left(\measTrue\right)$, by  Theorem~\ref{thm:bounds} we get
\begin{equation}
 \minVin(\mu_n,\kappa_n,\measTrue)=\funLValO. \label{eq:constantTheta1}
\end{equation}

By Theorem~\reF{thm:consistency}  we get
\begin{equation}
\mn\sqrt{n}\left(\ltwo{\funLArgHat[\kappa_{n}]}^{2}-\ltwo{\funLArg[\kappa_{n}]}^{2}\right) = \frac{\mn}{\kappa_{n}}  O_p(1)=o_p(1). \label{eq:consistency theta2 kappa}
\end{equation}

By Lemma~\reF{lem:continuityInMu}, $\funLArg$ and $\funLLagrange$ are continuous for $\mu>0$. The matrix function $\funCov[\theta]$ is continuous in $\theta$ and thus  $\argminS\mP$ is continuous in $\mu$ for $\mu>0$. 
So the limit  $\lim_{n\to\infty}\funLStd[\mu_{n}][\measTrue][2]$  exists and belongs to the set $\funLStd[0][\measTrue][2]$ which by  assumptions  implies  $\lim_{n\to\infty}\funLStd[\mu_{n}][\measTrue][2]>0$.
Result~\eqref{eq:coupling} in Theorem~\reF{thm:consistency} together with \eqreF{eq:constantTheta1} and \eqreF{eq:consistency theta2 kappa} imply by Slutsky's theorem that $\argminZ_{n}$ converges in distribution to $N\left(0,1\right)$. 

\noun{Step 2.} Consider the one\textendash sided confidence band
$\mbox{CB}_{\alpha,n}$.
\begin{eqnarray*}
 & \lim_{n\to\infty} & \measTrue\left\{ \setMarginal\subset\mbox{CB}_{\alpha,n}\right\} \\
= & \lim_{n\to\infty} & \measTrue\left\{ \funLValO \geq \minVin(\mu_n,\kappa_n,\measEmp)-\funLStdHat[\mn]z_{1-\alpha}n^{-1/2}\right\} \\
= & \lim_{n\to\infty} & \measTrue\left\{ \argminZ_{n} \leq z_{1-\alpha}\right\} \\
= & \Phi\left(z_{1-\alpha}\right) & =1-\alpha.
\end{eqnarray*}
 Proof for $\mbox{CI}_{\alpha,n}^{\mathcal{S}}$ follows immediately from the Bonferroni inequality. 
 Finally, consider the case $p=0$. Then by Lemma \reF{lem:LICQ} and Assumption~\ref{assu:Non--empty}, $\funLVal[0]<-\funUVal[0]$. 
So 
\begin{eqnarray}
\lim_{n\to\infty}\min_{ {\theta}\in\setID}\measTrue\left( {\theta}\in\mbox{CI}_{\alpha,n}^{\theta_1}\right)=\nonumber \\
\min\left\{ \lim_{n\to\infty}\measTrue\left\{ \funLValHat[\mn]-\mn\ltwo{\funLArgHat[\kappa_{n}]} ^{2}-\funLStdHat[\mn]z_{1-\alpha}n^{-1/2}\leq\funLValO\right\} ,1,\right.\dots\nonumber \\
\left.\lim_{n\to\infty}\measTrue\left\{ -\funUValHat[\mn]+\mn\ltwo{\funUArgHat[\kappa_{n}]}^{2}+z_{1-\alpha}\funUStdHat[\mn]n^{-1/2}\geq\funUValO\right\} \right\} =\nonumber \\
\min\left\{ \lim_{n\to\infty}\measTrue\left\{ \setMarginal\subset\mbox{CB}_{\alpha,n}\right\} ,1,\lim_{n\to\infty}\measTrue\left\{ \setMarginal\subset\mbox{CB}_{\alpha,n}^{R}\right\} \right\} =\\
\min\left\{ 1-\alpha,1,1-\alpha\right\} =1-\alpha.
\end{eqnarray}
To understand the second equation, consider the following argument.
Suppose that $ {\theta}\in\setID$ is such that $\funLValO<\theta_1<\funUValO$.
Then such $  {\theta_1}$ will be covered with probability 1 since  $\mbox{CI}_{\alpha,n}^{\theta}$ is the intersection of  $\mbox{CB}_{\alpha,n}$ and $\mbox{CB}_{\alpha,n}^{R}$ which cover correspondingly $\funLValO$ and $\funUValO$ . 

\noun{Step 3.} Finally, suppose that  $\mu_{n}/\kappa_{n} \to0$ does not hold, i.e. sequence  $\kappa_{n}\to0$ can take any positive values. 
Consider any auxiliary sequence $\kappa^*_n\to0$ such that  $\mu_{n}/\kappa^*_{n} \to0$ and $\kappa_{n} \leq \kappa^*_n$.
By Step 2,
\begin{eqnarray*}
 \lim_{n\to\infty}   \measTrue\left\{ \funLValO \geq \minVin(\mu_n,\kappa^*_n,\measEmp)-\funLStdHat[\mn]z_{1-\alpha}n^{-1/2}\right\}   =1-\alpha.
\end{eqnarray*}
But $\minVin(\mu_n,\kappa ,\measEmp)$ is a non-increasing function in $\kappa $. Indeed, by \eqref{eq:thetadotmu0} in proof of Lemma \ref{lem:constantTheta}, the partial derivative $\ltwo{\argminT(\kappa,\measEmp)}^2$ w.r.t. $\kappa$ at any point $\kappa_0>0$ is $2\argminT^\prime(\kappa_0,\measEmp) \dot{\argminT}(\kappa_0,\measEmp) \leq 0 $. So  $\ltwo{\argminT(\kappa,\measEmp)}$ is decreasing function of $\kappa>0$.

The derivative at $\kappa=0$ does not exist, so this case has to be considered separately. Take any $\theta^* \in \argminT(0,\measEmp)$. Since $\theta^*$ is a feasible point of Program~\eqref{prog:regProg} with $\measTrue$ replaced by $\measEmp$, 
\begin{equation}
\argminT_1(\kappa,\measEmp) + \kappa\ltwo{ \funLArg[\kappa][\measEmp]}^{2}\leq    \theta^*_1 + \kappa \ltwo{\theta^*}^{2}.\label{eq:growthinkappa}
\end{equation}
By Theorem \ref{thm:bounds}, for $\kappa\leq\bar\mu(\measEmp)$ we have $\argminT_1(\kappa,\measEmp) = \theta^*_1$, so for such small values \eqref{eq:growthinkappa} becomes
\begin{equation}
\ltwo{ \funLArg[\kappa][\measEmp]}^{2}\leq  \ltwo{\theta^*}^{2}.\notag
\end{equation}
This proves that  $\ltwo{\argminT(\kappa,\measEmp)}$ is   a non-increasing function of $\kappa$ for all values $\kappa\geq0$ including $0$.
So a.s.  $\minVin(\mu_n,\kappa_n,\measEmp)\leq\minVin(\mu_n,\kappa^*_n,\measEmp)$ and therefore
\begin{eqnarray*}
 & \lim_{n\to\infty} \measTrue\left\{ \funLValO \geq \minVin(\mu_n,\kappa_n,\measEmp)-\funLStdHat[\mn]z_{1-\alpha}n^{-1/2}\right\}  &\geq \\
& \lim_{n\to\infty}  \measTrue\left\{ \funLValO \geq \minVin(\mu_n,\kappa^*_n,\measEmp)-\funLStdHat[\mn]z_{1-\alpha}n^{-1/2}\right\}&=1-\alpha.
\end{eqnarray*}
It means that if one uses a sequence $\kappa_n$ that converges to 0 faster than $\mu_n$, it makes all the confidence sets more conservative than the ones corresponding to $\kappa^*_n$.  
\end{proof}

\subsection{A special case of over-identified moment conditions: Uniformly valid inference on intersection bounds} \label{subsec:overidentification}

The sufficient conditions for LICQ (Assumptions~\ref{assu:RankCondition}--\ref{assu:NoOveridentification}) considerably simplify the asymptotic analysis, but can be violated in some interesting cases.
In particular, the active inequality constraints can be linearly dependent.
In those cases, the regularization  approach can still be applied to the dual formulation of the linear program as in equation  \eqref{eq:DualProgram0}. 
A complete and comprehensive analysis of dual regularization goes beyond the scope of this paper.
In this section, I  restrict my attention to a special case of the following linear program:
\begin{align}
\funLValO=\min_{ \theta_1\in\R} \theta_1 \label{eq:intersectionBoundsmain},\\
\text{s.t. }    \theta_1 \geq &  \expect W_j,\text{ for }j=1,\dots,k.\nonumber 
\end{align}
This   function appears in many econometric applications and has been studied in the literature (see \eqref{eq:subsystemIntersectionBounds} in Example~\ref{exa:runningExample}; also see examples in  \cite{hall2010bootstrap}, \cite{chernozhukov2013intersection}). 
The value of this program can be written explicitly as $\max_{j=1,\dots,k}     \expGeneric W_j $ (also known as an \emph{intersection-bound} parameter). 
Such a maximum over a finite set can in turn be equivalently represented as a linear program,
\begin{align}
\funLValO=&    - \min_{\theta_1\in\R,{\lambda}\in\R ^{k}}   \left\{  -\theta_1\right\} \label{eq:intersectionBoundsmainConvex},\\
\text{s.t. } 
&   \theta_1 -   \sum_{j=1}^k  \lambda_j \expGeneric W_j =0,\text{ }  \sum_{j=1}^k  \lambda_j  =   1,\text{ } \lambda_j  \geq 0.\nonumber 
\end{align}
(Compare with the dual-linear-program representation \eqref{eq:DualProgramIntersectionBounds}.)
The problem of overidentifying inequality constraints (multiple dual solutions) in formulation \eqref{eq:intersectionBoundsmain} becomes a ``flat face'' problem (multiple primal solutions) in  formulation \eqref{eq:intersectionBoundsmainConvex}.
Constraints of the program in  \eqref{eq:intersectionBoundsmainConvex} are almost surely linearly independent, so one can directly verify that Assumptions \ref{assu:Non--empty} and \ref{assu:RankCondition}--\ref{assu:NoOveridentification} are also satisfied.
As a result, Theorems~\ref{thm:bounds} and \ref{thm:consistency} can be applied to the following dual regularized formulation:
\begin{align}
\funLdualVal{\mu_{n}} = &    - \min_{ \theta_1\in\R,{\lambda}\in\R ^{k} }   \left\{  -\theta_1 +\mu_n (\theta_1^2+\ltwo{\lambda}^2) \right\} ,\label{eq:dualregularized}\\
\text{s.t. } 
&   \theta_1 -   \sum_{j=1}^k  \lambda_j \expGeneric W_j =0,\quad \sum_{j=1}^k  \lambda_j  =   1,\quad  \lambda_j  \geq 0.\nonumber
\end{align}
This is completely analogous to \eqref{prog:minSSAregularized} (with the exception of the negative sign in front of the minimum operator).
Estimator $\funLdualVal{\mu_{n}}$ has the following bias-corrected version:
\begin{equation}
    \underline{v}^{dual,in} (\mu_n,\kappa_n,\measEmp) = \funLdualVal{\mu_{n}} + \mu_n (\underline{\theta}_1^2(\kappa_n,\measEmp)+\ltwo{\underline\lambda(\kappa_n,\measEmp)}^2). 
\end{equation}
Here, $\underline{\theta}_1(\kappa_n,\measEmp)$ and $\underline\lambda(\kappa_n,\measEmp)$ are solutions to \eqref{eq:dualregularized} with a tuning parameter $\kappa_n$ instead of $\mu_n$.

By implication of Theorem~\ref{thm:bounds}, the estimator $ \underline{v}^{dual,in} (\mu_n,\kappa_n,\measEmp)$ is either asymptotically unbiased for $\funLValO$ from \eqref{eq:intersectionBoundsmainConvex} (in the regular case, in which only one inequality is binding) or biased downward (in the nonregular case, in which multiple inequalities are binding).
As a result, $\underline{v}^{dual,in} (\mu_n,\kappa_n,\measEmp)$ can be used in a delta-method CS for $\funLValO$,
\begin{equation}
    \mbox{CB}^{dual}_{\alpha,n,\Measures}   \bydef\left[\underline{v}^{dual,in} (\mu_n,\kappa_n,\measEmp)-z_{1-\alpha}n^{-1/2}\funLStdHat[\mu_{n}],\infty\right),
\end{equation}
where $\funLStdHat[\mu_{n}]$ corresponds to the asymptotic variance estimator for $\funLdualVal{\mu_{n}}$.
Following the steps of the proof of Theorem~\ref{thm:TheoremUniform}, one can establish the correct uniform asymptotic coverage of $\funLValO$ with $\mbox{CB}^{dual}_{\alpha,n,\Measures}  $.
Moreover, in this case the coverage will be exactly $1-\alpha$ for any fixed DGP even in the nonregular case (analogous to the results in Theorem~\ref{thm:coveragePW}). 

To conclude, the regularization approach proposed  can be applied to systems of moment-inequality conditions that violate Assumptions~\ref{assu:RankCondition}  and \ref{assu:NoOveridentification}  after  an appropriate dual reformulation of the program.
A comprehensive study of this dual regularization is a topic for another paper. 
The intersection-bounds problem is a specific empirically relevant example in which the method can be applied with minimal modification to the convex reformulation in Equation \eqref{eq:intersectionBoundsmainConvex}.

\section{Discussion of  computational properties}\label{sec:Discussion}\label{sec:Computation}

\subsection{Fast convergence to a minimum}\label{sec:Convergnce rate}
The existing uniform methods of AS, BCS, and KMS are based on standardized moment conditions that are nonconvex even if the original inequalities are affine in $\theta$. 
Example 2 in the previous section illustrates this feature.  
The estimator $\funLArgHat[\mu_{n}]$ is a solution to a strictly convex quadratic program for any affine-moment-inequality model. 
For convex programs, the  Karush\textendash Kuhn\textendash{} Tucker (KKT)  conditions  provide necessary and sufficient conditions for the global optimum (see Lemma \ref{lem:UniqueMin} in Appendix).
Moreover, convex quadratic programs can be solved using  interior-point algorithms with a polynomial rate of convergence. (See, for example, \cite{ye1989extension}.)
This strict convexity gives  a dramatically faster rate of obtaining the optimum than the ones used in BCS and KMS.
These methods are based on nonconvex-constraint optimization problems, which are NP-hard.
Section~\ref{sec:Monte-Carlo} compares computational time in specific examples.

\subsection{Uniqueness of a global optimum} \label{sec:globalOptimum}
The KKT system for strictly convex optimization problems has a unique solution.
The number of KKT points of the optimization problems in the KMS, BCS, and AS procedures in affine moment inequality models can be large and typically grows exponentially with the dimension $d$ and number of inequalities $k$.
The following example illustrates this point. 
\begin{example}\label{ex:multipleMin} 
\label{exa:GrowingD}Consider a set of moment inequalities with coefficients
that have expectation
\[
\expect W=\left(\begin{array}{c|c}
-\I[d] & -\iota\\
\I[d] & -\iota
\end{array}\right).
\]
Suppose that components of $ W$ are independent and have
the same variance $s^{2}$. $\setID$ is a box $\left[-1,1\right]^{d}$.
The standardized moment conditions take the form
\begin{equation}
\frac{\pm\theta_{j}+1}{s\sqrt{1+\ltwo{\typVector{\theta}}^{2}}}\leq0,\text{ }j=1,\dots,d.\label{eq:standardizedMoments}
\end{equation}
The KMS procedure adds slack $c\left(\typVector{\theta}\right)$ to
the right-hand side of every standardized moment inequality. Consider, for example, $j=1$, where
\begin{equation}
\theta_{1}\geq1-c\left(\typVector{\theta}\right)s\sqrt{1+\ltwo{\typVector{\theta}}^{2}}.\label{eq:non-convexINeq}
\end{equation}

\begin{figure}[H]
\begin{centering}
\begin{tabular}{cc}
\includegraphics[scale=0.25]{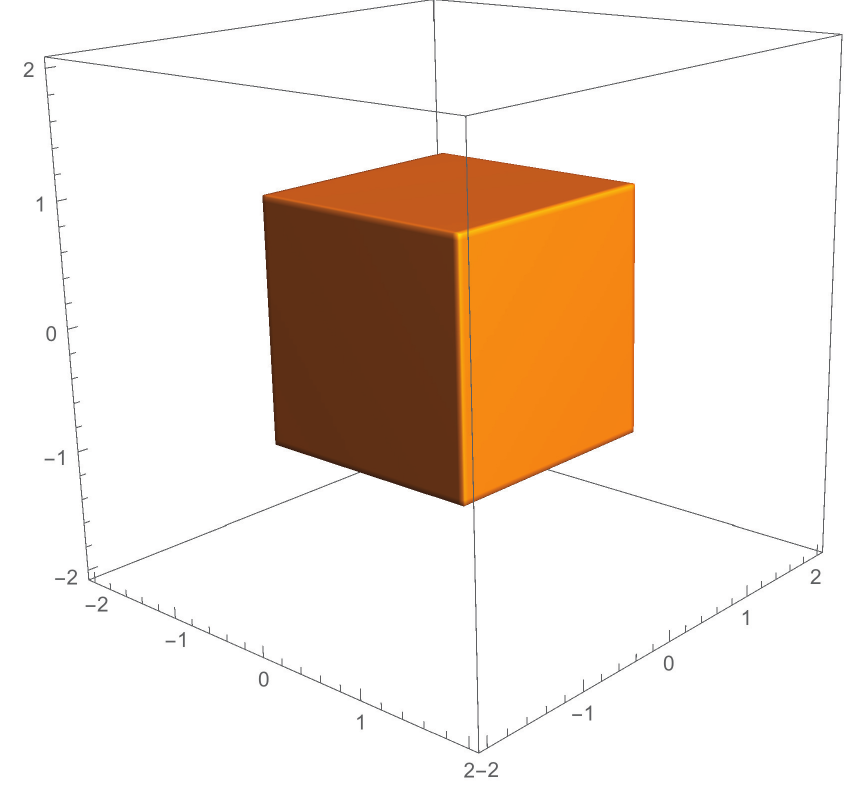} & \includegraphics[scale=0.25]{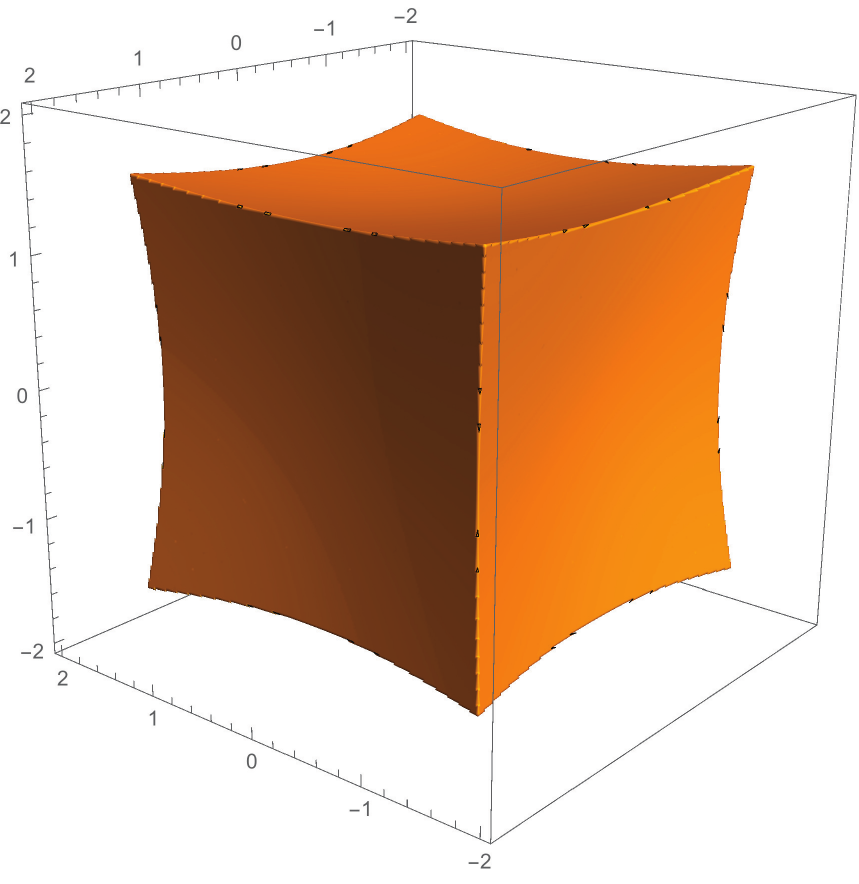}\tabularnewline
\end{tabular}
\par\end{centering}
\caption{\label{fig:Identified-set-and}The identified set and the corresponding
optimization domain of a KMS-type procedure for $d=3$ in Example~\ref{ex:multipleMin}.}
\end{figure}

The slack function $c\left(\typVector{\theta}\right)$ is computed using a resampling on a grid of points. 
Assume, for simplicity, that
$c\left(\typVector{\theta}\right)$ is a constant---for example, one provided using the Bonferroni approach. Figure \ref{fig:Identified-set-and}
shows the identified set and the corresponding expansion with $c\left(\typVector{\theta}\right)=const$.
The optimization domain of the E-A-M algoritm in KMS is similar to the nonconvex set on the right side of Figure \ref{fig:Identified-set-and}. 
Every vertex of the $\left[-1,1\right]^{d}$ with $\theta_{1}=-1$ corresponds to an isolated local minimum of the optimization procedure in KMS. Correspondingly,
the number of local minima grows exponentially with the dimension $d$.
For example, the number of local minima for $d=10$ is $512$.
The growth in the number of local optima is even faster in models with more than two inequalities per coordinate.  $\QEDB$
 
\end{example}

Multiplicity of KKT points both makes the procedures of KMS, AS, and BCS  computationally costly and  does not guarantee  convergence to a global optimum for large $d$. 

\subsection{Implications for validity of multiplier bootstrap}\label{sec:multiplier boostrap}
The multiplier bootstrap  provides a numerical way to implement detla-method inference. 
Main advantage of this procedure when compared to conventional analytical formulas is the fact that one can use simulations to obtain joint variance-covariance matrix of several asymptotically Gaussian estimators (for example, bounds on projections of an identified set for a set of directions) and/or directly make draws from    their joint limiting Gaussian distribution without necessity to derive the corresponding matrix formulas explicitly. 
In comparison to the standard non-parametric bootstrap methods, multiplier bootstrap eliminates the need to recompute estimators based on potentially costly optimization routines.\footnote{For example, FH and KMS solve mathematical programs repeatedly for every bootstrap sample. One of the recent papers in the moment inequality literature that made use of multiplier bootstrap is \cite{chernozhukov2023constrained}. That paper however uses multiplier bootstrap to simulate critical values of GMM-type statistics at a hypothetical parameter values for subsequnt test inversion.}
The multiplier-bootstrap approach would be particularly appealing in subvector inference on more that one component as discussed in Section~\ref{subsec:subvectors}.

The proposed estimators of the regularized support functions for fixed direction have a Bahadur-Kiefer representation with explicit influence functions given in Equation \eqref{eq:Bahadur}. 
One can use this property to justify multiplier bootstrap for inference on the support  of the identified set along the lines of proving validity of the delta-method.
Proof of Theorem \ref{thm:TheoremUniform} can be adapted by replacing estimators of standard deviation $\funLStdHat[\mu_{n}]$ with the multiplier bootstrap estimator of asymptotic standard deviation of  $ \sqrt{n} ( \funLValHat[\mu_{n}]- \minV\mnP )$.
Indeed, Theorem \ref{thm:consistency} shows that the unknown parameters in the influence function are (uniformly over $\mathcal{P}$) consistently estimated by their sample analog.
Multiplier bootstrap procedure could also be generalized to obtain non-standard critical values of $\sup t$ statistics as discussed earlier in Subsection~\ref{subsec:subvectors}.
I leave full analysis of such an extension for future work.

\subsection{Choice of solvers}
The point-wise CIs in (\ref{eq:CIsPoint}) can be computed using any
Newton-type optimization software that provides accurate Lagrange
multipliers. I use the \emph{fmincon} function of \noun{MATLAB} software. I recommend using the active-set or SQP option since the interior-point  option does not provide accurate Lagrange multipliers. 
The estimator $\ltwo{\theta^*\left(\measEmp\right)}^{2}$, which enters the uniformly valid CS described in Theorem \ref{thm:TheoremUniform}, is based on $2d$ linear programs. Linear programs typically scale very well.

 A  benchmarking of results for the state-of-the-art commercial LP solvers can be found, for example, at  http://plato.asu.edu/bench.html. The commercial solvers can tackle LP with tens of thousands of constraints and variables in a matter of minutes. Matlab's linprog solver tends to underperform in the time comparison.

\section{Additional Monte Carlo Results}

\begin{figure}[H]

\begin{centering}
\caption{\label{fig:pwFreq_opt2}Coverage frequency for   $\mbox{CB}_{\alpha,n}$  as function of $\omega$ in the $2$-dimensional  design with $\kappa_n=\mu_n$.}

\includegraphics[scale=1]{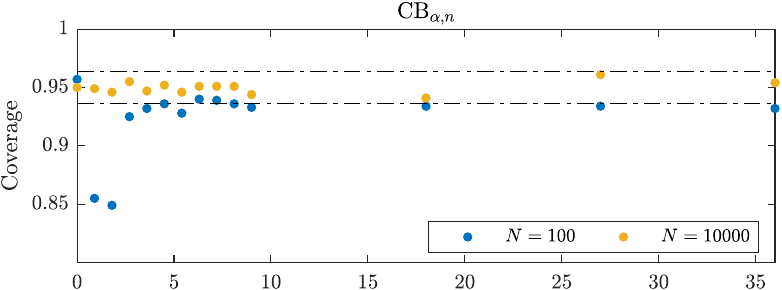} 
 
\par\end{centering}
   Note: The dashed lines correspond  to the asymptotic 95\% confidence interval (point-wise) for the parameter $p=0.95$ of Bernoulli random variable based on a random sample of 1000 observations.
   Some values of the estimated frequency can be slightly outside of the  confidence interval  as a result of multiple hypothesis testing.
   Values of $\omega$ close to zero result in negligible under-coverage. As sample size grows, the problematic area shrinks.
 
\end{figure}

\begin{figure}[H]

\begin{centering}
\caption{\label{fig:lengthCompare_opt2}MC average excess lengths  of  $\mbox{CB}_{\alpha,n}$ and   $\mbox{CB}_{\alpha,n,\Measures}$  as function of $\omega$ in the $2$-dimensional  design with $\kappa_n=\mu_n$.}

\includegraphics[scale=1]{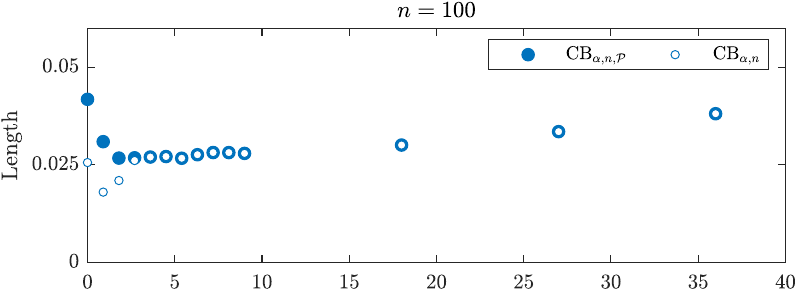} 
 
\par\end{centering}
 
\end{figure}

\newpage
\begin{figure}[H]

\begin{centering}
\caption{\label{fig:pwFreq_opt3}Coverage frequency for   $\mbox{CB}_{\alpha,n}$  as function of $\omega$ in the $2$-dimensional  design with $\kappa_n=0$.}

\includegraphics[scale=1]{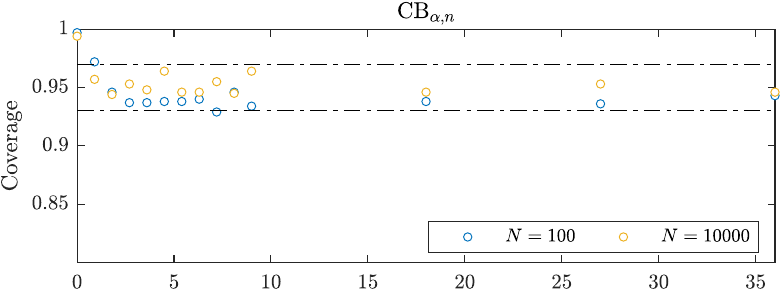} 
 
\par\end{centering}
   Note: The dashed lines correspond to the asymptotic 95\% confidence interval (point-wise) for the parameter $p=0.95$ of the Bernoulli random variable based on a random sample of 1000 observations.
   Some values of the estimated frequency can be slightly outside the confidence interval as a result of multiple hypothesis testing.
   Values of $\omega$ close to zero result in negligible conservative coverage.  
 
\end{figure}

\begin{figure}[H]

\begin{centering}
\caption{\label{fig:lengthCompare_opt3}MC average excess lengths  of  $\mbox{CB}_{\alpha,n}$ and   $\mbox{CB}_{\alpha,n,\Measures}$  as function of $\omega$ in the $2$-dimensional  design with $\kappa_n=0$.}

\includegraphics[scale=1]{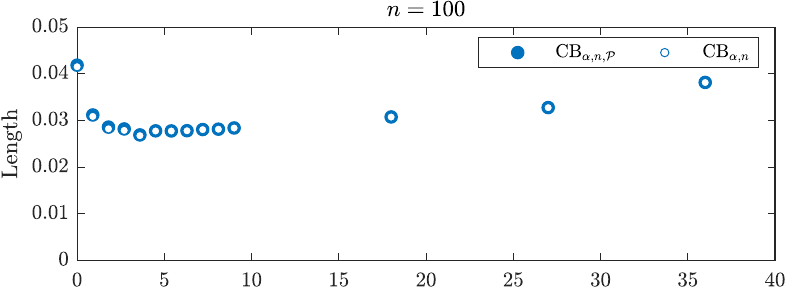} 
 
\par\end{centering}
 
\end{figure}

\newpage

\end{document}